\documentclass{amsart}
\usepackage[utf8]{inputenc}
\usepackage{amsmath}
\usepackage{array, amsfonts}
\usepackage{amsthm}
\usepackage{amssymb}
\usepackage{enumerate}
\usepackage{lineno}
\usepackage[bookmarks=true]{hyperref}
\usepackage{amsaddr}
\usepackage{enumitem}
\usepackage{mathrsfs}
\usepackage{stackengine}
\usepackage{bookmark}
\usepackage{amsrefs}
\usepackage{color,soul}
\usepackage[draft]{todonotes}
\usepackage{fancyhdr}

\definecolor{darkgreen}{rgb}{0,0.5,0}

\DeclareMathOperator*{\myDelta}{\Delta}

\newcommand{\comment}[1]

\newcommand\myfunc[5]{%
	\begingroup
	\setlength\arraycolsep{0pt}
	#1\colon\begin{array}[t]{c >{{}}c<{{}} c}
		#2 & \to & #3 \\ #4 & \mapsto & #5 
	\end{array}%
	\endgroup}

\title{Emergent Open-Endedness\\ from Contagion of the Fittest }

\author{Felipe S. Abrah\~{a}o, Klaus Wehmuth, Artur Ziviani}
\email{\{fsa,klaus,ziviani\}@lncc.br}

\address[A1,A2,A3]{National Laboratory for Scientific Computing (LNCC)
	\\ 25651-075 – Petropolis, RJ – Brazil}

\thanks{ In~\cite{Abrahao2018b}, some preliminary results of this article are presented as an extended abstract.  }

\thanks{Authors acknowledge the partial support from CNPq through their individual grants: F. S. Abrah\~{a}o (313.043/2016-7), K. Wehmuth (312599/2016-1), and A. Ziviani~(308.729/2015-3). Authors also acknowledge the INCT in Data Science -- INCT-CiD (CNPq 465.560/2014-8). Authors also acknowledge the partial support from CAPES, FAPESP, and FAPERJ}

\fancypagestyle{plain}{
	\fancyhead{}
	\fancyhead[C]{ \textit{Research Report} \\ National Laboratory for Scientific Computing (LNCC)  \\ \textbf{Extended version of the paper:} }
}
\begin{document}

	\maketitle\thispagestyle{plain}


	\newtheorem{thm}{Theorem}[subsection]
	\newtheorem{amslemma}{Lemma}[subsection]
	\newtheorem{amscorollary}{Corollary}[thm]
	\newtheorem{corollaryundersubsection}{Corollary}[subsection]

	\newtheorem{notation}{Notation}[subsection]
	
	\theoremstyle{definition}
	\newtheorem{amsproposition}{Proposition}[subsection]
	\newtheorem{amsdefinition}{Definition}[subsection]
	\newtheorem{subdefinition}{Definition}[amsdefinition]
	\newtheorem{subsubdefinition}{Definition}[subdefinition]
	\newtheorem{autodefinition}{Definition}
	
	\theoremstyle{remark}
	\newtheorem{subnotation}{Notation}[amsdefinition]
	
	\theoremstyle{remark}
	\newtheorem{remark}{Remark}[amsdefinition]
	\newtheorem{remarknote}{Note}[amsdefinition]
	\newtheorem{noteunderlemma}{Note}[amslemma]
	\newtheorem{noteunderthm}{Note}[thm]
	\newtheorem{subremarknote}{Note}[subdefinition]
	\newtheorem{subsubremarknote}{Note}[subsubdefinition]
	\newtheorem{subremarknote2}{Note}[remarknote]
	\newtheorem{note}{\textbf{Note}}[subsection]

	\newtheorem{Bthm}{Theorem}[section]
	\newtheorem{Blemma}{Lemma}[section]
	\newtheorem{Bcorollary}{Corollary}[Bthm]
	\newtheorem{Bcorollaryundersection}{Corollary}[section]
	\newtheorem{Bcorollaryundersubsection}{Corollary}[subsection]
	
	\theoremstyle{definition}
	\newtheorem{Bproposition}{Proposition}[section]
	\newtheorem{Bdefinition}{Definition}[section]
	\newtheorem{Bsubdefinition}{Definition}[Bdefinition]
	\newtheorem{Bsubsubdefinition}{Definition}[Bsubdefinition]
	\newtheorem{Bautodefinition}{Definition}
	\newtheorem{Bnotation}{Notation}[section]
	
	\theoremstyle{remark}
	\newtheorem{Bsubnotation}{Notation}[Bdefinition]
	
	\theoremstyle{remark}
	\newtheorem{Bremark}{Remark}[Bdefinition]
	\newtheorem{Bremarknote}{Note}[Bdefinition]
	\newtheorem{BnoteunderBlemma}{Note}[Blemma]
	\newtheorem{BnoteunderBthm}{Note}[Bthm]
	\newtheorem{Bsubremarknote}{Note}[Bsubdefinition]
	\newtheorem{Bsubsubremarknote}{Note}[Bsubsubdefinition]
	\newtheorem{Bsubremarknote2}{Note}[Bremarknote]
	\newtheorem{Bnote}{\textbf{Note}}[section]
	\newtheorem{BnoteunderBnotation}{Note}[Bnotation]

	\begin{abstract}\label{abstract}
		In this paper, we study emergent irreducible information in populations of randomly generated computable systems that are networked and follow a ``Susceptible-Infected-Susceptible'' contagion model of imitation of the fittest neighbor.  We show that there is a lower bound for the stationary prevalence (or average density of ``infected'' nodes) that triggers an unlimited increase of the expected local emergent algorithmic complexity (or information) of a node as the population size grows. We call this phenomenon expected (local) \emph{emergent open-endedness}. In addition, we show that static networks with a power-law degree distribution following the Barabási-Albert model satisfy this lower bound and, thus, display expected (local) emergent open-endedness.
	\end{abstract}
	
	\keywords{ emergence; information; complexity; complex networks; Turing machines; complex systems; contagion; spreading; Susceptible-Infected-Susceptible; open-endedness }
	
	\subjclass[2010]{ 68Q30; 68Q05; 05C82; 94A15 }

\section{Introduction}\label{sectionIntro}
	
The general scope of this work encompasses complex systems, complex networks, information theory, and computability theory. In particular, we study the general problem of emergence of complexity or information when complex systems are networked compared with when they are isolated. This issue has a pervasive importance in the literature about complex systems with applications on investigating systemic properties of biological, economical, or social systems. As discussed in~\cite{Abrahao2017}, it may be a subject connected to questions ranging from the problem of symbiosis~\cite{MargulisLynn1981}, cooperation~\cite{Axelrod2006}, and integration \cite{Oizumi2014} to biological \cite{Kim2015}, economic \cite{Schweitzer2009}, and social \cite{Miller2007} networks. 

From an information-theoretic perspective, emergence in complex systems is also studied in~\cite{Prokopenko2014,Prokopenko2009,Hernandez-Orozco2018}. In addition, a (statistical and/or algorithmic) information-theoretic study on complex networks or graphs is also found in~\cite{Zenil2014,Mowshowitz2012,Buhrman1999,Sole2004}. Thus, the present work and the investigation on networked computable systems using algorithmic networks~\cite{Abrahao2016b,Abrahao2017,Abrahao2018b} have shown a way on how to bring these topics into an abstract mathematical theory. Therefore, enabling one to formally define sound and crucial properties and to prove fruitful theorems. Besides complex systems and complex networks, our work is also related to, and inspired by, fundamental concepts in distributed computing, multi-agent systems, and evolutionary game theory \cite{Abrahao2017}. For example, such research may point to applicable future strategies for optimizing communication protocols in artificial networks of randomly generated systems which seek for a better solution (or approximation) to an undecidable or an intractable problem \cite{Zenil2016,Prokopenko2017,Hernandez-Orozco2017}.  
	
Following the issue raised in~\cite{Abrahao2017}, we present in this paper a study on the emergence of irreducible information in networked computable systems that follow an information-sharing (or communication) protocol based on contagion or infection models, as described in \cite{Pastor-Satorras2001a,Pastor-Satorras2002,Pastor-Satorras2001}. As supported by these references, such models of spreading using the approach from complex networks have been shown to be relevant in order to study epidemic and disease spreading, computer virus infections or the spreading of polluting agents. Consequently, it has helped, for instance, on immunization strategies, epidemiology, or pollution control \cite{Pastor-Satorras2001a,Pastor-Satorras2002,Pastor-Satorras2001}. 
	
However, instead of focusing on the pathological properties of such complex networks' contagion dynamics, we show that this dynamics may instead trigger an unlimited potential of optimization through diffusion. That is, diffusing the best solution (or the largest integer when one uses the Busy Beaver game~\cite{Abrahao2017} as a toy model) through the network may trigger an unlimited increase of expected emergent algorithmic information of the nodes as the randomly generated population of computable systems (i.e., nodes) grows. Thus, we aim to mathematically investigate under which conditions this phenomenon is expected to happen. For this purpose, we use the theoretical framework for networked computable systems developed in \cite{Abrahao2017} and a Susceptible-Infected-Susceptible (SIS)~\cite{Bailey1975} epidemiological model, which was also studied in \cite{Pastor-Satorras2001a,Pastor-Satorras2002,Pastor-Satorras2001}.
	
As a toy model, such theoretical approach to studying emergence of complexity or information in networked computable systems may help understand and establish foundational properties on why an information dynamics within a system displaying synergistic or emergent behavior might be advantageous from a computational, evolutionary, or game-theoretical point of view \cite{Abrahao2017}. Additionally, as it is our goal to suggest in the present work, these phenomena may be also related to infection dynamics \cite{Pastor-Satorras2001a,Pastor-Satorras2002,Pastor-Satorras2001}---either from computer viruses or diseases. Nevertheless, taken in an opposed but analogous perspective: Contagion of the fittest (or the best solution for a problem) element in a population instead of contagion of pathological or undesirable elements.


In order to tackle this general problem, we narrow our scope and we define a mathematical representation for randomly generated computable systems (i.e., systems that can be fully simulated in a Turing machine) that are networked in a time-varying topology (i.e., a dynamic network). Thus, in our model nodes are randomly generated Turing machines that can send and receive information (i.e., partial outputs) as each node runs its computations until returning a final output. We have defined this networked population of randomly generated Turing machines and a more general mathematical model for networked computable systems which we have called as \emph{algorithmic networks} in \cite{Abrahao2016b,Abrahao2017}. 

The population of our present model plays the \emph{Busy Beaver Imitation Game} (BBIG), in which each node always imitates the fittest neighbor only. Nevertheless, differently from the model in \cite{Abrahao2017}, we present a variation on the information-sharing (or communication) protocol. The major difference in respect to this previous work comes from allowing nodes to become ``cured'' (with rate $ \delta $). Additionally, now nodes also get ``infected'' with rate $ \nu $---which may have a value different from $1$. In \cite{Abrahao2017}, one has that $ \nu = 1 $ always holds. Summarizing, although still playing a BBIG, susceptible nodes follow a rule of imitating the neighbor that had output the largest integer (which corresponds to the fittest individual outcome in the population). However, they follow this rule with probability $ \nu $, and ``infected'' nodes come back---become ``cured''---to the initial stage with probability $ \delta $. Thus, the effective spreading rate $ \lambda = \nu / \delta $ defined in \cite{Pastor-Satorras2001a,Pastor-Satorras2002,Pastor-Satorras2001} assumes a direct interpretation of the rate in which the Imitation-of-the-Fittest Protocol \cite{Abrahao2017} was applied on a node---and this is the reason why we are using the words ``infection'' and ``cure'' between quotation marks. Therefore, the diffusion or ``infection'' scheme of the best output returned by a randomly generated node is ruled by the Susceptible-Infected-Susceptible epidemic model (SIS) in which susceptible nodes have a constant probability $ \nu $ of being ``infected'' by a previously ``infected'' neighbor and ``infected'' nodes have a constant probability $ \delta $ of becoming ``cured''. We also assume, as in \cite{Pastor-Satorras2001a,Pastor-Satorras2002,Pastor-Satorras2001}, that the prevalence of ``infected'' nodes (i.e., the average density of ``infected'' nodes) becomes stationary after sufficient time\footnote{ In particular, it holds if this amount of time is upper bounded by a computable function.}.


Our proofs follow mainly from information theory, computability theory, and graph theory applied on a variation on the information-sharing protocol of the model in \cite{Abrahao2017}.  In particular, we have proved results for general dynamic networks and for dynamic networks with a small diameter---$ \mathbf{O}( \log(N) ) $ compared to the network size $N$---in~\cite{Abrahao2017}. Further, these results are also directly extended to static networks \cite{Abrahao2017} with the small-diameter property. Therefore, we have shown that there are topological conditions that trigger a phase transition in which eventually the algorithmic network $ \mathfrak{N}_{BB}  $ begins to produce an unlimited amount of bits of average local emergent algorithmic complexity/information. These conditions come from a positive trade-off between the average diffusion density and the number of cycles (i.e., communication rounds).  We have called \emph{expected emergent open-endedness} (EEOE) \cite{Abrahao2017} this systemic property of there being such phase transition in an algorithmic network when the network/population size increases indefinitely. Thus, the diffusion power of a dynamic (or static) network has proved to be paramount with the purpose of optimizing the average fitness/payoff of an algorithmic network that plays the Busy Beaver Imitation Game in a randomly generated population of Turing machines. Furthermore, this diffusion power may come either from the cover time~\cite{Costa2015a} or from a small diameter~\cite{Bollobas2004,Albert1999} compared to the network size. 

Open-endedness is commonly defined in evolutionary computation and evolutionary biology as the inherent potential of a evolutionary process to trigger an endless increase of complexity or irreducible information \cite{Abrahao2017,Adams2017,Hernandez-Orozco2018}. That means that in the long run eventually will appear an organism that is as complex as one may want. It has been formally proved in~\cite{Chaitin2012,Chaitin2014} and experimentally supported by~\cite{Hernandez-Orozco2017} that cumulative darwinian-like evolution is expected to reach $N$ bits of algorithmic complexity/information after---realistic fast---$ \mathbf{ O }( N^2 ( \log(N) )^2 ) $ successive algorithmic mutations on one organism at the time, whether your organisms are computable, sub-computable, or hyper-computable (see also discussion on open-endedness in \cite{Abrahao2017}). However, we have found that open-endedness may also emerge as an akin---but different---phenomenon to evolutionary open-endedness: Instead of achieving an unbounded quantity of algorithmic complexity over time (e.g., after successive mutations), an unbounded quantity of emergent algorithmic complexity is achieved as the population/network size increases indefinitely. And since it is a property that emerges depending on the amount of parts of a system---only when these nodes are interacting somehow (e.g. exchanging information) ---, this additional irreducible information that appears only when the node is networked becomes by definition an emergent systemic property \cite{DOttaviano2004,Prokopenko2014,Prokopenko2009}.

Our proofs for an algorithmic network following the SIS diffusion model as in \cite{Pastor-Satorras2001a,Pastor-Satorras2002,Pastor-Satorras2001} also stems from the main idea of combining an estimation of a lower bound for the average algorithmic complexity/information of a networked node and an estimation of an upper bound for the expected algorithmic complexity/information of an isolated node. Additionally, as in~\cite{Abrahao2017}, the estimation of the latter still comes from the law of large numbers, Gibb's inequality, and algorithmic information theory applied on the randomly generated population.
However, now the estimation of the former comes from the SIS model with a stationary prevalence (i.e., a stationary average density of ``infected'' nodes). It gives directly this lower bound by the fact that the prevalence $ \rho \sim \exp(- 1 / m \lambda) $ in~\cite{Pastor-Satorras2001a,Pastor-Satorras2002,Pastor-Satorras2001} becomes equal to the average diffusion density $ \tau_{\mathbf{E}} $ in~\cite{Abrahao2017}.  

Then, we show that, for big enough values of $ m $ compared to $ \lambda $, if the time for achieving a stationary prevalence of ``infected'' nodes $  \rho \sim \exp(- 1 / m \lambda)  $ is upper bounded by a value given by a computable function, then the expected emergent algorithmic complexity/information of a node (i.e. the expected local emergent algorithmic complexity/information) goes to infinity as the network/population size $ N $ goes to infinity. In other words, the average local irreducible information that emerges when nodes are networked compared with when they are isolated is expected to always increase for large enough populations of randomly generated Turing machines. 

As a direct consequence of~\cite{Pastor-Satorras2001a,Pastor-Satorras2002,Pastor-Satorras2001}, our results also imply that the same emergent phenomenon occurs if the network is static and has a scale-free degree distribution in the form of a power law $ P(k) \sim \frac{2m^2}{ k^3} $. This topology and construction of the networks are defined by a random process connecting new nodes under a probability distribution given by a preferential attachment as in \cite{Barabasi1999}. That is, new nodes are more likely to be have connections to higher degree previous nodes. Thus, it will be a corollary of our main result that such scale-free static algorithmic networks also display expected local emergent open-endedness.

\section{Model}

In this section, we present the model of algorithmic networks on which we prove lemmas and theorems. The main idea that defines these algorithmic networks $ \mathfrak{N'}_{BB} $  is to formalize a Susceptible-Infected-Susceptible contagion scheme applied on the model previously defined in \cite{Abrahao2017}. Thus, in this Section, we provide only some basic ideas that were previously established. And, therefore, we focus on a description of the model. \footnote{ For extended formal definitions and extensive discussions see Appendix~\ref{appendix}. }

First, remember that \emph{algorithmic networks} $ \mathfrak{N} = (G, \mathfrak{P}, b)$ are defined\footnote{See Definition~\ref{BdefAN}.} in~\cite{Abrahao2017} upon a population of theoretical machines $\mathfrak{P}$, a generalized\footnote{See Definition~\ref{BdefGraph}.} graph $G=(\mathscr{A},\mathscr{E})$, and a function $b$ that makes aspects of $G$ to correspond to properties of $\mathfrak{ P }$, so that a node in $\mathrm{V}(G)$ is mapped one-to-one to an element of $\mathfrak{ P }$. 
The communication channels through which nodes can send or receive information from its neighbors are defined precisely by edges
in $ G $. 

Graphs $G$, which are in fact \emph{MultiAscpect Graphs} (MAGs), are generalized representations for different types of graphs \cite{Wehmuth2016b}. 
Since we aim at a wider range of different network configurations, MAGs allow one to mathematically represent abstract aspects that may appear in complex high-order networks. For example, these may be dynamic (or time-varying) networks, multicolored nodes or edges, multilayer networks, among others. Moreover, this representation facilitates network analysis by showing that their aspects can be isomorphically mapped into a classical directed graph~\cite{Wehmuth2016b}. Thus, the MAG abstraction has proved to be crucial in~\cite{Abrahao2017} to establish connections between the characteristics of the network and the properties of the population composed of theoretical machines.

As in~\cite{Abrahao2017}, we narrow our theoretical approach in order to study general and fruitful toy models. Thus, now we define\footnote{ See Definition~\ref{BdefN'_BB}.} a class of algorithmic networks $  \mathfrak{N'}_{BB}( N, f, t, j ) $---which can also be denoted as $ ( G_t, \mathfrak{ P' }_{BB}(N), b_j ) $---in which their populations $   {\mathfrak{P'}_{BB}}(N)  $ and graphs $ G_t \in \mathbb{G}_{SIS}( f, t ) $ have determined properties. The terms in parenthesis determines the fully characterization of the algorithmic network. $N$ is the network/population size, i.e., the number of nodes, and $j$ is the index of the arbitrarily chosen function $b_j$. Terms $f$ and $t$ are intrinsically defined\footnote{ See Definition~\ref{BdefFamilyGSIS}.} by the family of graphs $ \mathbb{G}_{SIS}( f, t ) $, as we will explain below. Each element of the population corresponds one-to-one to a node/vertex in $ G_t $ and each cycle of the population corresponds one-to-one to a time instant in $ G_t $. These mappings are also defined\footnote{ See Definition~\ref{BdefFunctionbinAN}.} by the function $ b_j $. The main idea is to build a modification on the algorithmic networks $  \mathfrak{N}_{BB}( N, f, t, \tau, j ) $ presented in~\cite{Abrahao2017} in order to make the networked nodes to play the SIS contagion scheme instead of a plain diffusion through imitation of the fittest.

The population $  {\mathfrak{P'}_{BB}}(N) $ is composed of randomly generated Turing machines (or randomly generated self-delimiting programs) that are represented in a self-delimiting universal programming language $ \mathbf{L_U} $.\footnote{ See Section~\ref{appendixsubsectionmachinesandlanguage} and Definitions~\ref{BdefL'_BB} and~\ref{BdefP'_BB}.} The population is also synchronous in respect to halting cycles, that is, in the end of a cycle (or communication round, as in distributed computing) every node returns its partial and final outputs at the same time.\footnote{ See Definition~\ref{BdefSynchronous}.} Nodes that do not halt in any cycle always return as final output the lowest fitness/payoff, that is, the integer value $ 0 $.\footnote{ Thus, as presented in~\cite{Abrahao2017}, these nodes are programs that ultimately run on an oracle Turing machine (or Hypercomputer) $ \mathbf{U'} $---this requirement is also analogous to the one presented in \cite{Chaitin2012,Chaitin2014,Chaitin2018}, which deal with a sole program at the time and not a population of them. The difference in the present work is that the oracle Turing machine also needs access to a randomly generated number in order to deal with the probabilities $\nu$ and $\delta$ in the SIS. See Definition~\ref{Bdefsensitivetooracles}. } Here, a straightforward interpretation is that nodes that eventually do not halt in a cycle are ``killed''\footnote{ See also \cite{Abrahao2017,Chaitin2012} for a complete evolutionary formalization of this property. Note that now there is a population of software, while in\cite{Chaitin2012, Hernandez-Orozco2018,Hernandez-Orozco2017} there is only one single organism at the time. }, so that their final output has the ``worst'' fitness/payoff. 

In addition, the networked population $  {\mathfrak{P'}_{BB}}(N)  $ follows an \emph{Imitation-of-the-Fittest Protocol} by a \emph{Susceptible-Infected-Susceptible} scheme (IFPSIS)\footnote{ See Definition~\ref{BdefIFPSIS}.} on the fittest randomly generated node (i.e., the node that partially outputs the largest integer in cycle $1$)\footnote{ As in~\cite{Abrahao2017,Abrahao2016b}, note that we still use the Busy Beaver function as our fitness function. Therefore, the largest integer directly represents the fittest final output of a node.}. Thus, every node still obeys the Imitation-of-the-Fittest Protocol (IFP) as in~\cite{Abrahao2017}, in which after the first cycle (i.e., after the first round of partial outputs) every node only imitates the neighbor that has partially output the largest integer, repeating this value as its own partial output in the next cycle. However, the difference now is that, if a node has not been ``infected'' by the fittest randomly generated node and one of its neighbors sends the largest integer, then the node obeys the IFP with probability  $ \nu $. Otherwise, the node just continues to be susceptible with probability $ 1- \nu $. Another difference is that, if a node got ``infected'' by the largest integer, then it may be ``cured'', returning to its partial output from cycle $1$, with probability $ \delta $. Otherwise, it remains ``infected'' with probability $ 1 - \delta $.

Graphs $ G_t =( \mathrm{V}, \mathscr{E}, \mathrm{T} ) $ are Time-Varying Graphs (TVGs) as defined\footnote{ See also Definition~\ref{BdefTVG}.} in~\cite{Costa2015a,Abrahao2017,Wehmuth2016b}. These are a special case of MAGs that have only one additional aspect relative to variation over time in respect to the set of nodes/vertices. Therefore, $\mathrm{V}( G_t )$ is the set of nodes, $\mathrm{T}( G_t )$ is the set of time instants, and $\mathscr{E} \subseteq \mathrm{V}( G_t ) \times \mathrm{T}( G_t ) \times \mathrm{V}( G_t ) \times \mathrm{T}( G_t )$ is the set of edges regarding $G_t$.\footnote{ We assume that an undirected graph (or MAG) is a special case of a directed graph (or MAG) in which each edge represents two opposing arrows.} In addition, a static network $ G_s $ is\footnote{ See Definition~\ref{BdefStaticNetwork}.} also  a special case of MAGs, which is obtained from collapsing all the aspects in $ \mathscr{A} $ into just one aspect (i.e., into the set of vertices/nodes $ \mathrm{V} $) where the set of edges of this MAG is invariant under any relation other than the set of vertices/nodes---see also sub-determination in~\cite{Wehmuth2016b}. Thus, a static network $ G_s $ is a classical graph $ G=( \mathrm{V}, \mathrm{E} ) $ with all relations (e.g., in respect to time instants or layers) depending only on its set of edges $ \mathrm{E} $. However, for present purposes, a TVG is sufficient to deal with the SIS model and, hence, there is only one aspect we are collapsing. Therefore, we define a \emph{static network} $ G_s =(\mathrm{V},\mathscr{E}, \mathrm{T}) $ as a TVG in which, for every fixed values of $ t_i, t_j, t_k, t_h \in \mathrm{T} $, 
\[
\{ ( v_i, v_j ) \mid ( v_i, t_ i , v_j , t_j ) \in \mathscr{E} \}
=
\{ ( v_i, v_j ) \mid ( v_i, t_ k , v_j , t_h ) \in \mathscr{E} \}
\]

Inspired by the networks in~\cite{Pastor-Satorras2001a,Pastor-Satorras2002,Pastor-Satorras2001}, let $ \mathbb{G}_{SIS}( f, t ) $ be\footnote{ See Definition~\ref{BdefFamilyGSIS}.} a family of Time-Varying Graphs in which
every $ G_t \in  \mathbb{G}_{SIS}( f, t ) $ achieves stationary  prevalence $ \rho $  in a number of time intervals $ \Delta^*_{t}  $ (after an arbitrary time instant $ t \in \mathrm{T}(G_t) $ from which contagion may have been started in first place) following the SIS scheme. Thus, $ \mathbb{G}_{SIS}( f, t ) $ defines a family of dynamic networks~\cite{Guimaraes2013,Costa2015a} that follows the SIS model. Since we have defined static networks as a special case of dynamic networks, family $ \mathbb{G}_{SIS}( f, t )  $ can be seen as a generalization of the model presented in~\cite{Pastor-Satorras2001a,Pastor-Satorras2002,Pastor-Satorras2001} to dynamic networks. Since function $ f $ and time instant $ t $ are not specified in the condition of the set $ \mathbb{G}_{SIS}( f, t ) $, then this family is independent of the choice of $ ( f, t ) $. However, the reader will see that this is crucial for extending the results in~\cite{Abrahao2017} in order to build the proof of Theorem~\ref{BthmMainCentralTimeSIS} and Corollary~\ref{BcorSIS}.


In addition, we define\footnote{ See Definition~\ref{BdefFamilyGBA}.}  a family  $ \mathbb{G}_{BA}( f, t ) $ of TVGs in $ \mathbb{G}_{SIS} $ that are static networks following a classical Barabási-Albert model~\cite{Barabasi1999,Barabasi1999a}. They have a scale-free distribution of connectivities as a consequence of an application of preferential attachment at the addition of each new node, which results in a degree distribution in the form of a power law $ P(k) \sim \frac{2m^2}{ k^3} $ as the number of nodes goes to infinity. The finite number of nodes of each graph in this family may vary from $1$ to $\infty$ as each new node is added with $m$ edges linked to previous nodes $i$ under probability distribution
\[
\Pi( k_i ) = \frac{ k_i }{ \sum_j k_j }
\]
\noindent Note that, as shown in~\cite{Pastor-Satorras2001,Pastor-Satorras2001a,Pastor-Satorras2002}, these networks in $ \mathbb{G}_{BA}( f, t ) $ are expected to display a stationary prevalence 
\[
\rho \sim \exp( - \frac{1}{m \lambda}) 
\]
\noindent for a large enough network size and for a small enough spreading rate $ \lambda $. If these two conditions are met, then $ \mathbb{G}_{BA}( f, t ) \subseteq \mathbb{G}_{SIS}( f, t )  $. Therefore, family $ \mathbb{G}_{BA}( f, t ) $ is defined to directly correspond to the networks presented in~\cite{Pastor-Satorras2001,Pastor-Satorras2001a,Pastor-Satorras2002}.

Thus, $ \mathfrak{N'}_{BB}( N, f, t, j ) $ is a synchronous algorithmic network populated by\footnote{ Note that in our model, once the first cycle is started, the population remains fixed. Thus, during the cycles (i.e., when the algorithmic networks is running its computations) no new node is created and no node is ``killed''.  } $N$ randomly generated nodes such that, after the first (or arbitrary $c_0$ cycles) cycle, it starts a diffusion process of the biggest partial output (given at the end of the first cycle) determined by network $ G_t $ that belongs to a family of graphs $ \mathbb{G}_{SIS}( f, t ) $---remember that each network in $ \mathbb{G}_{SIS}( f, t ) $ follows a SIS contagion scheme. At the first time instant each node may receive a network input $w$, which is given to every node in the network, and runs separately (i.e. not networked), returning its respective first partial output. At the last time instant contagion stops and one cycle (or more) is spent in order to make each node to return a final output.

\noindent \\

\section{Expected local emergent open-endedness from a SIS model}\label{sectionResults}

In this section, we present the central theorem and its two corollaries with the purpose of showing that $ \mathfrak{N'}_{BB}( N, f, t, j ) = ( G_t, \mathfrak{P'}_{BB}(N), b_j ) $ is an algorithmic network capable of exhibiting expected (local) emergent open-endedness (see EEOE in \cite{Abrahao2017})\footnote{ See Definition~\ref{BdefEEOESIS}.}. We show that it occurs under certain topological conditions of the graph $G_t$ in which the prevalence (or average density of ``infected'' nodes) becomes stationary within a computably bigger time interval. During these time intervals, the algorithmic network is running under the Imitation-of-the-Fittest Protocol with a SIS contagion scheme. As in \cite{Abrahao2017}, the proof follows from the fact that there is a trade-off between the prevalence and the cycle-bounded conditional halting probability\footnote{ See Definition~\ref{BdefOmegawcSIS}.} $ \Omega( w, x ) $, where $w$ is the initial network input and $x$ is the number of cycles, when estimating the lower bound for the expected (local) emergent algorithmic complexity of a node. Note that in algorithmic networks $ \mathfrak{N'}_{BB}( N, f, t, j ) $ the stationary prevalence $\rho$ becomes\footnote{ See Section~\ref{appendixsubsectionPrevalence} in Appendix.} exactly equal to the stationary average density of ``infected'' nodes $ {  \tau_{\mathbf{E}(\rho)}( N,f,t_i ) }|_{ t }^{ t' } $ in a time interval between $t$ and $t'$ in which contagion started at time instant $t_i$.

Moreover, once these topological properties are met, the concept of central time (denoted as $t_{cen_1}$) to trigger expected emergent open-endedness within the minimum number of cycles becomes well-defined. We define\footnote{ See Definition~\ref{BdefTimecentrality1SIS}} the central time $ t_{cen_1} $ in generating \emph{unlimited} expected emergent algorithmic complexity of a node (i.e., expected \emph{local} emergent algorithmic complexity) in a network $ \mathfrak{N'}_{BB}(N,f,t_{z_0}, j ) $ during $ c( t_{cen_1}(c) + f( N, t_{cen_1}(c)  ) + 1 ) $ cycles, where $c(x)$ is a non-decreasing total computable function and $f$ is an arbitrary function\footnote{ See Definition \ref{BdefN'_BB}.}, as the minimum time instant $t$ in which the expected local emergent algorithmic complexity goes to infinity as $N \to \infty$ after $ c( t + f( N, t  ) + 1 ) $ cycles. Note that the arbitrarily chosen function $f$ may not behave monotonically with $t$ in general. 

The expected local\footnote{ The term ``local'' here refers to the emergent algorithmic complexity of a node. The investigation of the emergent algorithmic complexity of the population as a whole, as also mentioned in~\cite{Abrahao2017}, is out of our current scope. One may call this latter as \emph{global} emergent algorithmic complexity. Note that it may behave differently from the local one. Thus, we leave the investigation of expected \emph{global} emergent open-endedness for future research.} emergent algorithmic complexity is defined\footnote{ See Definition~\ref{BdefEAC}.} in~\cite{Abrahao2017,Abrahao2016b} as the number of extra bits of algorithmic complexity (or information) that emerges from a comparison of the algorithmic complexity of the final output of a networked node with the algorithmic complexity of the final output of the same node in the case it was isolated.

As in \cite{Abrahao2017}, the main idea behind the construction of the proof of Theorem \ref{thmMainCentralTimeSIS} comes from combining an estimation of a lower bound for the average algorithmic complexity of a \emph{networked} node and an estimation of an upper bound for the expected algorithmic complexity of an \emph{isolated} node. While the estimation of the former comes from the very BBIG dynamics in a SIS contagion scheme, the estimation of the latter comes from the law of large numbers, Gibb's inequality, and algorithmic information theory applied on the randomly generated population $ \mathfrak{P'}_{BB}(N) $, which is analogously the same as $ \mathfrak{P}_{BB}(N) $ in \cite{Abrahao2017} for the isolated case. Thus, calculating the former estimation minus the latter gives\footnote{ See Definition~\ref{BdefEEACN'_BB}.} directly a lower bound for the expected local emergent algorithmic complexity of a node.

In this section, we present short proofs of Theorem \ref{BthmMainCentralTimeSIS} and Corollary \ref{BcorSIS}. These proof steps are based on a direct analogy to the proof steps developed in~\cite{Abrahao2017}, so only in Corollaries \ref{BcorSIS} and \ref{BcorBA} our new model would introduce new conceptual substantial differences in the mathematical formal text. For complete and self-contained definitions, lemmas, theorems, and corollaries, see Appendix~\ref{appendix}.

\subsection{Central time to trigger EEOE}

\begin{thm}[or extended Theorem \ref{thmMainCentralTimeSIS}] \label{BthmMainCentralTimeSIS}
	Let $ w \in \mathbf{L_U} $ be a network input. Let $ 0 < N \in \mathbb{N} $. Let $f$ be an arbitrary function where 
	\[
	\myfunc{f}{ \mathbb{N^*} \times X \subseteq \mathrm{T}(G_t)  } { \mathbb{N} } { ( x, t ) } { y }
	\]
	\noindent Let $ \myfunc{c}{ \mathbb{N} } { \mathfrak{C_{BB}} } { x } { c(x)=y } $, where $ \mathfrak{C_{BB}} $ is the set of cycles of the population $ \mathfrak{ P' }_{BB}(N) $, be a total computable non-decreasing function where 
	\[ 
	c(z + f( N, t_z  ) + 2) \geq c_0 + z + f( N, t_z  ) + 2 
	\]
	and
	\[
	c( z + f( N, t_z  ) + 2 ) - c_0 - 1 \leq t_{ |\mathrm{T}(G_t)|-1 } 
	\] 
	\noindent If there is $ 0 \leq z_0 \leq | \mathrm{T}(G_t) | -1 $ and $ \epsilon, \, \epsilon_2 > 0 $ such that\footnote{ $\lg(x)$ denotes the binary logarithm $ \log_2(x) $.} 
	\[
	z_0 + f( N, t_{z_0}  )  + 2
	= 
	\mathbf{ O }
	\left( \frac
	{ N^{ C } }
	{ \lg(N) } 
	\right)
	\]
	\noindent where
	\[ 
	0
	\leq
	C = 
	\]
	\[
	=
	\frac
	{
		{ \tau_{\mathbf{E}(\rho)}( N,f,t_{z_0}  ) }|_{ t_{z_0} }^{ c( z_0 + f( N, t_{z_0}  ) + 2 ) - c_0 - 1  }
		-
		\Omega(w, c_0 + z_0 + f( N, t_{z_0} ) + 2 )
		-
		\epsilon
	}
	{ \Omega(w,  c_0 + z_0 + f( N, t_{z_0}  ) + 2  ) }
	\leq
	\]
	\[
	\leq
	\frac{1}{ \epsilon_2 }
	\] \\
	\noindent and $ \mathfrak{N'}_{BB}(N,f,t_{z_0} ,j) = ( G_t, \mathfrak{ P' }_{BB}(N), b_j ) $ is well-defined.
	Then, there is $ t_{cen_1}(c) $ such that 
	\[
	 t_{cen_1}(c) \leq t_{ z_0 }
	\]
	\begin{proof}[Short proof]
		This proof follows from the six Lemmas, Theorem 8.1 and Corollary 8.1.1 in~\cite{Abrahao2017}.
		First, replace algorithmic network $ \mathfrak{N}_{BB}(N,f,t_{z_0},\tau,j) = ( G_t, \mathfrak{ P}_{BB}(N), b_j ) $ and its respective characteristics, e.g., population $ \mathfrak{ P}_{BB}(N) $ and family of graphs $ \mathbb{G}( f, t, \tau ) $, with $ \mathfrak{N'}_{BB}(N,f,t_{z_0} ,j) = ( G_t, \mathfrak{ P' }_{BB}(N), b_j ) $, $ \mathfrak{ P'}_{BB}(N) $,  $ \mathbb{G}_{SIS}( f, t ) $, etc in the six Lemmas, Theorem 8.1 and Corollary 8.1.1 in~\cite{Abrahao2017}.
		Note that in the proof of the sixth Lemma the average (singleton) diffusion density $ \tau_{\mathbf{E}(max)} $ is replaced with the prevalence $ \tau_{\mathbf{E}(\rho)} $.
		Also note that in Corollary 8.1.1 in~\cite{Abrahao2017} the last time instant $  t_{ z + f( N, t_z, \tau ) }  $ is replaced with $ c( z_0 + f( N, t_{z_0}  ) + 2 ) - c_0 - 1 $.
		Then, the proof of Theorem~\ref{BthmMainCentralTimeSIS} follows directly analogous to Theorem 8.2 in~\cite{Abrahao2017}.
	\end{proof}

\subsection{EEOE from a stationary prevalence}

\begin{corollaryundersubsection}[or extended Corollary~\ref{corSIS}]\label{BcorSIS}
	Let $ w \in \mathbf{L_U} $ be a network input. Let $ 0 < N \in \mathbb{N} $. Let $ {\mathfrak{N'}_{BB}}(N,f,t_{z_0} ,j) = ( G_t, \mathfrak{ P' }_{BB}(N), b_j ) $ be well-defined. Let $ \myfunc{c}{ \mathbb{N} } { \mathfrak{C_{BB}} } { x } { c(x)=y } $  be a total computable non-decreasing function where
	\[ 
	c(z_0 + f( N, t_{z_0}  ) + 2) \geq c_0 + z_0 + f( N, t_{z_0}  ) + 2 
	\]
	and
	\[
	c( z_0 + f( N, t_{z_0}  ) + 2 ) - c_0 - 1 \leq t_{ |\mathrm{T}(G_t)|-1 } 
	\]
	\noindent If 
	\[ f(N,t_{z_0} ) =\mathbf{O}\big( \lg( N ) \big) \]
	where every $ G_t \in \mathbb{G}_{SIS}( f, t_{z_0} ) $ achieves stationary  prevalence $ \rho $ in a number of time intervals
	\[
	\Delta^*_{t_{z_0}}  \leq c( z_0 + f( N, t_{z_0}  ) + 2 ) - c_0 - 1
	\]
	\noindent after time instant $ t_{z_0} $ and
	\[ 
	\rho \sim \exp( - \frac{1}{m \lambda}) 
	>
	\Omega(w, c_0 + z_0 + f( N, t_{z_0}  ) + 2 ) 
	\]
	\noindent then, there is $ t_{cen_1}(c) $ such that
	\[
	t_{cen_1}(c) \leq t_{ z_0 }
	\]

	\begin{proof}[Short proof]
	The proof follows directly from Theorem~\ref{BthmMainCentralTimeSIS} and the definition of the algorithmic network $ {\mathfrak{N'}_{BB}}(N,f,t_{z_0} ,j) = ( G_t, \mathfrak{ P' }_{BB}(N), b_j ) $ (see Definition~\ref{BdefN'_BB}) by noting that:
	\begin{align*}
	z_0 + f( N, t_{z_0}  )  + 2
	=
	z_0 + \mathbf{ O }\left( \lg(N) 
	\right)
	+ 2
	=
	\mathbf{ O }
	\left( \lg(N) 
	\right)
	\end{align*}
	\noindent and that there is $ \epsilon > 0 $ such that
	\begin{align*}
	& \frac{-1 - \epsilon}{ \epsilon_2 } 
	< 0 
	<
	C = 
	\frac
	{
		\frac{1}{ e^{ \left(\frac{1}{m \lambda} \right) } }
		-
		\Omega(w, c_0 + z_0 + f( N, t_{z_0}  ) + 2 )
		-
		\epsilon
	}
	{ \Omega(w,  c_0 + z_0 + f( N, t_{z_0}  ) + 2  ) }
	=
	\end{align*}
	\[
	=
	\frac
	{
		{ \tau_{\mathbf{E}(\rho)}( N,f,t_{z_0}  ) }|_{ t_{z_0} }^{ c( z_0 + f( N, t_{z_0}  ) + 2 ) - c_0 - 1  }
		-
		\Omega(w, c_0 + z_0 + f( N, t_{z_0}  ) + 2 )
		-
		\epsilon
	}
	{ \Omega(w,  c_0 + z_0 + f( N, t_{z_0}  ) + 2  ) }
	\leq
	\]
	\[
	\leq
	\frac
	{
		1
		-
		\Omega(w, c_0 + z_0 + f( N, t_{z_0}  ) + 2 )
		-
		\epsilon
	}
	{ \Omega(w,  c_0 + z_0 + f( N, t_{z_0}  ) + 2  ) }
	\leq
	\frac{1}{ \epsilon_2 }
	\]
	
	\end{proof}

\end{corollaryundersubsection}

\end{thm}

\subsection{EEOE from a scale-free algorithmic network}

\begin{corollaryundersubsection}\label{BcorBA}
	Let $ w \in \mathbf{L_U} $ be a network input. Let $ 0 < N \in \mathbb{N} $. Let $ {\mathfrak{N'}_{BB}}(N,f,t_{z_0} ,j) = ( G_s, \mathfrak{ P' }_{BB}(N), b_j ) $ be well-defined for every $ G_s \in \mathbb{G}_{BA} ( f, t ) $. Let $ \myfunc{c}{ \mathbb{N} } { \mathfrak{C_{BB}} } { x } { c(x)=y } $  be a total computable non-decreasing function where
	\[ 
	c(z_0 + f( N, t_{z_0}  ) + 2) \geq c_0 + z_0 + f( N, t_{z_0}  ) + 2 
	\]
	and
	\[
	c( z_0 + f( N, t_{z_0}  ) + 2 ) - c_0 - 1 \leq t_{ |\mathrm{T}(G_t)|-1 } 
	\]
	\noindent If 
	\[ f(N,t_{z_0} ) =\mathbf{O}\big( \lg( N ) \big) \]
	where every $ G_t \in \mathbb{G}_{SIS}( f, t_{z_0} ) $ achieves stationary  prevalence $ \rho $ in a number of time intervals
	\[
	\Delta^*_{t_{z_0}}  \leq c( z_0 + f( N, t_{z_0}  ) + 2 ) - c_0 - 1
	\]
	\noindent after time instant $ t_{z_0} $, then for a small enough value of $ \lambda = \frac{ \nu }{ \delta } $, there are $ t_{cen_1}(c) $ and a big enough value of $ m $ such that
	\[
	t_0 =  t_{cen_1}(c) \leq t_{ z_0 }
	\] 

	\begin{proof}
		Since, by supposition, $ {\mathfrak{N'}_{BB}}(N,f,t_{z_0} ,j) = ( G_s, \mathfrak{ P' }_{BB}(N), b_j ) $ and $ G_s \in \mathbb{G}_{BA} ( f, t ) $, then we will have from \cite{Pastor-Satorras2001,Pastor-Satorras2001a,Pastor-Satorras2002} that
		\[ 
		\mathbb{G}_{BA} \subseteq \mathbb{G}_{SIS}  
		\]
		\noindent and
		\[
		\rho \sim \exp( - \frac{1}{m \lambda}) 
		\]
		\noindent for sufficiently large populations and for a small enough value of $ \lambda $.
		Thus, as Theorem~\ref{BthmMainCentralTimeSIS} and Corollary~\ref{BcorSIS} hold where the population size tends to $ \infty $, we will have that condition
		\[ 
		\rho \sim \exp( - \frac{1}{m \lambda}) 
		>
		\Omega(w, c_0 + z_0 + f( N, t_{z_0}  ) + 2 ) 
		\]
		\noindent in Corollary~\ref{BcorSIS} holds for a big enough value of $m$ given a small enough value of $ \lambda $.
		Thus, from Corollary~\ref{BcorSIS}, we will have that there is $ t_{cen_1}(c) $ such that
		\[
		 t_{cen_1}(c) \leq t_{ z_0 }
		\]
		\noindent And, since every $ G_s $ is a static network, then 
		\[
		t_0 =  t_{cen_1}(c) \leq t_{ z_0 }
		\]
	\end{proof}

\end{corollaryundersubsection}

\section{Conclusion}\label{sectionConclusion}

In this article, we have presented a model for networked computable systems in order to investigate the problem of emergence of algorithmic complexity. In particular, we have mathematically investigated conditions that enable the triggering of emergent open-endedness, that is, the conditions that trigger an unlimited increase of emergent complexity as the population size grows toward infinity. We have shown that these conditions are met by dynamic networks (or static networks) that exhibit a stationary prevalence of ``infected'' nodes under a SIS model for contagion of the fittest randomly generated node.  As pointed in~\cite{Abrahao2017}, such research may be crucial for optimizing communication protocols in artificial networks of randomly generated systems which seek for a better solution to a problem.    

Our model for networked computable systems is based on that previously established in~\cite{Abrahao2017}. Nodes are randomly generated Turing machines that can send and receive information (partial outputs) as each node runs its computations until returning a final output and edges (or arrows) are communication channels. Thus, as defined in~\cite{Abrahao2017}, these algorithmic networks are composed of a synchronous population that follows a protocol of imitation of the ``best information'' shared by a neighbor. However, the present article introduced a variation on this model such that this protocol is followed under a Susceptible-Infected-Susceptible model~\cite{Pastor-Satorras2001,Pastor-Satorras2001a,Pastor-Satorras2002}.  

We have shown that, for big enough arbitrary values of $ m \in \mathbb{N} $ compared to the effective spreading rate $ \lambda $, if the time for achieving a stationary prevalence of ``infected'' nodes $  \rho \sim \exp(- 1 / m \lambda)  $ is upper bounded by a computably big enough function of $ \mathbf{O}( \log(N) ) $,\footnote{ For example, as a function of the expected diameter or average shortest path length in scale-free networks or in classical random networks \cite{Bollobas2004,Lewis2009}.} then a lower bound for the expected emergent algorithmic complexity/information of a node goes to infinity as the network/population size $ N $ goes to infinity. That is, the average local irreducible information that emerges when nodes are networked (from a comparison with the isolated case) is expected to always increase for large enough populations of randomly generated Turing machines. Thus, these dynamic (or static) algorithmic networks with stationary prevalence may cross the phase that we call expected local emergent open-endedness \cite{Abrahao2017} for sufficiently large randomly generated populations.

In addition, since our main result only depends on assuming a stationary prevalence in the form of $  \rho \sim \exp(- 1 / m \lambda)  $, we have shown as a corollary from our theorems and from~\cite{Pastor-Satorras2001,Pastor-Satorras2001a,Pastor-Satorras2002} that under the same conditions on $m$ and $\lambda$ the same lower bound holds for static algorithmic networks with a scale-free 
degree distribution in the form of a power law $ P(k) \sim \frac{2m^2}{ k^3} $ \cite{Barabasi1999}.  Therefore, synchronous algorithmic networks with a randomly generated population of computable systems and with a topology and a contagion model sufficiently close to the ones studied in ~\cite{Pastor-Satorras2001,Pastor-Satorras2001a,Pastor-Satorras2002} are also expected to display expected local emergent open-endedness. This suggests that contagion schemes like the SIS model, which have been shown to be important for studying epidemic and disease spreading and computer virus infections, may be also related to the emergence of complexity or irreducible information~\cite{Abrahao2016b,Abrahao2017} in networked systems.

Regarding only the lower bound for the expected emergent algorithmic complexity of a node, our main results show that a version of the halting probability for synchronous algorithmic networks may work like an asymptotic threshold for triggering expected local emergent open-endedness through a stationary prevalence. For example, in the case of the static algorithmic network with a Barabási-Albert scale-free degree distribution \cite{Barabasi1999}, we have shown that arbitrarily small values of the spreading rate can be overcome by big enough values of $m$ (i.e., the number of new edges per node addition) in order to surpass this ``threshold'', triggering the expected local emergent open-endedness. However, since we have only investigated a lower bound, this halting probability may not be actually the threshold for the actual expected emergent algorithmic complexity of a node. Thus, in order to study the existence of such threshold, we suggest for future research the investigation of an upper bound and an asymptotically tight bound for the expected local emergent algorithmic complexity.

\bibliographystyle{amsplain}
\bibliography{2.2.1_paper2018-3_FelipeKlausArtur.bib}


\section{Appendix}\label{appendix}


In this section, we present a self-contained appendix with definitions, notes and extended versions of the lemmas and theorems concerning section~\ref{sectionResults}. In order to improve readability and help check the proofs and definitions in comparison to the model and results in~\cite{Abrahao2017}, the definitions and notes that introduce new features or variations are marked with ``SIS''. The ones  without this mark are totally analogous to \cite{Abrahao2017}.

\subsection{Definition of MultiAspect Graphs}

\begin{Bdefinition}\label{BdefGraph}
	As defined in \cite{Costa2015a,Wehmuth2016b}, let $G=(\mathscr{A},\mathscr{E})$ be a graph, where $\mathscr{E}$ is the set of edges of the graph and $\mathscr{A}$ is a class of sets, each of which is an aspect.
	
	\begin{Bremarknote}
		Note that $\mathscr{E}$ determines the (dynamic or not) topology of $G$.
	\end{Bremarknote}

	\begin{Bremarknote}\label{noteStaticMAGs}
		Each aspect in $\mathscr{A}$ determines which variant of a graph $G$ will be (and how the set $\mathscr{E}$ will be defined). As in \cite{Costa2015a}, we will deal only with Time-Varying Graphs $G_t$ hereafter, so there will be only two aspects ($|\mathscr{A}|=2$): the set of nodes (or vertices) $\mathrm{V}(G_t)$ and the set of time instants $\mathrm{T}(G_t)$. An element in $ \mathrm{V}(G_t) \times \mathrm{T}(G_t) $ is a \textbf{composite vertex} (or composite node). The family of graphs $G$ that we will use in the present paper will be better explained in Definition \ref{BdefFamilyGSIS}. 
	\end{Bremarknote}
	
\end{Bdefinition}

\subsection{Definitions of general algorithmic networks}
\begin{Bdefinition}\label{BdefAN}
	We define an \textbf{algorithmic network} $ \mathfrak{N} = (G, \mathfrak{P}, b)$ upon a population of theoretical machines $\mathfrak{P}$, a graph $G=(\mathscr{A},\mathscr{E})$ and a function $b$ that makes aspects of $G$ correspond\footnote{ See Definition \ref{BdefFunctionbinAN}.} to properties of $\mathfrak{ P }$, so that a node in $\mathrm{ V(G) }$ corresponds one-to-one to an element of $\mathfrak{ P }$. The graph $G$ was previously defined in \ref{BdefGraph}, and we will define $\mathfrak{ P }$ and $b$ in definitions \ref{BdefPopulation} and \ref{BdefFunctionbinAN} , respectively.

	\begin{Bsubdefinition}\label{BdefPopulation}
		
		Let the \textbf{population} $\mathfrak{P}$ be a subset of $\mathrm{{L}}$ in which repetitions\footnote{ Thus, a population is a set or language which might contain repetitions among its elements. See also Definitions~\ref{BdefHaltSIS} and~\ref{BdefL'_BB}. } are allowed, where $\mathrm{{L}}$ is the language on which the chosen theoretical machine $\mathrm{{U}}$ are running. Each member of this population may receive inputs and return outputs through communication channels.
		
	\end{Bsubdefinition}
	
	\begin{Bsubremarknote}
		The choice of $\mathrm{{L}}$ and $\mathrm{{U}}$ determines the class of nodes/systems. For example, one may allow only time-bounded Turing machines in the population. In the present work, $\mathbf{L_U}$ will be a self-delimiting universal programming language for a extended universal Turing machine $\mathrm{\textbf{U'}}_R$ (see Definitions \ref{BdefU'_R} and \ref{BdefL'_BB}) --- i.e., an oracle Turing machine --- that returns zero whenever a non-halting computation occur.  
	\end{Bsubremarknote}
	
	\begin{Bsubsubdefinition}\label{BdefCycle} 
		Let $ \mathfrak{C} $ be a set of the maximum number of cycles that any node/program $o$ in the population $ \mathfrak{P} $ can perform in order to return a final output. A \textbf{node cycle} in an algorithmic network $ \mathfrak{N} $ is defined as a node/program\footnote{ Once there is a mapping of the set of nodes into the population of programs, the expression ``node/program'' becomes well-defined within the theory of algorithmic networks.} returning a \textbf{partial output} (which, depending on the language and the theoretical machine the nodes are running on, is equivalent to a node completing a halting computation)\footnote{ In the present article for example --- see Definition \ref{BdefU'_R}.}  and sharing (or not) this partial output  with its neighbors (accordingly to a specific information-sharing protocol or not --- see Definition \ref{BdefP'_BB}). 
	\end{Bsubsubdefinition}
	
	\begin{Bsubsubremarknote} 
		So, if the algorithmic network is asynchronous, a cycle can be seen as an individual communication round that doesn't depend on whether its neighbors are still running or not, while if the network is synchronous a cycle can be seen as the usual communication round in synchronous distributed computing. Also note that one may also refer to a \textbf{network cycle}, which denotes when all nodes of the algorithmic network have returned their \textbf{final outputs} (if it is the case). Thus, a network cycle must not be confused with a node cycle.
	\end{Bsubsubremarknote}
	
	\begin{Bsubdefinition}
		A \textbf{communication channel} between a pair of elements from $\mathfrak{P}$ is defined in $\mathscr{E}$ by an edge (whether directed or not) linking this pair of nodes/programs.
		
		\begin{Bsubremarknote}
			A directed edge (or arrow) determines which node/program sends an output to another node/program that takes this information as input. An undirected edge (or line) may be interpreted as two opposing arrows. 
		\end{Bsubremarknote}
		
	\end{Bsubdefinition}

	\begin{Bsubdefinition}\label{BdefFunctionbinAN}
		Let 
		\[ \myfunc{b}{ Y \subseteq \mathscr{A}(G) } { X \subseteq Pr(\mathfrak{P}) } { \mathbf{\overline{a}}  } { b( \mathbf{\overline{a}} ) = \mathbf{\overline{p_r}} } \]
		be a function that \textbf{maps} a subspace of aspects $Y$ in $\mathscr{A}$ into a subspace of properties $X$ in the set of properties $Pr(\mathfrak{P})$ of the respective population in the graph $G=(\mathscr{A},\mathscr{E})$ such that there is an bijective function $ f_{V\mathfrak{ P }} $ such that, for every $ (v,\mathbf{ \overline{x} }) \in Y \subseteq \mathscr{A}(G) $ where $ b( v,\mathbf{ \overline{x} } ) = ( o, b_{ | Y | - 1 }( \mathbf{ \overline{x} } ) ) \in X$, $v$ is a vertex, and $o$ is an element of $\mathfrak{ P }$
		\[ \myfunc{f_{V\mathfrak{ P }}}{ \mathrm{ V(G) } } { \mathfrak{P} } { v } { f_{V\mathfrak{ P }}(v) = o } \]

	\end{Bsubdefinition}

\end{Bdefinition}

\begin{Bdefinition}
	We say an element $ o_i \in \mathfrak{P} $ is \textbf{networked} \textit{iff} there is $ \mathfrak{N} = (G, \mathfrak{P}, b) $, where $G$ has a non-empty set of edges, such that $o_i$ is running on it.
	
	\begin{Bsubdefinition}
		We say $o_i$ is \textbf{isolated} otherwise. That is, it is only functioning as an element of $ \mathfrak{P} $ and not $ \mathfrak{N} = (G, \mathfrak{P}, b) $.
	\end{Bsubdefinition}
\end{Bdefinition}

\begin{Bdefinition}\label{BdefNetworkinput}
	We say that an input $ w \in \mathrm{L} $ is a \textbf{network input} \textit{iff} it is the only external source of information every node/program receives and it is given to every node/program before the algorithmic network begins any computation. 
	
	\begin{Bremarknote}
		Note that letter $w$ may also appear across the text as denoting an arbitrary element of a language. It will be specified in the assumptions before $w$ appears or in the statement of the definition, lemma, theorem or corollary.
	\end{Bremarknote}
\end{Bdefinition}

\noindent \\

\subsection{Definitions on networks and graphs}

\begin{Bdefinition}\label{BdefTVG} 
	As defined in \cite{Costa2015a}, let $G_t=(\mathrm{V},\mathscr{E},\mathrm{T})$ be a \textbf{Time-Varying Graph} (TVG), where $\mathrm{V}$ is the set of nodes, $\mathrm{T}$ is the set of time instants, and $\mathscr{E} \subseteq \mathrm{V} \times \mathrm{T} \times \mathrm{V} \times \mathrm{T}$ is the set of edges.
	
	\begin{Bsubnotation}
		Let $\mathrm{V}(G_t)$ denote the set of nodes (or vertices) of $G_t$.
	\end{Bsubnotation}
	
	\begin{Bsubnotation}
		Let $|\mathrm{V}(G_t)|$ be the size of the set of nodes in $G_t$.
	\end{Bsubnotation}
	
	\begin{Bsubnotation}
		Let $\mathrm{T}(G_t)$ denote the set of time instants in $G_t=(\mathrm{V},\mathscr{E},\mathrm{T})$. 
	\end{Bsubnotation}
	
	\begin{Bsubnotation}
		Let $ G_t( t ) $ denote the graph $G_t$ at time instant $ t \in \mathrm{T}(G_t) $
	\end{Bsubnotation}

	\begin{Bsubdefinition}
		We define the set of time instants of the graph $G_t$ as $ \mathrm{T}(G_t)=\{t_0, t_1, \dotsc, t_{|\mathrm{T}(G_t)|-1} \} $.
	\end{Bsubdefinition}
	
	\begin{Bsubremarknote}
		For the sake of simplifying our notations in the theorems below one can take a natural ordering for $ \mathrm{T}(G_t) $ such that
		\[ 
		\forall i \in \mathbb{N} \; \left( \, 0 \leq i \leq | \mathrm{T}(G_t) | - 1 \implies t_i = i + 1 \, \right) 
		\].
	\end{Bsubremarknote}
	
\end{Bdefinition}

\begin{Bdefinition}
	Let $d_t(G_t, t_i, u, \tau)$ be the minimum number of time instants (steps, time intervals \cite{Pan2011a} or, in our case, cycles) for a diffusion starting on node $u$ at time instant $t_i$ to reach a fraction $\tau$ of nodes in the graph $G_t$.
	
	\begin{Bsubnotation}\label{BdefTemporaldiameter}
		Let $D(G_t,t)$ denote the \emph{temporal diffusion diameter} of the graph $G_t$ taking time instant $t$ as the starting time instant of the diffusion process. That is,
		\[
		D(G_t,t) =
		\begin{cases}
		max \{ x \mid \, x=d_t(G_t,t,u,1) \, \land \, u \in \mathrm{V}(G_t) \}; \\
		\infty \quad \quad if \, \,  \exists u \in \mathrm{V}(G_t) \forall x \in \mathbb{N}  \big( x \neq d_t(G_t,t,u,1) \big);
		\end{cases}
		\]
	\end{Bsubnotation}
\end{Bdefinition}

\begin{Bdefinition}[SIS]\label{BdefStaticNetwork}
	Let $G_s=(\mathrm{V},\mathscr{E}, \mathrm{T})$ be a \textbf{static network}, where $G_s$ is a TVG in which, for every fixed values of $ t_i, t_j, t_k, t_h \in \mathrm{T} $, 
	\[
	\{ ( v_i, v_j ) \mid ( v_i, t_ i , v_j , t_j ) \in \mathscr{E} \}
	=
	\{ ( v_i, v_j ) \mid ( v_i, t_ k , v_j , t_h ) \in \mathscr{E} \}
	\]
	
	\begin{Bremarknote}[SIS]
		A general way to define a classical static graph is from collapsing all the aspects in $ \mathscr{A} $ into just one aspect (i.e., into the set of vertices/nodes $ \mathrm{V} $) where the set of edges of this MAG is invariant under any relation other than the set of vertices/nodes --- see also sub-determination in~\cite{Wehmuth2016b}. Thus, a static network is a classical static graph $ G=( \mathrm{V}, \mathrm{E} ) $ for all relations depending only on its set of edges $ E $. However, for present purposes a TVG (which is a subclass of MAGs) is sufficient to deal with the SIS model and, hence, there is only one aspect we are collapsing.
	\end{Bremarknote}

	\begin{Bremarknote}[SIS]
		This definition~\ref{BdefStaticNetwork} is quite general, so that it comprises even the case in which there could be arrows pointing backwards in time. However, for the present purposes in the static networks in~\cite{Pastor-Satorras2001,Pastor-Satorras2001a,Pastor-Satorras2002} a static network can be more easily defined as a snapshot dynamic network \cite{Rossetti2018} in which the topology is exactly the same in each temporal snapshot of the network. 
	\end{Bremarknote}
	
\end{Bdefinition}




\begin{Bdefinition}[SIS]\label{BdefFamilyGSIS}
	Inspired by the networks defined in~\cite{Pastor-Satorras2001a,Pastor-Satorras2002,Pastor-Satorras2001}, let 
	\begin{align*}
	\mathbb{G}_{SIS}( f, t ) = \Bigl\{ G_t \Big|  & \, i = |\mathrm{V}(G_t)| \,\land \,  \forall i \in \mathbb{N^*} \exists!G_t \in \mathbb{G}_{SIS}( f, t )  (\, |\mathrm{V}(G_t)|=i \,)  \Bigr\}
	\end{align*} 
	\noindent where 
	\[
	\myfunc{f}{ \mathbb{N^*} \times X \subseteq \mathrm{T}(G_t)  } { \mathbb{N} } { ( x, t ) } { y }
	\]
	\noindent be a family of TVGs that depends on the choice of function $f$, the time instant $ t $ and on the fact that every $ G_t \in  \mathbb{G}_{SIS} $ achieves stationary  prevalence $ \rho $  (i.e., the average density of ``infected'' nodes  $ \tau_{\mathbf{E}(\rho)} $ in Definition~\ref{BdefPrevalence} ) in a number of time intervals $ \Delta^*_t  $ (after an arbitrary time instant $ t \in \mathrm{T}(G_t) $ ) following a Susceptible-Infected-Susceptible (SIS) contagion scheme.
	
	\begin{Bremarknote}[SIS]
		Thus, $ \mathbb{G}_{SIS}( f, t ) $ defines a family of dynamic networks \cite{Pan2011a,Guimaraes2013,Costa2015a,Rossetti2018} that follows the SIS model. 
	\end{Bremarknote}
	
	\begin{Bremarknote}
		Since function $ f $ and time instant $ t $ are not specified in the condition of the set $ \mathbb{G}_{SIS}( f, t ) $, then this family is independent of the choice of $ ( f, t ) $. However, the reader will see that this is crucial for extending the results in~\cite{Abrahao2017} in order to build the proof of Theorem~\ref{thmMainCentralTimeSIS} and Corollary~\ref{corSIS}.
	\end{Bremarknote}

\end{Bdefinition}


\begin{Bdefinition}[SIS] \label{BdefFamilyGBA}
	We define a family of static networks analogous to the ones presented in~\cite{Pastor-Satorras2001,Pastor-Satorras2001a,Pastor-Satorras2002} as
	\begin{align*}
		\mathbb{G}_{BA}( f, t ) = \Bigl\{ G_s \Big|  & \, i = |\mathrm{ V(G_s) }| \,\land \,  \exists!G_s \in \mathbb{G}_{BA}( f, t )  (\, |\mathrm{ V(G_s) }|=i \,)  \Bigr\}
	\end{align*} 
	\noindent where 
	\[
	\myfunc{f}{ \mathbb{N^*} \times X \subseteq \mathrm{T}(G_t)  } { \mathbb{N} } { ( x, t ) } { y }
	\]
	\noindent Thus, $ \mathbb{G}_{BA}( f, t )  $ is a family of TVGs that are static networks following a classical Barabási-Albert model~\cite{Barabasi1999,Barabasi1999a} such that every $ G_s \in  \mathbb{G}_{BA} $ achieves stationary  prevalence $ \rho $  (i.e., the average density of ``infected'' nodes  $ \tau_{\mathbf{E}(\rho)} $ in Definition~\ref{BdefPrevalence} ) in a number of time intervals $ \Delta^*_t  $ (after an arbitrary time instant $ t \in \mathrm{T}(G_t) $ ) following a Susceptible-Infected-Susceptible (SIS) contagion scheme. These networks have a scale-free distribution of connectivities as a consequence of an application of preferential attachment at the addition of each new node, which results in a degree distribution in the form of a power law $ P(k) \sim \frac{2m^2}{ k^3} $ as the number of nodes goes to infinity. The number of nodes of each graph in this family may vary from $1$ to $\infty$ as each new node is added with $m$ edges linked to previous nodes $i$ under probability distribution
	\[
	\Pi( k_i ) = \frac{ k_i }{ \sum_j k_j }
	\]
	
	\begin{Bremarknote}[SIS]
		Note that, as shown in~\cite{Pastor-Satorras2001,Pastor-Satorras2001a,Pastor-Satorras2002}, these networks in $ \mathbb{G}_{BA}( f, t ) $ are expected to display a stationary prevalence 
		\[
		\rho \sim \exp( - \frac{1}{m \lambda}) 
		\]
		\noindent for a large enough network size and for a small enough spreading rate $ \lambda $. Thus, if these two conditions are met, then $ \mathbb{G}_{BA}( f, t )  \subseteq \mathbb{G}_{SIS}( f, t )  $.
	\end{Bremarknote}

\end{Bdefinition}

\subsection{Definitions on Turing machines and languages}\label{appendixsubsectionmachinesandlanguage}

\begin{Bnotation}
	Let $\lg(x)$ denote the binary logarithm $\log_{2}(x)$.
\end{Bnotation}

\begin{Bnotation}\label{BdefFunctionU}
	Let $ \mathbf{U}(x) $ denote the output of a universal Turing machine $\mathbf{U}$ when $x$ is given as input in its tape. So, $ \mathbf{U}(x) $ denote a partial recursive function $\varphi_{\mathbf{U}}(x)$ that is a universal partial function \cite{Rogers1987,Li1997}. If $x$ is a non-halting program on $\mathbf{U}$, then this function $\mathbf{U}(x)$ is undefined for $x$.
\end{Bnotation}

\begin{Bnotation}\label{BdefConcatenation}
	Let $ \mathrm{\textbf{L}}_{\mathrm{\textbf{U}}} $ be a binary self-delimiting universal programming language for a universal Turing machine $\mathrm{\textbf{U}}$ such that there is a concatenation of strings $w_1, \dots , w_k$ in the language $ \mathrm{\textbf{L}}_{\mathrm{\textbf{U}}} $, which preserves\footnote{ For example, by adding a prefix to the entire concatenated string $ w_1 w_2  \dots  w_k $ that encodes the number of concatenations. Note that each string was already self-delimiting. See also \cite{Abrahao2016}. } the self-delimiting (prefix-free) property of the resulting string, denoted by 
	\[
	w_1 \circ  \dots  \circ w_k \in \mathrm{\textbf{L}}_{\mathrm{\textbf{U}}}
	\]
\end{Bnotation}

\begin{Bnotation}\label{BdefAlgComp}
	Let $ \mathrm{\textbf{L}}_{\mathrm{\textbf{U}}} $ be a binary self-delimiting universal programming language for a universal Turing machine $\mathrm{\textbf{U}}$.
	The (prefix) \textbf{algorithmic complexity} (Kolmogorov complexity, program-size complexity or Solomonoff-Komogorov-Chaitin complexity) of a string $ w \in \mathrm{\textbf{L}}_\mathrm{\textbf{U}} $, denoted by $A(w)$, is the size of the smallest program $p^* \in \mathrm{\textbf{L}}_\mathrm{\textbf{U}}$ such that $ \mathrm{\textbf{U}}(p^*) = w $.

	\begin{BnoteunderBnotation}[SIS]
		The reader may also find in the literature the prefix algorithmic complexity denoted by $H(w)$ or --- more frequently used --- $K(w)$. As introduced in~\cite{Abrahao2017}, this work might have several intersections with other fields. Thus, we choose a self-explaining approach on notation in order to avoid ambiguity and notation conflicts in future work. We denote the (prefix) algorithmic complexity/information\footnote{ That is, the algorithmic information contained in a object about itself \cite{Li1997}.} by $ I_A( w ) $. However, for the sake of simplifying our notation, we chose to denote it only by $ A( w ) $ in \cite{Abrahao2017} and in the present article.  
	\end{BnoteunderBnotation} 
	
\end{Bnotation}



\begin{Bdefinition}[SIS]\label{BdefU'_R}
	Given a binary self-delimiting universal programming language $ \mathrm{\textbf{L}}_{\mathrm{\textbf{U}}} $ for a universal Turing machine $\mathrm{\textbf{U}}$, where there is a constant $ \epsilon \in \mathbb{R} $, with $ 0 < \epsilon \leq 1 $, and a constant $ 0 \leq C_{L} \in \mathbb{N} $ such that, for every $ N \in \mathbb{N} $, 
	\[
	A(N) \leq \lg(N) + ( 1+\epsilon )\lg(\lg(N)) + C_{L}
	\] 
	\noindent we then define an oracle\footnote{ Or any hypercomputer with a respective Turing degree higher than or equal to $1$.} Turing machine $\mathrm{\textbf{U'}}_R$ such that, for arbitrarily chosen $ S, I \in  \mathrm{\textbf{L}}_{\mathrm{\textbf{U}}} $,
	\begin{enumerate} 
		\item for every $  p, w \in \mathbf{L_U} $ 
		\[ {\mathrm{\textbf{U'}}}(w) =
	\begin{cases}
	{\mathrm{\textbf{U}}}(w) + 1 & \quad \text{  if  } \mathrm{\textbf{U}} \text{ halts on } w \\
	0 & \quad \text{  if  } \mathrm{\textbf{U}} \text{ does not halt on } w \\
	\end{cases}
	\]
	
	\item $ \exists P_{prot'} \in \mathbf{L_U} $ such that, for every $w, p \in \mathrm{\textbf{L}}_{\mathrm{\textbf{U}}}$, 
	\[ \mathbf{U'}(P_{prot'} \circ p \circ z \circ w) =
	\begin{cases}
	\mathbf{U'}(P_{prot_\nu} \circ p \circ w)	& \,  \text{ if } x=``1'' \in \mathbf{L_U} \text{ and } y=S \in \mathbf{L_U}	\text{ and } z=x \circ y \circ c \\
	\mathbf{U'}(P_{prot_\nu} \circ p \circ w)	& \, \text{ if } x=``0'' \in \mathbf{L_U} \text{ and } y=S \in \mathbf{L_U}		\text{ and } z=x \circ y \circ c \\
	\mathbf{U'}(P_{prot_\delta} \circ p \circ w) & \, \text{ if } x=``1''\text{ and } y=I \in \mathbf{L_U}	 	\text{ and } z=x \circ y \circ c \\
	\mathbf{U'}(P_{prot_\delta} \circ p \circ w) & \, \text{ if } x=``0''\text{ and } y=I \in \mathbf{L_U}		\text{ and } z=x \circ y \circ c  \\
	\mathbf{U'}(P_{prot_1} \circ p \circ w) & \, \text{ if } c=1  \\
	\mathbf{U'}(P_{prot_f} \circ p \circ w) & \, \text{ if } c=max \{ c \in \mathfrak{C}_{BB} \}  \\
	\end{cases}
	\]
	
	\item $ \exists P_{prot'_R} \in \mathbf{L_U} $ such that, for every $ w, p \in \mathbf{L_U} $ , $ \mathbf{U'}_R ( P_{prot'_R} \circ w ) $ has access to a randomly generated number $ x \in \{ 1, 0 \} $ in way such that
	\[ \mathbf{U'}_R ( P_{prot'_R} \circ p \circ z \circ w ) =
	\begin{cases}
		 \mathbf{U'}(P_{prot'} \circ p \circ 1 \circ y \circ c \circ w) & \quad \text{ with probability } \nu  \text{ if } y=S \in \mathbf{L_U} \\
		 \mathbf{U'}(P_{prot'} \circ p \circ 0 \circ y \circ c \circ w) & \quad \text{ with probability } 1- \nu \text{ if } y=S \in \mathbf{L_U} \\
		 \mathbf{U'}(P_{prot'} \circ p \circ 1 \circ y \circ c \circ w) & \quad \text{ with probability } \delta \text{ if } y=I \in \mathbf{L_U} \\
		 \mathbf{U'}(P_{prot'} \circ p \circ 0 \circ y \circ c \circ w) & \quad \text{ with probability } 1- \delta \text{ if } y=I \in \mathbf{L_U} \\
		 \mathbf{U'}(P_{prot'} \circ p \circ z \circ w) & \quad \text{ otherwise } \\
	\end{cases}
	\]
	
	\item For every $ P,w \in \mathbf{L_U} $, if $P$ accesses a randomly generated number, then 
	\[
	\mathbf{U'}_R ( P \circ w )  = \mathbf{U'}_R ( P_{prot'_R} \circ w ) 
	\]
	\end{enumerate}
	\begin{Bremarknote}[SIS]
		The oracle Turing machine is basically (except for a trivial bijection and a specific and determined way to access a randomly generated number) the same as the chosen universal Turing machine. The oracle is only triggered to know whether the program halts or not in first place (see \cite{Chaitin2012,Chaitin2014,Chaitin2018}) or to call a function or procedure that reads the output from a external probabilistic source\footnote{ A non published paper by Jef Raskin containing discussions on such machines was hosted at \href{http://humane.sourceforge.net/unpublished/turing_machine.html}{ Computers are not Turing Machines}. }). In fact, most programming languages have a function random, the difference here is that it is called depending on the Susceptible-Infected-Susceptible procedure and the source is independent and identically distributed with the respective probability distributions. Also note that $ \mathbf{U'}( w ) $ and $ \mathbf{U'}_R( w ) $ are \textbf{total} functions, and not a partial function as $ \mathbf{U}(w) $ --- see Definition \ref{BdefFunctionU}.
	\end{Bremarknote}

	\begin{Bremarknote}
		Note that from algorithmic information theory (AIT) we know that the algorithmic complexity (see Definition \ref{BdefAlgComp} ) $ A(\textbf{U'}(w)) $ only differs from $ A(\textbf{U}(w)) $ by a constant, if $\textbf{U}$ halts on $w$. This constant is always limited by the size of the smallest program that adds (or subtracts --- whichever is larger) $1$ to any other halting computation. Therefore, these machines belong to an algorithmic complexity equivalence class (the modulus of the subtraction upper bounded by a constant --- see also the invariance theorem in~\cite{Li1997}) everytime $w$ is a halting program. This is the reason why the algorithmic complexity of the final outputs of node/programs in $ \mathfrak{N'}_{BB} (N, f, t, j) $ only differ by a constant, should nodes be halting programs. Also note that, since a non-halting program gives an output always equal to zero when running on machine $ \textbf{U'} $,  the algorithmic complexity of the output of $w$ on $ \textbf{U'} $ is always equal to a constant (see Lemma~\ref{lemmaComplexityonBarHaltSIS}). Then, these make Lemma~\ref{lemmaComplexityp_iSIS} and the Definition~\ref{BdefEAC} sound.
	\end{Bremarknote}
	
\end{Bdefinition}
	
\subsection{Definitions on the populations of algorithmic networks}
	
	\begin{Bdefinition}\label{BdefRandompopulation}
		We say a population $ \mathfrak{ P } \subseteq \mathbf{L_U} $ is \textbf{randomly generated} \textit{iff} $ \mathfrak{ P } $ is a sample generated by $ | \mathfrak{ P } | $  i.i.d. trials accordingly to a probability distribution where, for a constant $C$, 
		
		\[
		p \in \mathfrak{ P }
		\]
		\[ iff \]
		\[
		\mathbf{P} \left[ \, p \, \right] = C \frac{1}{2^{ | p | }}
		\]
		

		\begin{Bremarknote}
			The constant $C$ is important for us because it allows us to characterise population $ \mathfrak{P'}_{BB}(N) $ in Definition \ref{BdefP'_BB} as randomly generated, taking into account that there are global information-sharing protocols that could not be previously determined. However, our forthcoming proofs stems from the idea that only the suffixes were randomly generated and the global information-sharing protocol was previously given (i.e., determined) as an assumption in our Lemmas, Theorems and Corollaries. Hence, this constant $C$ is not taken into account in the present work, so that one assume $ C=1 $.
		\end{Bremarknote}
	
		\begin{Bremarknote}
			The reader is invited to note that this constant $C$ would only affect Lemma~\ref{lemmaSLLNandAITSIS} as a subtractive constant and the other Lemmas, Theorems and Corollaries as a multiplicative constant on where $ \Omega( w, c(x) ) $ already appears as multiplicative constant. Therefore, since we are investigating asymptotic behaviors as the population size grows toward infinity, our final results hold in the case of considering the probability of generating the global information-sharing protocol too. 
		\end{Bremarknote}
		
	\end{Bdefinition}

	\begin{Bdefinition}\label{Bdefsensitivetooracles}
		We say a population $\mathfrak{P}$ is \textbf{sensitive to oracles} \textit{iff} whenever an oracle is triggered during any cycle in order to return a partial output the final output of the respective node/program is also $0$.\footnote{ Since we are assuming $0$ as the assigned non-halting output for $\mathbf{U'}$ in relation to the machine $\textbf{U}$.} Or more formally: 
		
		Let $ p_{net_{ \mathbf{U} } } $ be a program such that $ \mathbf{U}( p_{net_{ \mathbf{U} } } \circ o_i \circ c ) $ computes on machine $ \mathbf{U} $ cycle-by-cycle what a node/program $o_i \in \mathfrak{P} $ does on machine $ \mathbf{U'}_R $ until cycle $c$ when networked. Let $ p_{iso_{ \mathbf{U} } } $ be a program such that $ \mathbf{U}( p_{iso_{ \mathbf{U} } } \circ o_i \circ c ) $ computes on machine $ \mathbf{U} $ cycle-by-cycle what a node/program $o_i \in \mathfrak{P} $ does on machine $ \mathbf{U'}_R $ until cycle $c$ when isolated. Let $p_{o_i,c}$ be the \textbf{partial output} sent by node/program $o_i$ at the end of cycle $c$. Also, $p_{o_i,max \{ c \mid c \in \mathfrak{C} \}}$ denotes the final output of the node/program $o_i$. Then, for every $ o_i \in \mathfrak{P} $, if there is $c$ such that $ \nexists p_{net_{ \mathbf{U} } } $ or $ \nexists p_{iso_{ \mathbf{U} } } $, then the respective networked or isolated final output
		$ p_{o_i,max \{ c \mid c \in \mathfrak{C} \}} = 0 $.  
		
	
		\begin{Bremarknote}
			So the property of being sensitive to oracles may be also understood as the extension of machine $ \mathbf{U'}_R $ defined in~\ref{BdefU'_R} for returning zero for non-halting network cycles (see Definition~\ref{BdefCycle}). This way, one may also define another oracle machine $ \mathbf{U''}_R $ that runs the entire algorithmic network sensitive to oracles.
		\end{Bremarknote}
	\end{Bdefinition}
	
	\begin{Bdefinition}\label{BdefSynchronous}
		In a \textbf{synchronous} population of an algorithmic network each node/program is \textbf{only} allowed to receive inputs from its incoming neighbors and to send information to its outgoing neighbors at the end of each cycle (or communication round\footnote{ As in distributed computing.}), except for the last cycle. Each node cycle always begins and ends at the same time even if the computation time of the nodes/programs is arbitrarily different. Or more formally: 
		
		Let $G_t$ be a Time-Varying Graph such that each time interval corresponds to nodes sending information to their neighbors at the same time.
		Then, there is a partial function $f$ such that for every $ c \in \mathfrak{C} $ there is a constant $ t \in \mathrm{T}(G_t) $ such that for every $ o_i \in \mathfrak{P} $ where $ c(o_i) = c $
		\[ \myfunc{f}{ \mathfrak{C}(o_i) } { \mathrm{T}(G_t) } { c } { f(c)=t }  \]
		\noindent where $\mathfrak{C}(o_i)$ is the set of node cycles of node $o_i$.
		
		\begin{Bremarknote}
			If the population is \textbf{isolated}, then only each partial output counts in the respective individual loop.
		\end{Bremarknote}
		
		\begin{Bsubdefinition}
			In the \textbf{last cycle}, every node only returns its \textbf{final output}. 
		\end{Bsubdefinition}

	\end{Bdefinition}



\begin{Bdefinition}[SIS]\label{BdefIFPSIS}
	We say a networked population $ \mathfrak{P} $ follows an \textbf{Imitation-of-the-Fittest Protocol} by a \textbf{Susceptible-Infected-Susceptible} (IFPSIS) scheme on the fittest randomly generated node (i.e., the node that partially outputs the largest integer in cycle $1$) \textit{iff} each susceptible node/program obeys protocols defined in \ref{BdefBBcontagionSIS} , \ref{BdefMaxCoopSIS} and \ref{BdefITFOprotSIS} with probability $ \nu $ and each infected node/program to its initial stage at cycle $1$ with probability $ \delta $. Or, more formally, one can describe the algorithm (which runs on $ \mathbf{U'}_R $) for the protocol as: 
	
	Let $\mathbf{X}_{neighbors}(o_j,c)$ be the set of incoming neighbors of node/program $o_j$ that have sent partial outputs to it at the end of the cycle $c$.
	Let $ \{ p_{o_i,c} \mid o_i \in \mathbf{X}_{neighbors}(o_j, c) \land i \in \mathbb{N} \land c \in \mathfrak{C} \} $ be the set of partial outputs relative to $\mathbf{X}_{neighbors}(o_j,c)$.
	Let $ w $ be the network input as defined in \ref{BdefNetworkinput}. Let $\circ$ denote a recursively determined concatenation of finite strings.
	Then, for every $ o_j,o_i \in \mathfrak{ P } $ and $ c,c-1 \in \mathfrak{C} $,
	\begin{enumerate}
		\item if $ max \{ c \mid c \in \mathfrak{C} \} = 1 $, then
		\[
		p_{o_j,c} = \mathbf{U'}_R( o_j \circ w )
		\] \label{Bprotcmax1}
		
		\item if $c=1$ and $ c \neq max \{ c \mid c \in \mathfrak{C} \} $ and $ \mathbf{U'}_R( o_j \circ w ) \neq max \{ \mathbf{U'}_R ( o_i \circ w ) \vert o_i \in \mathfrak{P} \} $, then
		\[
		p_{o_j,c} = S \circ c \circ w \circ o_j \circ \mathbf{U'}_R( o_j \circ w )
		\] \label{Bprotc=1neqmax}
		
		\item if $c=1$ and $ c \neq max \{ c \mid c \in \mathfrak{C} \} $ and $ \mathbf{U'}_R( o_j \circ w ) = max \{ \mathbf{U'}_R ( o_i \circ w ) \vert o_i \in \mathfrak{P} \} $, then
		\[
		p_{o_j,c} = I \circ c \circ w \circ o_j \circ \mathbf{U'}_R( o_j \circ w )
		\] \label{Bprotc=1max}
		
		\item if $ c \neq 1 $ and $ c \neq max \{ c \mid c \in \mathfrak{C} \} $ and $ p_{o_j,c-1} = S \circ c-1 \circ w \circ o_j \circ x $ and 
		\begin{align*} & max \left\{ 
		x \, \middle\vert \begin{array}{l}
		p_{o_j,c-1}= y \circ c-1 \circ w \circ o_j \circ x \, \lor  \\
		\lor \,  y \circ c-1 \circ w \circ o_i \circ x \in  \left\{ p_{o_i,c-1} \middle\vert \begin{array}{l}
		o_i \in \mathbf{X}_{neighbors}(o_j, c-1) \land \\
		\land \, i \in \mathbb{N} \land c-1 \in \mathfrak{C} 
		\end{array} \right\}
		\end{array}
		\right\} = \\
		& = max \{ \mathbf{U'}_R ( o_i \circ w ) \vert o_i \in \mathfrak{P} \} 
		\end{align*}
		 \noindent then 
		\begin{align*}
		p_{o_j,c} = & I \circ c \circ w \circ o_j \circ \\
		& \circ max \left\{ 
		x \, \middle\vert \begin{array}{l}
		p_{o_j,c-1}= y \circ c-1 \circ w \circ o_j \circ x \, \lor  \\
		\lor \, y \circ c-1 \circ w \circ o_i \circ x \in  \left\{ p_{o_i,c-1} \middle\vert \begin{array}{l}
		o_i \in \mathbf{X}_{neighbors}(o_j, c-1) \land \\
		\land \, i \in \mathbb{N} \land c-1 \in \mathfrak{C} 
		\end{array} \right\}
		\end{array}
		\right\}
		\end{align*}
		\noindent with probability $ \nu $. \\ \label{BprotSusceptiblenumax}
		
		\item if $ c \neq 1 $ and $ c \neq max \{ c \mid c \in \mathfrak{C} \} $ and $ p_{o_j,c-1} = S \circ c-1 \circ w \circ o_j \circ x $ and
		\begin{align*} & max \left\{ 
		x \, \middle\vert \begin{array}{l}
		p_{o_j,c-1}= y \circ c-1 \circ w \circ o_j \circ x \, \lor  \\
		\lor \, y \circ c-1 \circ w \circ o_i \circ x \in  \left\{ p_{o_i,c-1} \middle\vert \begin{array}{l}
		o_i \in \mathbf{X}_{neighbors}(o_j, c-1) \land \\
		\land \, i \in \mathbb{N} \land c-1 \in \mathfrak{C} 
		\end{array} \right\}
		\end{array}
		\right\} \neq \\
		& \neq max \{ \mathbf{U'}_R ( o_i \circ w ) \vert o_i \in \mathfrak{P} \} 
		\end{align*}
		\noindent then 
		\begin{align*}
		p_{o_j,c} & = p_{ o_j, c-1 }
		\end{align*} \\ \label{BprotSusceptiblenuneqmax}

		 \item if $ c \neq 1 $ and $ c \neq max \{ c \mid c \in \mathfrak{C} \} $ and $ p_{o_j,c-1} = S \circ c-1 \circ w \circ o_j \circ x $ and
		 \begin{align*} & max \left\{ 
		 	x \, \middle\vert \begin{array}{l}
		 		p_{o_j,c-1}= y \circ c-1 \circ w \circ o_j \circ x \, \lor  \\
		 		\lor \,  y \circ c-1 \circ w \circ o_i \circ x \in  \left\{ p_{o_i,c-1} \middle\vert \begin{array}{l}
		 			o_i \in \mathbf{X}_{neighbors}(o_j, c-1) \land \\
		 			\land \, i \in \mathbb{N} \land c-1 \in \mathfrak{C} 
		 		\end{array} \right\}
		 	\end{array}
		 	\right\} = \\
		 	& = max \{ \mathbf{U'}_R ( o_i \circ w ) \vert o_i \in \mathfrak{P} \} 
		 \end{align*}
		 \noindent then
		\begin{align*}
		p_{o_j,c} & = p_{ o_j, c-1 }
		\end{align*}
		\noindent with probability $ 1 - \nu $. \\ \label{BprotSusceptible1-nu}

		\item if $ c \neq 1 $ and $ c \neq max \{ c \mid c \in \mathfrak{C} \} $ and $ p_{o_j,c-1} = I \circ c-1 \circ w \circ o_j \circ x $, then 
		\begin{align*}
		p_{o_j,c} & = p_{ o_j, 1 }
		\end{align*}
		\noindent with probability $ \delta $. \\ \label{BprotInfecteddelta}
		
		\item if $ c \neq 1 $ and $ c \neq max \{ c \mid c \in \mathfrak{C} \} $ and $ p_{o_j,c-1} = I \circ c-1 \circ w \circ o_j \circ x $, then 
		\begin{align*}
		p_{o_j,c} & = p_{ o_j, c-1 }
		\end{align*}
		\noindent with probability $ 1 - \delta $. \\ \label{BprotInfected1-delta}

		\item if $ c = max \{ c \mid c \in \mathfrak{C} \} $ and $ p_{o_j,c-1} = w \circ o_i  \circ o_j \circ x $, then
		\begin{center}
			\noindent $ p_{o_j,c} = x $ 
		\end{center} \label{BprotLastcycle}
	\end{enumerate}
\end{Bdefinition}

\begin{Bremarknote}[SIS]
	Since we will be working with synchronous algorithmic networks, these global sharing protocols applies at the end of each cycle (or communication round) --- see Definition \ref{BdefCycle}. So, after the first cycle the diffusion of the biggest partial output works like a \textbf{spreading} in time-varying networks \cite{Guimaraes2013} \cite{Costa2015a}. However, under the SIS contagion scheme as in~\cite{Pastor-Satorras2001,Pastor-Satorras2001a,Pastor-Satorras2002}. And the last cycle (or more cycles --- see Definition \ref{BdefN'_BB}) is spent in order to make each node/program return a number --- from which we measure the complexity of the respective node/program as discussed in~\cite{Abrahao2017}.	
\end{Bremarknote}

\begin{Bremarknote}
	In order to simplify our notation we let $ w \circ \mathbf{U'}_R(x ) $ denote the prefix preserving concatenation $\circ$ (see Notation \ref{BdefConcatenation}) of the string $ w \in \mathbf{L_U} $ with the string $ y \in \mathbf{L_U} $ such that $y$ represents the number $\mathbf{U'}_R(x)$ in the language $\mathbf{L_U}$.
\end{Bremarknote}

\begin{Bremarknote}[SIS]\label{BdefProt'}
	Note that there is a program that runs on $ \mathbf{U'}_R $ rules \ref{BprotSusceptiblenumax}, \ref{BprotSusceptiblenuneqmax} and  \ref{BprotSusceptible1-nu} which we can denote by $ P_{prot_\nu} $. The same for: rules \ref{BprotInfected1-delta} and \ref{BprotInfecteddelta}, denoted by $ P_{prot_\delta} $; rules \ref{Bprotc=1max}, \ref{Bprotc=1neqmax} and \ref{Bprotcmax1}, denoted by $ P_{prot_1} $; rule \ref{BprotLastcycle} denoted by $ P_{prot_f} $; Thus, making Definition~\ref{BdefU'_R} sound.
	
	\begin{Bsubremarknote2}[SIS]
		In fact, the rules \ref{BprotSusceptible1-nu} and \ref{BprotSusceptiblenuneqmax} in Definition~\ref{BdefIFPSIS} may instead allow that the susceptible node still performs the IFP like in~\cite{Abrahao2017} while remaining Susceptible (i.e., of the form $ p_{ o_j , c }  = S \circ w  $). However, the reader is invited to note that it will only makes the expected emergent algorithmic complexity larger in our final results, so that it does not change our final conclusions.
	\end{Bsubremarknote2}

	\begin{Bsubremarknote2}[SIS]
		The reader is also invited to note that $ P_{prot'_R} $ may access and/or decide the value of the biggest output in cycle $1$ within the randomly generate population, for example, from two procedures: 
		\begin{enumerate}
			\item It calculates the biggest value of the biggest output of a program $ p \in \mathbf{L_U} $ running on $ \mathbf{U} $ where $ | p |  \leq \log(N) - \mathbf{O}(1) $. From Lemma~\ref{lemmaSLLNandAITSIS}, we know that a node that calculates such value is expected to occur for large enough populations. Thus, this procedure can only be performed by a hypercomputer or an oracle Turing machine, e.g., machine $ \mathbf{U'}_R $; \\

			\item As presented in~\cite{Abrahao2017}, it records its initial randomly generated program and performs the IFP for a number of cycles at least as big as the diameter of the network before starting the IFPSIS. Also note that the diameter of scale-free networks in \cite{Pastor-Satorras2001,Pastor-Satorras2001a,Pastor-Satorras2002} is expected to be dominated by $ \mathbf{O}( \log(N) ) $ (see \cite{Bollobas2004}). Thus, the reader is invited to note that our final results in Corollaries~\ref{corSIS} and~\ref{BcorBA} also hold for spending $ \mathbf{O}( \log(N) ) $ more cycles in order to run protocol IFP before starting to run IFPSIS. Also note that, in order to run to the IFP in first place, one needs to know the partial output of each node/program at the end of cycle $1$ in first place. Thus, this procedure can only be performed by a hypercomputer or an oracle Turing machine, e.g., machine $ \mathbf{U'}_R $; \\
		\end{enumerate}
	\end{Bsubremarknote2}

	\begin{Bsubremarknote2}[SIS] 
		Note that, if one enforces that the number of cycles needs to be informed to the this global information-sharing protocol, then the expected algorithmic complexity of the networked population will be even larger in respect to the expected algorithmic complexity of the isolated population in Corollary~\ref{corMainSIS}. That is, the additional input of the number of cycles in the networked case will ``cancel'' the one in the isolated case. Thus, our final results on the lower bound for the expected emergent algorithmic complexity of a node (EEAC) can even be increased. This is expected to happen for example in the case it was possible to simulate the entire algorithmic network on $ \mathbf{U'}_R $. Thus, it will be important for migrating our results to resource-bounded algorithmic networks as suggested in~\cite{Abrahao2017} for future research. 
	\end{Bsubremarknote2}
\end{Bremarknote}

\begin{Bsubdefinition}\label{BdefBBcontagionSIS}
	We call a \textbf{Busy Beaver contagion protocol} as a global information-sharing protocol in which every node/program runs the node/program of the neighbor that have output --- a partial output --- the largest integer instead of its own program \textit{iff} the partial output of this neighbor is bigger than the receiver's own partial output.

\end{Bsubdefinition}

\begin{Bsubremarknote}
	Note that a node/program only needs to take into account the biggest partial output that any of its neighbors have sent. If more than one sends the largest integer as partial output, the receiver node/program choose one of these respective neighbors accordingly to an arbitrary rule. Then, this partial output is the one that will be compared to the partial output from the receiver node/program.
\end{Bsubremarknote}

\begin{Bsubdefinition}\label{BdefMaxCoopSIS}
	In a \textbf{maximally cooperative protocol} every node/program shares its own program and its latest partial output with all its neighbors at the end of each cycle and before the next cycle begins.
	
	\begin{Bsubremarknote}
		In the model defined in \ref{BdefP'_BB} sharing only the last partial output turns out to be equivalent\footnote{ However, at the expense of using more computation time.} to the definition of maximal cooperation.	
	\end{Bsubremarknote} 
	
	\begin{Bsubremarknote}
		Remember that $ \mathfrak{N'}_{BB} (N, f, t, j)  $ only lets its nodes/programs perform computation in the first cycle (see Definition \ref{BdefITFOprotSIS}). This is the reason why only the network input $w$ matters in its respective maximally cooperative protocol.
	\end{Bsubremarknote}
	
	\begin{Bsubremarknote}
		As in \cite{Costa2015a} this diffusion process may be interpreted as following a Breadth-First Search (BFS), in which each node starts a diffusion by sending the specified information in Definition \ref{BdefIFPSIS} to all of its adjacent
		nodes. Then, these adjacent nodes relay information for their own adjacent nodes in the next time instant, and so on.
	\end{Bsubremarknote}
	
\end{Bsubdefinition}

\begin{Bsubdefinition}\label{BdefITFOprotSIS}
	We call a \textbf{contagion-only protocol} as a global information-sharing protocol in which every node/program only plays the Busy Beaver contagion and does not perform any other computation after the first cycle \textbf{when networked} in some $ \mathfrak{N} = (G, \mathfrak{P}, b)$.

	\begin{Bsubremarknote}
		This is the condition that allows us to investigate the ``worst'' (see discussion in~\cite{Abrahao2017}) case in which no node/program spends computational resources other than playing its global sharing protocols.\footnote{ In our case, generating expected emergent algorithmic complexity of a node.} In other words, it forces the algorithmic network to rely on a diffusion process only.
	\end{Bsubremarknote} 
	
	\begin{Bsubremarknote}
		In the main model presented in this article the first time instant occurs after the first cycle --- see \ref{BdefN'_BB}. 
	\end{Bsubremarknote}
\end{Bsubdefinition}

\noindent \\



\begin{Bdefinition}[SIS]\label{BdefL'_BB}
	Let $ {\mathbf{L'}}_{BB} \subset \mathrm{\textbf{L}}_{\mathrm{\textbf{U}}} $  be a language of programs of the form $P_{prot'_R} \circ p$ where $p \in \mathbf{L}_{\textbf{U}}$  and the prefix $ P_{prot'_R} \in \mathbf{L}_{\textbf{U}} $ defined in~\ref{BdefU'_R} is any program that always ensures that $P_{prot'_R} \circ p$ obeys when running on $\textbf{U'}_R$ the \textbf{IFPSIS} protocol when it is \textbf{networked}. Otherwise, if the node/program $ P_{prot'_R} \circ p $ is \textbf{isolated}, then $\textbf{U'}_R(P_{prot'_R} \circ p) = \textbf{U'}_R(p)$ and every subsequent cycle works like a reiteration of partial outputs as immediate next input for the same node/program.\footnote{ Thus, the isolated case may be represented (and is equivalent to) by same algorithmic network built on a population in language $ \mathbf{L_U} $ that does not follow any information-sharing protocol and the topology of the MultiAspect Graph is composed by one-step loops on each node/program only.} 
\end{Bdefinition}

\begin{Bdefinition}[SIS]\label{BdefP'_BB}
	Let $\mathfrak{P'}_{BB}(N) \subseteq \mathrm{\textbf{L'}}_{BB}$ be a population of $N$ elements that is
	\textbf{synchronous}\footnote{ The procedure responsible for performing the synchronization may be abstract-hypothetical or defined on an underlying oracle machine that makes each individual cycle start at the same time (or after every node/program returns its partial output in the respective cycle).}, \textbf{sensitive to oracles} and with \textbf{randomly generated}\footnote{ Note that we are dealing with self-delimiting languages, so that one can always define algorithmic probabilities or an optimal prefix-free language \cite{Cover2005,Li1997,Chaitin2004}. } suffixes $ p \in \mathbf{L}_{ \mathfrak{P'}_{BB}(N) } \subset \mathbf{L_U} $ such that\footnote{ Hence, $ \mathbf{L}_{ \mathfrak{P'}_{BB}(N) } $ is defined as the population of suffix nodes/programs that were randomly generated in order to constitute $ \mathfrak{P'}_{BB}(N) $. Thus, $ \mathbf{L}_{ \mathfrak{P'}_{BB}(N) } $ is a population and not a language (see Definition~\ref{BdefPopulation} and Note~\ref{notePopulationnotlanguageSIS}), so that there may be repetitions within $ \mathbf{L}_{ \mathfrak{P'}_{BB}(N) } $. It is important to note this since the letter $L$ is used to denote languages in other parts in this paper.} 
	
	\[
	p \in \mathbf{L}_{ \mathfrak{P'}_{BB}(N) }
	\]
	\[
	iff 
	\]
	\[
	P_{prot'_R} \circ p \in \mathfrak{P'}_{BB}(N) \subseteq \mathrm{\textbf{L}}_{BB}
	\]
	
	\noindent where $ \mathbf{L}_{ \mathfrak{P'}_{BB}(N) } $ is a population of suffixes $p$.
	\begin{Bremarknote}
		Note that all conditions and protocols in Definition~\ref{BdefP'_BB} define the set of properties $Pr(\mathfrak{P'}_{BB}(N))$ of the population $\mathfrak{P'}_{BB}(N)$ as in Definition~\ref{BdefFunctionbinAN}.
	\end{Bremarknote}
	
	\begin{Bremarknote}\label{notePopulationnotlanguageSIS}
		There is a misplaced usage of the operator $ \subseteq $ in $ \mathfrak{P'}_{BB}(N) \subseteq \mathrm{\textbf{L'}}_{BB} $ here. Since $ \mathfrak{P'}_{BB}(N) $ is a population, it may contain repeated elements of $ \textbf{L'}_{BB} $. However, for the sake of simplicity, we say a population $ \mathfrak{ P } $ is contained in a language $ \mathbf{L} $ \textit{iff} \[  \forall p_i, 1 \leq i \leq | \mathfrak{ P } | \, \big( \, p_i \in \mathfrak{ P }  \implies p_i \in \mathbf{L} \, \big) \] 
	\end{Bremarknote}
	
\end{Bdefinition}

\begin{Bdefinition}\label{BdefL_USIS}
	For the sake of simplifying our notation, we denote the language\footnote{ Note that, since it is a language and not a population, no repetitions are allowed in $\mathbf{L_U}(N) $. } of the size-ordered\footnote{ That is, an ordering from the smallest to the largest size. Also note that ordering members with the same size may follow a chosen arbitrary rule.} smallest 
	\[ p \in \lim_{ N \to \infty } \mathbf{L}_{ \mathfrak{P'}_{BB}(N) }
	\]
	\noindent as
	\[ 
	\mathbf{L_U}(N) 
	\]
\end{Bdefinition}	

\subsection{Definitions on the algorithmic network model}



\begin{Bdefinition}[SIS]\label{BdefN'_BB}
	Let
	\[
	\mathfrak{N'}_{BB} (N, f, t,   j)=(G_t, \mathfrak{P'}_{BB} (N),b_j)
	\]
	\noindent be an algorithmic network where $f$ is an arbitrary well-defined function such that 
	\[
	\myfunc{f}{ \mathbb{N^*} \times X \subseteq \mathrm{T}(G_t) } { \mathbb{N} } { ( x,t ) } { y \geq x }
	\]
	\noindent and $ G_t \in \mathbb{G}_{SIS}( f, t )$, $ | \mathrm{V}(G_t) | = N $, $ | \mathrm{T}(G_t) | > 0 $ and there are arbitrarily chosen\footnote{ Since they are arbitrarily chosen, one can take them as minimum as possible in order to minimize the number of cycles for example. That is, $ c_0 = 0 $ and $ n = | \mathrm{T}(G_t) | + 1 $ for example. } $ c_0, n \in \mathbb{N} $ where $ c_0 + | \mathrm{T}(G_t) | + 1 \leq n \in \mathbb{N} $ such that $b_j$ is an injective function
	\[
	\myfunc{b_j} { \mathrm{V}(G_t) \times \mathrm{T}(G_t) } {\mathfrak{P'}_{BB}(N) \times \mathbb{N}|_1^{ n }} { (v,t_{c-1}) } { b_j(v,t_{c-1})=( o, c_0 + c ) }  
	\]
	\noindent such that\footnote{ See Definition \ref{BdefFunctionbinAN}.}, since one has fixed the values of $c_0$ and $n$,
	\[
	\left| \left\{ 
	b_j \, \Bigg| \, 
	\myfunc{b_j} { \mathrm{V}(G_t) \times \mathrm{T}(G_t) } {\mathfrak{P}_{BB}(N) \times \mathbb{N}|_1^{ n }} { (v,t_{c-1}) } { b_j(v,t_{c-1})=( o, c_0 + c ) } 
	\right\} \right| \leq N^{ N }
	\]

	\begin{Bremarknote}[SIS]
		In summary, $ \mathfrak{N'}_{BB} ( N, f, t, j ) $ is an algorithmic network populated by $N$ nodes/programs from $ \mathfrak{P'}_{BB} (N) $ such that, after the first (or $c_0$ cycles) cycle, it starts a diffusion\footnote{ There may be other diffusions too. However, only the one from the biggest partial output is independent of neighbor's partial outputs.} process of the biggest partial output (given at the end of the first cycle) determined by network $ G_t $ that belongs to a family of graphs $ \mathbb{G}_{SIS}( f, t ) $ as defined in~\ref{BdefFamilyGSIS}. Then, at the last time instant diffusion stops and one cycle (or more) is spent\footnote{ This condition is necessary to make this algorithmic network defined even when $|\mathrm{T}(G_t)|=1$.} in order to make each node return a final output. Each node returns as final output its previous partial output determined at the last time instant --- see Definition \ref{BdefITFOprotSIS}. 
	\end{Bremarknote}
	
	\begin{Bremarknote}
		Note that in this model the aspects of the graphs in the family $\mathbb{G}_{SIS}$ that are mapped into the properties of the population $ \mathfrak{P'}_{BB}(N) $ by functions $b_j$ are nodes and time instants. 
	\end{Bremarknote}
	
	\begin{Bremarknote}\label{BnoteRestrictedfamilyGinN_BBSIS}
		The reader is invited to note that the main results presented in this paper also hold for only one function $b_j$ per graph in the family $ \mathbb{G}_{SIS} $ (see Note~\ref{BnoteunderBthmSinglegraphinfamilyGSIS}).
	\end{Bremarknote}
	
	\begin{Bsubdefinition}\label{BdefC_BBSIS}
		We also denote $ \mathbb{N}|_1^{ n } $ as $ \mathfrak{C_{BB}} $. 
	\end{Bsubdefinition}

\end{Bdefinition}

\noindent \\

\subsection{Definitions on emergent algorithmic information}

\begin{Bdefinition}\label{BdefEAC}
	The \textbf{emergent algorithmic complexity (EAC)} of a node/program $o_i$ in $c$ cycles is given in an algorithmic network that always produces partial and final outputs by

	\[
	{\displaystyle{\myDelta_{iso}^{net(b)}}} A (o_i,c) = 
	A (\mathrm{\textbf{U}}(p_{net}^{b} ( o_i , c )) - 
	A (\mathrm{\textbf{U}}(p_{iso} ( o_i , c) )
	\]
	
	\noindent where:
	
	\begin{enumerate}
		\item $ o_i \in \mathbf{L} $;
		
		\item $ p_{net}^{b} $ is the program that computes cycle-per-cycle the partial outputs of $o_i$ when networked assuming the position $ v $, where $ b(v,\mathbf{\bar{x}}) = (o_i, b(\mathbf{\bar{x}})) $, in the graph $G$ in the specified number of cycles $c$ with network input $w$; 
		
		\item $ p_{iso} $ is the program that computes cycle-per-cycle the partial outputs of $o_i$ when isolated in the specified number of cycles $ c $ with network input $w$;
	\end{enumerate}
	
	\begin{Bremarknote}
		Note that:
		\begin{enumerate}
			\item $ A (\mathrm{\textbf{U}}(p_{net}^{b} ( o_i , c) )) $ is the algorithmic complexity of what the node/program $o_i$ does \textbf{when networked};
			
			\item $ A (\mathrm{\textbf{U}}(p_{iso} ( o_i , c) ) $ is the algorithmic complexity of what the node/program $o_i$ does \textbf{when isolated};
		\end{enumerate}
	\end{Bremarknote}

	\begin{Bremarknote}
		While program $p_{iso}$ may be very simple, since it is basically a program that reiterates partial outputs of $o_i$ as inputs to itself at the beginning of the next cycle (up to $c$ times), program  $p_{net}^{b}$ may also comprise giving the sent partial outputs from $o_i$'s incoming neighbors at the end of each cycle as inputs to $o_i$ at the beginning of the respective next cycle, so that it may be only described by a much more complex procedure.  
	\end{Bremarknote}
	
	\begin{Bremarknote}
		Note that the algorithmic complexity of $p_{iso}$ or $p_{net}^{b}$ may be not directly linked to $ A (\mathrm{\textbf{U}}(p_{iso} ( o_i , c) ) $ or $ A (\mathrm{\textbf{U}}(p_{net}^{b} ( o_i , c) )) $ respectively, since the $ A (\mathrm{\textbf{U}}(p_{iso} ( o_i , c) ) $ and $ A (\mathrm{\textbf{U}}(p_{net}^{b} ( o_i , c) )) $ are related to the final outputs (if any) of each node's computation. 
	\end{Bremarknote}
	
	\begin{Bremarknote}
		Remember definitions \ref{Bdefsensitivetooracles} and \ref{BdefU'_R} which states that even when a node/program does not halt in some cycle, machine $ \mathbf{U'}_R $ was defined in order to assure that there is always a partial output for every node/program for every cycle. If population $ \mathfrak{ P } $ is defined in a way that eventually a partial or final output is not obtained when running on the respective theoretical machine (see Definition \ref{BdefPopulation}), Definition \ref{BdefEAC} would be inconsistent. This is the reason we stated in its formulation that there always is partial and final outputs. However, it is not necessary in the case of $ \mathfrak{N'}_{BB} (N, f, t, j) $ (see Definition \ref{BdefP'_BB} and \ref{BdefEAConN'_BB}).
	\end{Bremarknote}
	\begin{Bremarknote}
		If one defines the \textbf{emergent creativity} of a node/program as
		\[
		A ( \textbf{U}(p_{net}^{b} ( o_i , c )) | \textbf{U}(p_{iso} ( o_i , c) )
		\]
		our results also hold for replacing the expected emergent algorithmic complexity (EEAC) with expected emergent algorithmic creativity (EEACr). Since we are estimating lower bounds, note that from AIT we have that
		\[
		A ( \textbf{U}(p_{net}^{b} ( o_i , c )) | \textbf{U}(p_{iso} ( o_i , c) )
		\geq
		A (\mathrm{\textbf{U}}(p_{net}^{b} ( o_i , c )) - 
		A (\mathrm{\textbf{U}}(p_{iso} ( o_i , c) )
		+
		\mathbf{O}\big( 1 \big)
		\] \\
	\end{Bremarknote}

	
	
	\begin{Bsubdefinition}[SIS] \label{BdefEAConN'_BB}
		More specifically, one can denote the \textbf{emergent algorithmic complexity} of a node/program $o_i$ \textbf{in an algorithmic network} $ \mathfrak{N'}_{BB} (N, f, t , j) $ ($ \mathbf{ EAC_{BBSIS} } $) during $c$ cycles as

		\[
		{\displaystyle{\myDelta_{iso}^{net(b_j)}}} A (o_i,c) = 
		A (\mathrm{\textbf{U}}(p_{net}^{b_j} ( o_i ,  c ) )) - 
		A (\mathrm{\textbf{U}}(p_{iso} ( p_i ,  c )))
		\]
		
		\noindent where:
		
		\begin{enumerate}
			\item $ o_i = P_{prot'_R} \circ p_i \in \mathfrak{P'}_{BB}(N) \subseteq \mathrm{\textbf{L}}_{BB} $ ;
			
			\item $ p_{net}^{b_j} $ is the program that computes cycle-per-cycle what a program $o_i$ does \textbf{when networked} assuming the position $ v $, where $ b_j(v)=(o_i) $, in the graph $G_t$ in $c$ cycles with network input $w$; 
			
			\item $ p_{iso} $ is the program that computes cycle-per-cycle what a program $p_i$ does \textbf{when isolated}\footnote{ Which is a redundancy since we are refering to $p_i$ instead of $o_i$ here.} in $c$ cycles with network input $w$; 
		\end{enumerate}

		\begin{Bsubremarknote}
			In order to check if Definition \ref{BdefEAConN'_BB} is well-defined, remember Definitions \ref{BdefU'_R} and \ref{BdefL'_BB}. 
		\end{Bsubremarknote}
		
	\end{Bsubdefinition}
	
\end{Bdefinition}

\begin{Bdefinition}[Notation]
	For the sake of simplifying our notation, let $ \{ b_j \} $ denote 
	\[ 
	\left\{ b_j \, \Bigg| \, \myfunc{b_j} { \mathrm{V}(G_t) \times \mathrm{T}(G_t) } {\mathfrak{P'}_{BB}(N) \times \mathbb{N}|_1^{ n }} { (v,t_{c-1}) } { b_j(v,t_{c-1})=( o, x + c ) } \right\}
	\]
	and $ b_j $ in the sum $ \sum\limits_{ b_j } $ denote 
	\[ 
	\myfunc{b_j} { \mathrm{V}(G_t) \times \mathrm{T}(G_t) } {\mathfrak{P'}_{BB}(N) \times \mathbb{N}|_1^{ n }} { (v,t_{c-1}) } { b_j(v,t_{c-1})=( o, x + c ) }
	\] \\
\end{Bdefinition}

\begin{Bdefinition}\label{BdefAEAC}
	
	We denote the \textbf{average emergent algorithmic complexity of a node/program (AEAC)} for an algorithmic network $ \mathfrak{N} = (G, \mathfrak{ P }, b) $, as 
	\begin{align*}
	&\mathbf{E}_{ \mathfrak{N} } 
	\left(
	{ {\displaystyle{\myDelta_{iso}^{net}} A} (o_i,c)} 
	\right)
	= \\
	&=
	{\tiny 
		\sum\limits_{ b }
	}
	\frac{
		\frac{
			\sum\limits_{ o_i \in \mathfrak{P} } { {\displaystyle{\myDelta_{iso}^{net(b)}} A} (o_i,c)}
		}
		{N} 
	}
	{ | \{ b \} | } 
	\end{align*}
	\begin{Bremarknote}\label{BnoteRestrictedfamilyGinAEAC}
		As in note \ref{BnoteRestrictedfamilyGinN_BBSIS} , if only one function $b$ exists per population, then there is only one possible network's topology linking each node/program in the population. So, in this case, 
		\[
		\mathbf{E}_{ \mathfrak{N} } 
		\left(
		{ {\displaystyle{\myDelta_{iso}^{net}} A} (o_i,c)} 
		\right)
		=
		\frac{
			\sum\limits_{ o_i \in \mathfrak{P} } { {\displaystyle{\myDelta_{iso}^{net(b)}} A} (o_i,c)}
		}
		{ N } 
		\] \\

	\end{Bremarknote}
	
\end{Bdefinition}


\begin{Bdefinition}[SIS]\label{BdefEEACN'_BB}
	
	We denote the \textbf{expected emergent algorithmic complexity of a node/program} for algorithmic networks $ \mathfrak{N'}_{BB}(N,f,t ,j) $ ($ \mathbf{ EEAC_{BBSIS} } $) with network input $w$, where $0<j \leq |\{ b_j \}|$, as 
	
	\begin{align*}
	&\mathbf{E}_{ \mathfrak{N'}_{BB}(N,f,t ) } 
	\left(
	{ {\displaystyle{\myDelta_{iso}^{net}} A} (o_i,c)} 
	\right)
	= \\
	&=
	{\tiny 
		\sum\limits_{ b_j  }
	}
	\frac{
		\frac{
			\sum\limits_{ o_i \in \mathfrak{P'}_{BB}(N) } { {\displaystyle{\myDelta_{iso}^{net(b_j)}} A} (o_i,c)}
		}
		{N} 
	}
	{ | \{ b_j \} | }
	= \\ 
	&= 
	\frac{
		{
			\sum\limits_{ b_j  }
		}
		\frac{
			\sum\limits_{ o_i \in \mathfrak{P'}_{BB}(N) } { A (\mathrm{\textbf{U}}(p_{net}^{b_j} ( o_i ,  c ) )) - 
				A (\mathrm{\textbf{U}}(p_{iso} ( p_i ,  c ))) }
		}
		{N}
	}
	{ | \{ b_j \} | } 
	\end{align*}
	
	\begin{Bremarknote}[SIS]\label{BnoteRestrictedfamilyGinEEAC_BBSIS}
		As in note \ref{BnoteRestrictedfamilyGinAEAC} , if only one function $b_j$ exists per population, then there is only one possible network's topology linking each node/program in the population. So, in this case, 
		\[
		\mathbf{E}_{ \mathfrak{N'}_{BB}(N,f,t ) } 
		\left(
		{ {\displaystyle{\myDelta_{iso}^{net}} A} (o_i,c)} 
		\right)
		=
		\frac{
			\sum\limits_{ o_i \in \mathfrak{P'}_{BB}(N) } { {\displaystyle{\myDelta_{iso}^{net(b_j)}} A} (o_i,c)}
		}
		{N} 
		\] \\

	\end{Bremarknote}
	
\end{Bdefinition}

\begin{Bdefinition}\label{BdefAEOE}
	We say an algorithmic network $ \mathfrak{N} $ with a population of $N$ nodes has the property of \textbf{average emergent open-endedness (AEOE)} for a given network input $w$ in $c$ cycles \textit{iff} 
	\[
	\lim_{ N \to \infty } \mathbf{E}_{ \mathfrak{N} } 
	\left(
	{ {\displaystyle{\myDelta_{iso}^{net}} A} (o_i,c)} 
	\right)
	=
	\infty
	\]  
\end{Bdefinition}


\begin{Bsubdefinition}[SIS]\label{BdefEEOESIS}
	We say an algorithmic network $ \mathfrak{N'}_{BB}(N,f,t ) $ has the property of \textbf{expected emergent open-endedness (EEOE)} for a given network input $w$ in $c$ cycles \textit{iff} 
	\[
	\lim_{ N \to \infty } \mathbf{E}_{ \mathfrak{N'}_{BB}(N,f,t ) } 
	\left(
	{ {\displaystyle{\myDelta_{iso}^{net}} A} (o_i,c)} 
	\right)
	=
	\infty
	\] 
\end{Bsubdefinition}

\noindent \\

\subsection{Definitions on cycle-bounded halting probability}


\begin{Bdefinition}[SIS]\label{BdefHaltSIS}
	Let $ \mathfrak{N'}_{BB} (N, f, t , j) $ be an algorithmic network. We denote the population of $ p_i \in \mathfrak{P} $, where $ 1 \leq i \leq N $ and $ P_{prot'_R} \circ p_i \in \mathfrak{P'}_{BB} (N) $, such that $p_i$ always halts on network input $w$ in \textbf{every} cycle until $c$ when $ P_{prot'_R} \circ p_i $ is \textbf{isolated} as
	\[
	Halt_{iso} ( \mathfrak{P}, w, c )
	\]
	
	\begin{Bsubdefinition}[SIS]\label{BdefBarHaltSIS}
		Analogously, we denote the population of $ p_i \in \mathfrak{P} $, where $ 1 \leq i \leq N $ and $ P_{prot'_R} \circ p_i \in \mathfrak{P'}_{BB} (N) $ with $ \mathfrak{N'}_{BB} ( N, f, t,  j ) = (G_t, \mathfrak{P'}_{BB} (N),b_j) $, such that $p_i$ does \textbf{not} halt on network input $w$ in at least one cycle until $c$ when $ P_{prot'_R} \circ p_i $ is isolated as
		
		\[
		\overline{ Halt }_{iso}( \mathfrak{P}, w, c ) 
		\] 
	\end{Bsubdefinition}
	
	\begin{Bremarknote}[SIS]
		Note that both Definitions~\ref{BdefHaltSIS} and~\ref{BdefBarHaltSIS} are well-defined for an arbitrary language $L$ instead of a population $ \mathfrak{P} $, so that one can denote it as $ Halt_{iso}( L, w, c ) $ and $ \overline{ Halt }_{iso}( L, w, c ) $ respectively.
	\end{Bremarknote}
\end{Bdefinition}

\begin{Bdefinition}[SIS]\label{BdefOmegawcSIS}
	We denote the \textbf{cycle-bounded conditional halting probability}\footnote{ That is, a conditional Chaitin's Omega number for isolated programs in a population of an algorithmic network.} of a program in a language $\mathbf{L_U}$ that always halts for an initial input $w$ in $c$ cycles as 
	\[
	\Omega( w,c )
	=
	\sum\limits_{ p_i \in Halt_{iso}( \mathbf{L_U}, w, c )  } 
	{ \frac {1} { 2^{ |p_i| } } }
	=
	\lim\limits_{ N \to \infty }
	\sum\limits_{ p_i \in Halt_{iso}( \mathbf{L_U}(N), w, c )  } 
	{ \frac {1} { 2^{ |p_i| } } }
	\]
	\begin{Bremarknote}[SIS]
		Since $\mathbf{L_U}$ is self-delimiting, the algorithmic probability of each program is well-defined. Hence, one can define the halting probability $\Omega$ \footnote{ Note that the Greek letter $\Omega$ here does not stand for an asymptotic notation opposed to the big O notation. } for $\mathbf{L_U}$. Further, the same holds for conditional halting probability $\Omega(w)$, i.e. the probability that a program halts when $w$ is given as input. Then, the set of programs that always halt on initial input $w$ in $c$ cycles is a proper subset of the set of programs that halt, so that 
		\[
		\Omega(w,c) \leq \Omega(w) = \Omega(w,1) < 1
		\]
		\noindent In fact, one can prove that
		\[
		\Omega(w,c') \leq \Omega(w,c)
		\]
		\noindent when $ c' \geq c > 0 $. On the other hand, one can also build programs that always halt for every input and for every number of cycles. So, for every $ c > 0 $
		\[
		0 < \Omega(w,c) < 1
		\] 
	\end{Bremarknote}

\end{Bdefinition}

\subsection{Definitions on the prevalence of ``infected'' nodes}\label{appendixsubsectionPrevalence}

\begin{Bdefinition}
	Let $ A_{max}( \mathfrak{N} , c ) $ denote the algorithmic complexity of the biggest final output returned by a member of the population $\mathfrak{P}$ in a maximum number of $c$ cycles, where $\mathfrak{N}=(G, \mathfrak{P},b)$.

	\begin{Bsubdefinition}\label{BdefA'_max}
		In the case of $ \mathfrak{N'}_{BB}(N,f,t , j) $ and $c=1$, for the sake of simplifying our notation, we will just denote it as $ A'_{max} $. 
	\end{Bsubdefinition}
	
\end{Bdefinition}


	\begin{Bdefinition}[SIS]\label{BdefPrevalence}
		Let an algorithmic network $ \mathfrak{N}=(G,\mathfrak{P},b) $ and a set $ \mathbf{X}_{ partial }(  \mathfrak{N} )|_{t}^{t'} $ of partial outputs of its nodes/programs $ o_i \in \mathfrak{P} $, which appear during time instants $ t $ until $ t' $, be well-defined. 
		Let $ t_i, t, t' \in \mathrm{T}(G_t) $ with $ t_i \leq t \leq t' $.
		We denote the fraction of nodes/programs in $ \mathfrak{P} $ with partial outputs in $ \mathbf{X}_{ partial }(  \mathfrak{N} )|_{t}^{t'} $ that have at time instant $t'$ a ``better'' (or equal to) partial output than any node/program's partial output at time instant $ t_i $ as
		\[ 
		{ \tau_{\rho}( \mathfrak{N}, t_i ) }|_{t}^{t'} 
		=
		\frac{ \left| \mathbf{X}_{ { \tau_{\rho}( \mathfrak{N}, t_i ) }|_{t}^{t'} } \right| }{ \left| \mathrm{ V(G) } \right| }
		\] 
		
		\noindent and we call it as \textbf{density of ``infected'' nodes/programs in a time interval} or \textbf{persistence\footnote{ This second term is usually seen in complex networks theory in the case which the density tends to spreads to the entire network or to remain stationary. For our purposes, we will study the stationary case. See \cite{Pastor-Satorras2001,Pastor-Satorras2001a,Pastor-Satorras2002}. } of the ``infection'' in a time interval}.
		
		\begin{Bremarknote}\label{BbestpartialoutputSIS}
			The notion of what is the ``best'' partial output may vary on how the algorithmic network $\mathfrak{N}$ is defined --- in some algorithmic networks the notion of ``best'' partial output may even be not defined. We consider the ``best'' partial output as being the one that always affects the neighbors to which it is shared by making them to return a final result that is at least as ``good'' as the one that the node with the ``best'' partial output --- that is, the one that started this diffusion --- initially had. How good is a final result also depends on how $ \mathfrak{N} $ is defined and on how one defines what makes a result better than another (e.g., a fitness function). While this general definition is not formally stated, in the particular case of the present paper, these matters become formal and precise in Definition \ref{BdefPrevalenceN'_BB} --- see also~\cite{Chaitin2012,Chaitin2014,Hernandez-Orozco2018,Abrahao2017,Chaitin2018} for the Busy Beaver function as a measure of fitness.  
		\end{Bremarknote}

		\begin{Bsubdefinition}[SIS]\label{BdefPrevalenceN'_BB}
			In the case of an algorithmic network $ \mathfrak{N'}_{BB}(N,f,t_i , j) $ we denote the fraction of nodes/programs in $ \mathfrak{P'}_{BB}(N) $ with partial outputs at time instant $ t' $ in $ \mathbf{X}_{ partial }(  \mathfrak{N'}_{BB}(N,f,t_i , j) )|_{t}^{t'} $ that are equal to the biggest partial output of a node/program at time instant $ t_i $ as
			
			\[
			{ \tau_{\rho}( N,f,t_i , j ) }|_{t}^{t'}
			=
			\frac{ \left| \mathbf{X}_{ { \tau_{\rho}( N,f,t_i , j ) }|_{t}^{t'} } \right| }
			{ N }
			\] 
		\end{Bsubdefinition}
		
		\begin{Bsubsubdefinition}[SIS]\label{BdefAveragePrevalenceN'_BB}
			In the average case for all possible respective node mappings into the population, we define the \textbf{prevalence in a time interval} (or \textbf{average/expected density of ``infected'' nodes}) as
			
			\[
			{ \tau_{\mathbf{E}(\rho)}( N,f,t_i  ) }|_{t}^{t'} 
			= 
			\sum_{\tiny 
				b_j
			} \frac{ { \tau_{\rho}( N,f,t_i , j ) }|_{t}^{t'}
			}
			{ | \{ b_j \} |
			}
			\] 
			
			\begin{Bsubsubremarknote}
				Note that this mean is being taken from a uniform distribution on the space of functions $ b_j $. An interesting future research will be to extend the results of this article to non-uniform cases on $ b_j $. 
			\end{Bsubsubremarknote}
		\end{Bsubsubdefinition}
	\end{Bdefinition}

\noindent \\

\subsection{Definitions of time centralities}


\begin{Bdefinition}[SIS]\label{BdefTimecentrality1SIS}
	Let $ w \in \mathbf{L_U} $ be a network input.
	Let $ 0 < N \in \mathbb{N} $.
	Let $ c(x) $ be a non-decreasing total computable function where
	\[
	\myfunc{c}{ \mathbb{N} } { \mathbb{N}^{*} } { x } { c(x)=y }
	\]
	\noindent Let $ \mathfrak{N'}_{BB}(N,f,t_z ,j) = ( G_t, \mathfrak{ P' }_{BB}(N), b_j ) $, where $ 0 \leq z \leq | \mathrm{T}(G_t) | -1 $, be well-defined, where there is $ t_{z_0} \in \mathrm{T}(G_t) $ such that\footnote{ This condition directly assures that this definition of time centrality is well-defined.} 
	\[
	\lim\limits_{ N \to \infty } 
	\mathbf{E}_{ \mathfrak{N'}_{BB}(N,f,t_{z_0} ) } 
	\left(
	{ {\displaystyle{\myDelta_{iso}^{net}} A} (o_i, c( z_0 + f( N, t_{z_0} ) + 2 ) )} 
	\right)
	=
	\infty
	\]
	We define the central time $ t_{cen_1} $ in generating \textbf{unlimited} expected emergent algorithmic complexity of a node in a network $ \mathfrak{N'}_{BB}(N,f,t_{z_0}, j ) $ during $ c( t_{cen_1}(c) + f( N, t_{cen_1}(c)  ) + 1 ) $ cycles as
	
	\begin{align*}
	& t_{cen_1}(c) = \min \Bigg\{
	t_i \, 
	\Bigg| \,
	\lim\limits_{ N \to \infty } 
	\mathbf{E}_{ \mathfrak{N'}_{BB}(N,f,t_{i} ) } 
	\Big(
	{\displaystyle{\myDelta_{iso}^{net}} A} (o_i, c( i + f( N, t_{i}  ) + 2 ) ) 
	\Big)
	=
	\infty  
	\Bigg\} 
	\end{align*}

\end{Bdefinition}

\begin{Bdefinition}[SIS]\label{BdefTimecentrality2SIS}
	Let $ w \in \mathbf{L_U} $ be a network input.
	Let $ 0 < N \in \mathbb{N} $.
	Let $ c(x) $ be a non-decreasing total computable function where
	\[
	\myfunc{c}{ \mathbb{N} } { \mathbb{N}^{*} } { x } { c(x)=y }
	\]
	\noindent Let $ \mathfrak{N'}_{BB}(N,f,t_z ,j) = ( G_t, \mathfrak{ P' }_{BB}(N), b_j ) $, where $ 0 \leq z \leq | \mathrm{T}(G_t) | -1 $, be well-defined, where there is $ t_{z_0} \in \mathrm{T}(G_t) $ such that\footnote{ This condition directly assures that this definition of time centrality is well-defined.} 
	
	\begin{align*}
	& \forall x \in { \mathbb{N} }|_{0}^{ |\mathrm{T}(G_t)| - 1 } \Bigg( \lim\limits_{ N \to \infty } 
	\mathbf{E}_{ \mathfrak{N'}_{BB}(N,f,t_x ) } 
	\left(
	{ {\displaystyle{\myDelta_{iso}^{net}} A} (o_i, c( x + f( N, t_x  ) + 2 ) )} \right)
	\leq \\
	& \leq 
	\lim\limits_{ N \to \infty } 
	\mathbf{E}_{ \mathfrak{N'}_{BB}(N,f,t_{z_0} ) } 
	\left(
	{ {\displaystyle{\myDelta_{iso}^{net}} A} (o_i, c( z_0 + f( N, t_{z_0}  ) + 2 ) )} 
	\right)
	\Bigg)
	\end{align*}
	
	We define the central time $ t_{cen_2} $ in generating the \textbf{maximum} expected emergent algorithmic complexity of a node in a network $ \mathfrak{N'}_{BB}(N,f,t_{z_0}, j ) $ during $ c( t_{cen_1}(c) + f( N, t_{cen_1}(c)  ) + 1 ) $ cycles as
	
	\begin{align*}
	& t_{cen_2}(c) = \min \Bigg\{
	t_i \, 
	\Bigg| \,
	\forall x  \Big( 
	\lim\limits_{ N \to \infty } 
	\mathbf{E}_{ \mathfrak{N'}_{BB}(N,f,t_x ) } 
	\left(
	{ {\displaystyle{\myDelta_{iso}^{net}} A} (o_i, c( x + f( N, t_x  ) + 2 ) )} \right)
	\leq \\
	& \leq 
	\lim\limits_{ N \to \infty } 
	\mathbf{E}_{ \mathfrak{N'}_{BB}(N,f,t_{i} ) } 
	\left(
	{ {\displaystyle{\myDelta_{iso}^{net}} A} (o_i, c( i + f( N, t_{i}  ) + 2 ) )} 
	\right)
	\Big)
	\Bigg\}
	\end{align*}
	
	\begin{Bremarknote}\label{BnoteunderdefRelatingtwotimecentralitiesSIS}
		Note that, by definition, since $ | \mathrm{T}(G_t) |  \leq \infty $ and the expected emergent algorithmic complexity is always $ \leq \infty $, if $ t_{cen_1} $ is well-defined, then $ t_{cen_2} = t_{cen_1} $. 
	\end{Bremarknote}
	
	\begin{Bremarknote}[SIS]
		Note that these time centralities depends on fucntion $c$. If function $c$ gives the smaller upper bound $ c( x + f( N, t_x  ) + 2 ) - c_0 -1 $ on achieving stationary prevalence, then $ t_{cen_1} $ or $ t_{cen_2} $ will refer to the central time in generating expected emergent algorithmic complexity of a node in a network $ \mathfrak{N'}_{BB}(N,f,t_{z_0}, j ) $ during the minimum time interval for achieving stationary prevalence.
	\end{Bremarknote}
	
\end{Bdefinition}

\noindent \\


\subsection{Lemma 1 extended}

\begin{amslemma}\label{lemmaSLLNandAITSIS}
	\noindent \\
	Let $ \mathfrak{N'}_{BB}( N, f, t, \tau, j ) = ( G_t, \mathfrak{P'}_{BB}(N), b_j ) $ be an algorithmic network.\\
	Let $ N = | \mathbf{L}_{ \mathfrak{P'}_{BB}(N) } | $.\\
	Thus, on the average as $N$ grows, we will have that there is a constant $C_4$ such that
	\[ A'_{max} \geq \lg(N) - C_4 \] \\
	
	\begin{noteunderlemma}\label{noteSLLN}
		Let $ \mathbf{P} \left[ \, X = a \, \right] $ be the usual notation for the probability of a random variable $X$ assuming value $a$. Or $ \mathbf{P} \left[ \, \mathit{statement \; S} \, \right] $ denote the probability of a $ statement \; S $ be true.
		Thus, this theorem is formaly given by the strong law of large numbers \cite{Billingsley2012} as: there is a constant $C_4$ such that 
		
		\[
		\mathbf{P} \left[ \, \lim\limits_{ N \to \infty } A'_{max} - \lg(N) + C_4  \geq 0 \, \right] = 1
		\] 
		
		In fact, the main results \ref{thmMainSIS} and \ref{thmMainCentralTimeSIS} presented in this article can be translated into such probabilistic form by putting their last statements into the square brackets like we just did to the above. For example, the reader is invited to check that, for finite subsets of $ \mathbf{L_U} $, the strong law of large numbers on a re-normalized probability distribution in the form $ \frac{1}{C \, 2^{|p|}} $ straightforwardly holds. Hence, in the limit as the size of this subset tends to $\infty$, a multiplicative term that tends to $1$ or an additive term that tends to $0$ would appear in Lemmas~\ref{lemmaSLLNandAITSIS}, \ref{lemmaComplexityonHaltp_iSIS} and \ref{lemmaComplexityonBarHaltSIS} and Theorem~\ref{thmMainSIS}. For the sake of simplifying the notation and shortening the formulas, we have chosen to state our results without using this probabilistic form. \\
	\end{noteunderlemma}
	
	\begin{proof}
		\noindent \\
		From AIT, we know that the algorithmic probability of occurring a program $ p \in \mathbf{L_U} $ is
		
		\begin{equation}
		\frac
		{ 1 }
		{ 2^{ | p | } }
		\end{equation} \\
		
		Let $\phi_{N}(p)$ be the frequency that $p$ occurs in a random sample of size $N$. In the case, this random sample is the randomly generated population $ \mathfrak{P} $.\\
		
		Define a Bernoulli trial on a random variable $ \mathbf{X} $ that assumes value $1$ \textit{iff} program $p$ occurs and assumes value $0$ \textit{iff} otherwise. Since this random sample is identically distributed and/or define a binomial distribution where $ \sum\limits_{n=1}^{\infty} \frac{ Var[ X_n ] }{n^2} < \infty $, we will have from the \textbf{strong law of large numbers} that 
		
		\begin{equation}\label{stepSLLN}
		\mathbf{P} \left[ \, \lim\limits_{ N \to \infty } \phi_{N}(p) = \frac{ 1 }{ 2^{ | p | } } \, \right]
		=
		1
		\end{equation} \\
		
		Thus, when $N$ is large enough, one expects that $p$ occurs $ \frac{ N }{ 2^{ | p | } } $ times within $N$ random tries. That is, since $p$ was arbitrary, the probability distribution in $N$ random tries tends to match the algorithmic probability distribution on $\mathbf{L_U}$ when $N$ goes to $\infty$. \\
		
		Let $BB(k)$ be the Busy Beaver value for an arbitrary large enough $ k \in \mathbb{N} $ defined on machine $ \mathbf{U} $. We choose, for example, the definition of the Busy Beaver function in which $ BB(k) $ gives the biggest value that a program $ p \in \mathbf{L_U} $, where $ |p| \leq k $, returns when running on machine $ \mathbf{U} $. \\
		
		From AIT~\cite{Chaitin2012,Chaitin2004,Li1997}, we know there are constants $C_{\Omega}, C_{BB} \geq 0 $ and a program $p_{BB(k)} \in \mathbf{L_U} $ such that, for every $ w \in \mathbf{L_U} $,
		
		\[
		\mathbf{U}( p_{BB(k)} \circ w) = \mathbf{U}( p_{BB(k)} ) = BB(k)
		\]
		
		\begin{equation}\label{stepAITmagic}
		\text{and}
		\end{equation}
		
		\[
		k - C_{\Omega}  \leq A(BB(k)) \leq | p_{BB(k)} | \leq k + C_{BB} 
		\]

		\begin{equation*}
		\text{and}
		\end{equation*}

		\[
		\forall x \geq BB(k) \; \Big( \, A(x)  \geq k - C_{\Omega} \, \Big)
		\] \\
		
		Since $k$ was arbitrary, let $ k = \lg(N) - C_{BB} $. \\
		
		From Step \eqref{stepSLLN} we will have that, when $N$ is large enough, one should expect that $ p_{BB( \lg(N) - C_{BB} )} $ occurs at least $ \frac{ N }{ 2^{ | p_{BB( \lg(N) - C_{BB} )} | } }  $ times where
		
		\[ 
		\frac{ N }{ 2^{ | p_{BB( \lg(N) - C_{BB} )} | } } 
		\geq
		\frac{ N }{ 2^{ \lg(N) } } 
		=
		1
		\]

		\noindent That is, from conditional probabilities,
		
		\begin{equation}
		\mathbf{P} \left[ \, \lim\limits_{ N \to \infty } \phi_{N}( p_{BB( \lg(N) - C_{BB} )} ) N \geq \frac{ N }{ 2^{ \lg(N) } } = 1 \, \right]
		\geq
		\end{equation}
		\[
		\geq \mathbf{P} \left[ \, \lim\limits_{ N \to \infty } \phi_{N}(p_{BB( \lg(N) - C_{BB} )})N = \frac{ N }{ 2^{ | p_{BB( \lg(N) - C_{BB} )} | } } \, \right]
		=
		1
		\] \\
		
		Let $ C_4 =  2 \, C_{\Omega} + 2 \, C_{BB} $. 
		
		From Definitions \ref{BdefL'_BB}, \ref{BdefA'_max} and \ref{BdefN'_BB} and Step \eqref{stepAITmagic} , since any node/program count as isolated from the network when $c=1$, we will have that, for large enough $ N $, \\
		\begin{center}if\end{center}
		
		\[
		 \phi_{N}( p_{BB( \lg(N) - C_{BB} )} ) N \geq  1
		\]
		
		\begin{equation} \label{stepInequalitieschain}
		\text{then}
		\end{equation}
		
		\begin{align*}
		& A'_{max} 
		\geq \\
		& \geq 
		\lg(N) - C_{BB} - C_{\Omega} 
		\geq \\
	    & \geq
	    \big| p_{BB( \lg(N) - C_{BB} )} \big| - C_{BB} - C_{\Omega} 
	    \geq \\
	    & \geq 
	    A(BB( \lg(N) - C_{BB} )) - C_{BB} - C_{\Omega}   = \\
	    & =
		A(\mathbf{U}( P_{prot'_R} \circ p_{BB( \lg(N) - C_{BB} )} )) - C_{BB} - C_{\Omega} 
		= \\
		&=
		A(\mathbf{U}(p_{BB( \lg(N) - C_{BB} )})) - C_{BB} - C_{\Omega}  = \\
		& =
		A(BB( \lg(N) - C_{BB} )) - C_{BB} - C_{\Omega} \geq \\
		& \geq 
		\lg(N) - C_{BB} - C_{\Omega} - C_{BB} - C_{\Omega} 
		= \\
		& = 
		\lg(N) - C_4  
		\end{align*} \\
		
		Thus, from conditional probabilities, we will have that
		
		\[
		\mathbf{P} \left[ \, \lim\limits_{ N \to \infty } A'_{max} - \lg(N) + C_4  \geq 0 \, \right] \geq
		\]
		\[
		\geq
		\mathbf{P} \left[ \, \lim\limits_{ N \to \infty } \phi_{N}( p_{BB( \lg(N) - C_{BB} )} ) N \geq 1 \, \right]
		=
		1
		\]
	\end{proof}
\end{amslemma}

\noindent \\

\subsection{Lemma 2 extended}

\begin{amslemma}\label{lemmaComplexityp_iSIS}
	Given a population $ \mathfrak{P'}_{BB}(N) $ defined in \ref{BdefP'_BB} , where \\ $ p_i \in Halt_{iso}( \mathbf{L}_{ \mathfrak{P'}_{BB}(N) }, w, c ) $ and $ P_{prot'_R} \circ p_i \in \mathfrak{P'}_{BB}(N) $ and $ N \in \mathbb{N^*} $ is arbitrary, there is a constant $C_1$ such  that 
	
	\begin{align*}
	A (\mathrm{\textbf{U}}(p_{iso} ( p_i , c )))
	\leq
	C_1 + |p_i| + A(w) + A(c)
	\end{align*} \\
	
	\begin{noteunderlemma}
		Note that from the Definition \ref{BdefEAConN'_BB} this result is independent of any topology in which $ \mathfrak{P'}_{BB}(N) $ could be networked.
	\end{noteunderlemma}
	
	\begin{proof}
		\noindent \\
		Let $ N \in \mathbb{N^*} $ be arbitrary.
		Remember the definition of $ \mathbf{L'}_{BB} $ in \ref{BdefL'_BB}. And note that $p_i$ is a program in $ \mathbf{L_U} $. \\
		Then, from Definition~\ref{BdefHaltSIS}, there is at least one program $p'$ such that 
		
		\[
		\mathrm{\textbf{U}}(p_{iso} ( p_i , c ))
		=			 \mathbf{U} (p' \circ p_i \circ w \circ c)
		\] 
		is a well-defined value for every $p_i \in Halt_{iso}( \mathbf{L}_{ \mathfrak{P'}_{BB}(N) }, w, c ) $.
		
		Take the smallest such program $p'$ and let $ C_1 = |p'| + C_{\circ 4} $, where from AIT there is constant $ C_{\circ 4} $ such that
		\[
		A ( \mathbf{U} ( p' \circ p_i \circ w \circ c ) )
		\leq
		C_{\circ 4} + |p'| + |p_i| + A(w) + A(c)
		\]
		
		Then, from AIT, we will have that
		
		\[
		A ( \mathrm{\textbf{U}}(p_{iso} ( p_i , c )) )			 =
		A ( \mathbf{U} (p' \circ p_i \circ w \circ c) )
		\leq
		C_1 + |p_i| + A(w) + A(c)
		\] \\
	\end{proof}
\end{amslemma}

\noindent \\

\subsection{Lemma 3 extended}

\begin{amslemma}\label{lemmaGibbsandalgorithmicentropySIS}
	Given a population $ \mathfrak{P'}_{BB}(N) $ defined in \ref{BdefP'_BB} , where $ p_i \in \mathbf{L_U} $ and $ P_{prot'_R} \circ p_i \in \mathfrak{P'}_{BB}(N) $, from AIT and Gibb's (or Jensen's) inequality, we will have that
	
	\[
	\frac
	{1}
	{ \Omega(w,c) }
	\left( \lim_{ N \to \infty } \sum_{ p_i \in Halt_{iso}( \mathbf{L_U}(N), w, c ) } \frac
	{ |p_i| }
	{ 2^{ |p_i| } } \right)
	+ 
	\lg( \Omega(w,c) )
	\leq
	\lim_{ N \to \infty }
	\lg(\Omega(w,c) N) 
	\] \\
	
	\begin{proof}
		\noindent \\
		Since $\mathbf{L_U}$ is self-delimited, from AIT and the Definition \ref{BdefOmegawcSIS} , we will have that
		\begin{equation}
		\lim_{ N \to \infty }
		\sum_{ p_i \in Halt_{iso}( \mathbf{L_U}(N), w, c ) } \frac
		{1}
		{ 2^{ |p_i| } }
		=
		\Omega(w,c)
		<
		1
		\end{equation} \\
		
		Hence,
		\begin{equation}\label{stepNormalizationofomegawc}
		\lim_{ N \to \infty }
		\sum_{ p_i \in Halt_{iso}( \mathbf{L_U}(N), w, c ) } \frac
		{1}
		{ \Omega(w,c) 2^{ |p_i| } }
		=
		\lim_{ N \to \infty }
		\sum_{ p_i \in Halt_{iso}( \mathbf{L_U}(N), w, c ) } \frac
		{1}
		{ 2^{ |p_i| + \lg( \Omega(w,c) ) } }
		=
		1
		\end{equation} \\
		
		Thus, from Step \eqref{stepNormalizationofomegawc} ,

		\begin{equation}\label{stepNormalizedalgorithmicentropy}
		\lim_{ N \to \infty }
		\sum_{ p_i \in Halt_{iso}( \mathbf{L_U}(N), w, c ) } \frac
		{ |p_i| + \lg(\Omega(w,c)) }
		{ 2^{ |p_i| + \lg(\Omega(w,c)) } }
		= 
		\end{equation}
		
		\[
		= \lim_{ N \to \infty }
		\sum_{ p_i \in Halt_{iso}( \mathbf{L_U}(N), w, c ) } \frac
		{ |p_i| }
		{ 2^{ |p_i| } \Omega(w,c) }
		+
		\lim_{ N \to \infty }
		\sum_{ p_i \in Halt_{iso}( \mathbf{L_U}(N), w, c ) } \frac
		{ \lg(\Omega(w,c)) }
		{ 2^{ |p_i| } \Omega(w,c) }
		=
		\]
		
		\[
		= \frac
		{1}
		{ \Omega(w,c) }
		\left( \lim_{N \to \infty} \sum_{ p_i \in Halt_{iso}( \mathbf{L_U}(N), w, c ) } \frac
		{ |p_i| }
		{ 2^{ |p_i| } } \right)
		+ 
		\lg( \Omega(w,c) )
		\] \\
		
		Since \eqref{stepNormalizationofomegawc} holds, from Gibb's (or Jensen's) inequality \cite{Cover2005} \cite{MacKay2005}, we will have that
		
		\begin{equation}\label{stepGibbsinequality}
		\lim_{ N \to \infty }
		\sum_{ p_i \in Halt_{iso}( \mathbf{L_U}(N), w, c ) } \frac
		{ |p_i| + \lg(\Omega(w,c)) }
		{ 2^{ |p_i| + \lg(\Omega(w,c)) } }
		\leq
		\lim_{ N \to \infty }
		\lg( | Halt_{iso}( \mathbf{L_U}(N), w, c ) | ) 
		\end{equation} \\
		
		We also have by Definitions \ref{BdefOmegawcSIS} and \ref{BdefL_USIS} and by the law of large numbers as in Lemma~\ref{lemmaSLLNandAITSIS} that\footnote{In fact, due to the repetitions in the population, one may have that $ | \mathbf{L_U}(N) | < N $. Thus, out final results can be even improved in the future. }
		
		\begin{equation}\label{stepLimitofHalt}
		\lim_{ N \to \infty }
		\lg( | Halt_{iso}( \mathbf{L_U}(N), w, c ) | )
		\leq
		\lim_{ N \to \infty }
		\lg( \Omega(w,c) \left| \mathbf{L_U}(N)  \right| )
		\leq
		\lim_{ N \to \infty }
		\lg( \Omega(w,c) N )
		\end{equation} \\
		
		Therefore, from Steps \eqref{stepNormalizedalgorithmicentropy} , \eqref{stepGibbsinequality} and \eqref{stepLimitofHalt} we will have that
		
		\[
		\frac
		{1}
		{ \Omega(w,c) }
		\left( \lim_{ N \to \infty } \sum_{ p_i \in Halt_{iso}( \mathbf{L_U}(N), w, c ) } \frac
		{ |p_i| }
		{ 2^{ |p_i| } } \right)
		+ 
		\lg( \Omega(w,c) )
		\leq
		\lim_{ N \to \infty }
		\lg(\Omega(w,c) N)	
		\]
	\end{proof}
\end{amslemma}
\noindent \\

\subsection{Lemma 4 extended}

\begin{amslemma}\label{lemmaComplexityonHaltp_iSIS}
	Given a population $ \mathfrak{P'}_{BB}(N) $ defined in \ref{BdefP'_BB} , where $ p_i \in \mathbf{L_U} $ and $ P_{prot'_R} \circ p_i \in \mathfrak{P'}_{BB}(N) $, from Lemma \ref{lemmaGibbsandalgorithmicentropySIS} we will have that
	
	\[
	\lim_{ N \to \infty }
	\frac{
		\sum\limits_{ b_j } 
		\frac{
			\sum\limits_{ p_i \in Halt_{iso}( \mathbf{L}_{\mathfrak{P'}_{BB}(N)} , w, c )  } 
			{ | p_i | }
		}
		{N}
	}
	{ | \{ b_j \} | }
	\leq
	\lim_{ N \to \infty }
	\Omega(w,c) \lg(N)
	\] \\
	
	\begin{noteunderlemma}
		This theorem gives an upper bound for the algorithmic complexity of the randomly generated part of the elements of the population $ \mathfrak{ P' }_{BB}(N) $. And it will be crucial to prove a lower bound for the emergent algorithmic complexity. However, the upper bound of Lemma \ref{lemmaComplexityonHaltp_iSIS} is overestimated since an algorithmic probability distribution is far from being uniform, which is the case where Gibb's equality applies on entropies.
	\end{noteunderlemma}

	\begin{proof}
		\noindent \\
		
		From the definition of language $ \mathbf{L'}_{BB} $ in \ref{BdefL'_BB} we have that $p_i$ is independent of any topology, so that
		
		\begin{equation} \label{stepTopologicalindependenceofp_i}
		\frac
		{
			\sum\limits_{\tiny b_j } 
			\frac{
				\sum\limits_{ p_i \in Halt_{iso}( \mathbf{L}_{\mathfrak{P'}_{BB}(N)}, w, c )  } 
				{ | p_i | }
			}
			{N}
		}
		{ | \{ b_j \} | }
		=
		\end{equation}
		\[
		=
		\frac{
			\sum\limits_{ p_i \in Halt_{iso}( \mathbf{L}_{\mathfrak{P'}_{BB}(N)}, w, c )  } 
			{ | p_i | }
		}
		{N}
		\] \\
		
		And from Definition~\ref{BdefL_USIS} we have that the ramdonly generated population $ \mathbf{L}_{\mathfrak{P'}_{BB}(N)} $ tends to include all programs in $ \mathbf{L_U} $ in the limit.		
		Since one can define algorithmic probabilities in $\mathbf{L_U}$, by the Strong Law of Large Numbers, as in Lemma~\ref{lemmaSLLNandAITSIS} , we have that in the limit $\mathbf{L_U}(N)$ tends to follow the same distribution.
		
		Thus, from Definition~\ref{BdefHaltSIS}, we will have that\footnote{ Note that, since $ | Halt_{iso}( \mathbf{L_U}(N), w, c ) | < N $, then $ \sum\limits_{ p_i \in Halt_{iso}( \mathbf{L_U}, w, c )  } 
			\frac
			{ 1 }
			{ 2^{ |p_i| }  } = \Omega(w,c) < 1$. Thus, $ \lim\limits_{ N \to \infty } \sum\limits_{ p_i \in Halt_{iso}( \mathbf{L_U}(N), w, c )  } 
			\frac
			{ |p_i| }
			{ 2^{ |p_i| }  } $ is not a proper Shannon entropy. Also note that $ \mathbf{L}_{\mathfrak{P'}_{BB}(N)} $ may contain equal $ p_i $'s. However, in $ \mathbf{L_U}(N) $, each $p_i$ is unique, since it is a language and not a population. }
		
		\begin{equation} \label{stepSLLNonalgorithmicentropy}
		\lim_{ N \to \infty }
		{ \frac{1}{N} }
		\sum\limits_{ p_i \in Halt_{iso}( \mathbf{L}_{\mathfrak{P'}_{BB}(N)}, w, c )  } { | p_i | }
		=
		\lim_{ N \to \infty }
		{ \frac{1}{N} }
		\sum\limits_{ p_i \in Halt_{iso}( \mathbf{L_U}(N), w, c )  } 
		\frac
		{ N \, |p_i| }
		{ 2^{ |p_i| }  }
		=
		\end{equation}
		\[
		=
		\lim_{ N \to \infty }
		\sum\limits_{ p_i \in Halt_{iso}( \mathbf{L_U}(N), w, c )  } 
		\frac
		{ |p_i| }
		{ 2^{ |p_i| }  } 
		\] \\
		
		From Lemma \ref{lemmaGibbsandalgorithmicentropySIS} we will have that
		
		\begin{equation} \label{stepLemmaGibbsandalgorithmicentropy}
		\lim_{ N \to \infty } \sum_{ p_i \in Halt_{iso}( \mathbf{L_U}(N), w, c ) } \frac
		{ |p_i| }
		{ 2^{ |p_i| } }
		\leq
		\lim_{ N \to \infty }
		\Omega(w,c) \lg( \Omega(w,c) N )		
		-
		\Omega(w,c) \lg( \Omega(w,c) )
		\end{equation} \\
		
		And
		
		\begin{equation} \label{stepSimplifyingLemmaGibbsandalgorithmicentropy}
		\Omega(w,c) \lg( \Omega(w,c) N )		
		-
		\Omega(w,c) \lg( \Omega(w,c) )
		=
		\end{equation}
		\[
		=
		\Omega(w,c) \lg( \Omega(w,c))
		+
		\Omega(w,c) \lg( N )
		-
		\Omega(w,c) \lg( \Omega(w,c) )=
		\]		
		\[
		=
		\Omega(w,c) \lg(N)
		\] \\
		
		So, from Steps \eqref{stepTopologicalindependenceofp_i} , \eqref{stepSLLNonalgorithmicentropy} , \eqref{stepLemmaGibbsandalgorithmicentropy} and \eqref{stepSimplifyingLemmaGibbsandalgorithmicentropy}
		
		\[
		\lim_{ N \to \infty }
		\frac
		{
			\sum\limits_{\tiny b_j } 
			\frac{
				\sum\limits_{ p_i \in Halt_{iso}( \mathbf{L}_{\mathfrak{P'}_{BB}(N)}, w, c )  } 
				{ | p_i | }
			}
			{N}
		}
		{ | \{ b_j \} | }
		=
		\]
		\[
		=
		\lim_{ N \to \infty } \sum_{ p_i \in Halt_{iso}( \mathbf{L_U}(N), w, c ) } \frac
		{ |p_i| }
		{ 2^{ |p_i| } }
		\leq
		\Omega(w,c) \lg(N)
		\]
	\end{proof}
\end{amslemma}

\noindent \\

\subsection{Lemma 5 extended}

\begin{amslemma}\label{lemmaComplexityonBarHaltSIS}
	Given a population $ \mathfrak{P'}_{BB}(N) $ defined in \ref{BdefP'_BB} , where $ p_i \in \mathbf{L_U} $ and $ P_{prot'_R} \circ p_i = o_i \in \mathfrak{P'}_{BB}(N) $, there is a constant $C_0$ such that
	
	\[
	\lim_{ N \to \infty }
	\frac{
		\sum\limits_{ b_j } 
		\frac{
			\sum\limits_{ p_i \in \overline{ Halt }_{iso} ( \mathbf{L}_{\mathfrak{P'}_{BB}(N)}, w, c )  } 
			{ A ( \mathrm{\textbf{U}}(p_{iso} ( p_i , c )) ) }
		}
		{N}
	}
	{ | \{ b_j \} | }
	=
	C_0 ( 1 - \Omega(w,c) )
	\] \\
	
	\begin{noteunderlemma}
		Note that every lemma until here deals with the behavior of \\ 
		$ A ( \mathrm{\textbf{U}}(p_{iso} ( p_i , c )) ) $, so that they gave tools to obtain an upper bound for the expected algorithmic complexity of what each node can do when isolated. Since it is an upper bound for the algorithmic complexity of what each node can do when isolated and in the emergent algorithmic complexity it contributes negatively as defined in~\ref{BdefEAC}, then these results will help us to achieve a lower bound for the expected emergent algorithmic complexity. Furthermore, these results are independent of any topological feature that the algorithmic network $ \mathfrak{N'}_{BB}( N, f, t, \tau, j ) $ might have. 
	\end{noteunderlemma}
	
	\begin{proof}
		\noindent \\
		As in Step \eqref{stepTopologicalindependenceofp_i} , $ A ( \mathrm{\textbf{U}}(p_{iso} ( p_i , c )) ) $ is independent of any topology, so that
		
		\begin{equation} \label{stepTopologicalindependenceonBarHalt}
		\lim_{ N \to \infty }
		\frac{
			\sum\limits_{ b_j } 
			\frac{
				\sum\limits_{ p_i \in \overline{ Halt }_{iso} ( \mathbf{L}_{\mathfrak{P'}_{BB}(N)}, w, c )  } 
				{ A ( \mathrm{\textbf{U}}(p_{iso} ( p_i , c )) ) }
			}
			{N}
		}
		{ | \{ b_j \} | }
		=
		\lim\limits_{ N \to \infty }
		\frac{
			\sum\limits_{ p_i \in \overline{ Halt }_{iso} ( \mathbf{L}_{\mathfrak{P'}_{BB}(N)}, w, c )  } 
			{ A ( \mathrm{\textbf{U}}(p_{iso} ( p_i , c )) ) }
		}
		{N}
		\end{equation} \\
		
		From the Definition \ref{BdefBarHaltSIS} , as in Step \eqref{stepSLLNonalgorithmicentropy} , we will have that
		
		\begin{equation} \label{stepSLLNonBarHalt}
		\lim\limits_{ N \to \infty }
		\frac{
			\sum\limits_{ p_i \in \overline{ Halt }_{iso} ( \mathbf{L}_{\mathfrak{P'}_{BB}(N)}, w, c )  } 
			{ A ( \mathrm{\textbf{U}}(p_{iso} ( p_i , c )) ) }
		}
		{N}
		=
		\lim\limits_{ N \to \infty }
		\sum\limits_{ p_i \in \overline{ Halt }_{iso} ( \mathbf{L_U}(N), w, c )  } 
		{ \frac {A ( \mathrm{\textbf{U}}(p_{iso} ( p_i , c )) ) } { 2^{ |p_i| } } }
		\end{equation}	\\
		
		Now, let $ A(0)=C_0 $. \\
		Since $ p_i \in \overline{ Halt }_{iso} ( \mathbf{L_U}(N), w, c ) $ and $ \mathfrak{P'}_{BB}(N) $ is sensitive to oracles as defined in \ref{BdefP'_BB}, then by the definition of $p_{iso}$ in \ref{BdefEAConN'_BB} we will have that, for every $p_i \in \overline{ Halt }_{iso} ( \mathbf{L_U}, w, c )$,
		
		\begin{equation}
		A ( \mathrm{\textbf{U}}(p_{iso} ( p_i , c )) ) = A(0) = C_0
		\end{equation} \\
		
		Thus, by the definition of $\Omega(w,c)$ in \ref{BdefOmegawcSIS} ,
		
		\begin{equation} \label{stepAlgorithmicentropyfromC_0}
		\lim\limits_{ N \to \infty }
		\sum\limits_{ p_i \in \overline{ Halt }_{iso} ( \mathbf{L_U}(N), w, c )  } 
		{ \frac {A ( \mathrm{\textbf{U}}(p_{iso} ( p_i , c )) ) } { 2^{ |p_i| } } }
		=
		\lim\limits_{ N \to \infty }
		\sum\limits_{ p_i \in \overline{ Halt }_{iso} ( \mathbf{L_U}(N), w, c )  } 
		{ \frac { C_0 } { 2^{ |p_i| } } }
		=
		\end{equation}
		\[ =
		\lim\limits_{ N \to \infty }
		C_0
		\sum\limits_{ p_i \in \overline{ Halt }_{iso} ( \mathbf{L_U}(N), w, c )  } 
		{ \frac {1} { 2^{ |p_i| } } }
		=
		C_0 ( 1 - \Omega(w,c) )
		\] \\
		
		And we conclude from Steps \eqref{stepTopologicalindependenceonBarHalt} , \eqref{stepSLLNonBarHalt} and \eqref{stepAlgorithmicentropyfromC_0} that
		
		\[
		\lim_{ N \to \infty }
		\frac{
			\sum\limits_{ b_j } 
			\frac{
				\sum\limits_{ p_i \in \overline{ Halt }_{iso} ( \mathbf{L}_{\mathfrak{P'}_{BB}(N)}, w, c )  } 
				{ A ( \mathrm{\textbf{U}}(p_{iso} ( p_i , c )) ) }
			}
			{N}
		}
		{ | \{ b_j \} | }
		=
		\lim\limits_{ N \to \infty }
		\sum\limits_{ p_i \in \overline{ Halt }_{iso} ( \mathbf{L_U}(N), w, c )  } 
		{ \frac {A ( \mathrm{\textbf{U}}(p_{iso} ( p_i , c )) ) } { 2^{ |p_i| } } }
		= \]
		\[
		= C_0 ( 1 - \Omega(w,c) )
		\]
	\end{proof}
\end{amslemma}

\noindent \\

\subsection{Lemma 6 extended}

\begin{amslemma}\label{lemmaMinComplexityonDiffusionSIS}
	Let $ \mathfrak{P'}_{BB}(N) $ be a population in an arbitrary algorithmic network $ \mathfrak{N'}_{BB} (N, f, t , j)=(G_t, \mathfrak{P'}_{BB} (N),b_j) $ as defined in \ref{BdefN'_BB} and \ref{BdefP'_BB}. \\
	Let $ t_0 \leq t \leq t' \leq t_{ |\mathrm{T}(G_t)|-1 } $. \\ 
	Let $ c \in \mathfrak{C_{BB}}$ be an arbitrary number of cycles where $ c_0 + t' + 1 \leq c $. \\ 
	Then, there is a constant $C_2$ such that
	
	\[
	\frac{
		{
			\sum\limits_{ b_j  }
		}
		\frac{
			\sum\limits_{ p_i \in \mathfrak{P}_{BB}(N) } { A (\mathrm{\textbf{U}}(p_{net}^{b_j} ( o_i ,  c ) ))  }
		}
		{N}
	}
	{ | \{ b_j \} | }
	\geq
	( A'_{max} - C_2 ) \,
	{ \tau_{\mathbf{E}(\rho)}( N,f,t  ) }|_{t}^{t'}
	+ C_2
	\] \\
	
	\begin{proof}
		\noindent \\
		
		Let  $ \mathbf{X}_{ { \tau_{\rho}( N,f,t , j ) }|_{t}^{t'} } $ denote the set of nodes/programs that belong to fraction $ { \tau_{\rho}( N,f,t , j ) }|_{t}^{t'} $ as defined in \ref{BdefPrevalenceN'_BB}. 
		
		Hence, 
		\[ 
		| \mathbf{X}_{ { \tau_{\rho}( N,f,t , j ) }|_{t}^{t'} } | = N { \tau_{\rho}( N,f,t , j ) }|_{t}^{t'}
		\] \\
		
		Let $ C_2 = \min \{ A(w) \mid \exists x \in \mathbf{L_U} ( \mathbf{U}( x ) = w ) \} $. \footnote{ Note that depending on the choice of the programming language one may have $ C_2 \leq A(0) $ for example .} \\
		
		From the Definition \ref{BdefPrevalenceN'_BB} we will have that
		
		\[
		\frac{
			{
				\sum\limits_{ b_j  }
			}
			\frac{
				\sum\limits_{ o_i \in \mathfrak{P'}_{BB}(N) } { A (\mathrm{\textbf{U}}(p_{net}^{b_j} ( o_i ,  c ) ))  }
			}
			{N}
		}
		{ | \{ b_j \} | } =
		\]
		\begin{align} \label{stepSumoffractionsmax2}
		& =
		\frac
		{ \sum\limits_{ b_j  } 
			\left(
			\frac
			{ \sum\limits_{ o_i \in \mathbf{X}_{ { \tau_{\rho}( N,f,t , j ) }|_{t}^{t'} } } { A (\mathrm{\textbf{U}}(p_{net}^{b_j} ( o_i ,  c ) ))  }  }
			{ { \tau_{\rho}( N,f,t , j ) }|_{t}^{t'} N }
			{ \tau_{\rho}( N,f,t , j ) }|_{t}^{t'}
			\right)
		}
		{ | \{ b_j \} | } +
		\end{align}
		\[
		+ 
		\frac
		{ \sum\limits_{ b_j  } 
			\left(
			\frac
			{ \sum\limits_{ o_i \in \mathfrak{P'}_{BB}(N) \setminus \mathbf{X}_{ { \tau_{\rho}( N,f,t , j ) }|_{t}^{t'} } } { A (\mathrm{\textbf{U}}(p_{net}^{b_j} ( o_i ,  c ) ))  }  }
			{ \left| \mathfrak{P'}_{BB}(N) \setminus \mathbf{X}_{ { \tau_{\rho}( N,f,t , j ) }|_{t}^{t'} } \right| }
			{ \left( \frac 
				{ \left| \mathfrak{P'}_{BB}(N) \setminus \mathbf{X}_{ { \tau_{\rho}( N,f,t , j ) }|_{t}^{t'} } \right| }
				{N} \right) }
			\right)
		}
		{ | \{ b_j \} | }
		\] \\

		And from Definitions \ref{BdefPrevalenceN'_BB}, \ref{BdefL'_BB} and \ref{BdefA'_max} we will have that
		
		\begin{equation} \label{stepA_maxandfractionmax2}
		\frac
		{ \sum\limits_{ o_i \in \mathbf{X}_{ { \tau_{\rho}( N,f,t , j ) }|_{t}^{t'} } } { A (\mathrm{\textbf{U}}(p_{net}^{b_j} ( o_i ,  c ) ))  }  }
		{ { \tau_{\rho}( N,f,t , j ) }|_{t}^{t'} N }
		{ \tau_{\rho}( N,f,t , j ) }|_{t}^{t'}
		\geq
		A'_{max} \, { \tau_{\rho}( N,f,t , j ) }|_{t}^{t'}
		\end{equation} \\
		
		\noindent
		and, analogously, the following always holds despite on which node fraction
		\[ 
		\frac 
		{ \left| \mathfrak{P'}_{BB}(N) \setminus \mathbf{X}_{ { \tau_{\rho}( N,f,t , j ) }|_{t}^{t'} } \right| }
		{N}
		\] 
		is centered and whenever it starts to diffuse
		
		\begin{equation} \label{stepC_2andfractionmax2}
		\frac
		{ \sum\limits_{ o_i \in \mathfrak{P'}_{BB}(N) \setminus \mathbf{X}_{ { \tau_{\rho}( N,f,t , j ) }|_{t}^{t'} } } { A (\mathrm{\textbf{U}}(p_{net}^{b_j} ( o_i ,  c ) ))  }  }
		{ \left| \mathfrak{P'}_{BB}(N) \setminus \mathbf{X}_{ { \tau_{\rho}( N,f,t , j ) }|_{t}^{t'} } \right|  }
		\frac 
		{ \left| \mathfrak{P'}_{BB}(N) \setminus \mathbf{X}_{ { \tau_{\rho}( N,f,t , j ) }|_{t}^{t'} } \right| }
		{N}
		\geq
		\end{equation} 
		\[
		\geq
		C_2 \, \frac 
		{ \left| \mathfrak{P'}_{BB}(N) \setminus \mathbf{X}_{ { \tau_{\rho}( N,f,t , j ) }|_{t}^{t'} } \right| }
		{N}
		\]\\
		
		Thus, since we have that $ { \tau_{\rho}( N,f,t , j ) }|_{t}^{t'} + \frac 
		{ \left| \mathfrak{P'}_{BB}(N) \setminus \mathbf{X}_{ { \tau_{\rho}( N,f,t , j ) }|_{t}^{t'} } \right| }{N} = 1 $, then by Steps \eqref{stepSumoffractionsmax2} , \eqref{stepA_maxandfractionmax2} and \eqref{stepC_2andfractionmax2}
		
		\begin{align}
		& \frac
		{ \sum\limits_{ b_j  } 
			\left(
			\frac
			{ \sum\limits_{ o_i \in \mathbf{X}_{ { \tau_{\rho}( N,f,t , j ) }|_{t}^{t'} } } { A (\mathrm{\textbf{U}}(p_{net}^{b_j} ( o_i ,  c ) ))  }  }
			{ { \tau_{\rho}( N,f,t , j ) }|_{t}^{t'} N }
			{ \tau_{\rho}( N,f,t , j ) }|_{t}^{t'}
			\right)
		}
		{ | \{ b_j \} | } +
		\end{align}
		\[
		+ 
		\frac
		{ \sum\limits_{ b_j  } 
			\left(
			\frac
			{ \sum\limits_{ o_i \in \mathfrak{P'}_{BB}(N) \setminus \mathbf{X}_{ { \tau_{\rho}( N,f,t , j ) }|_{t}^{t'} } } { A (\mathrm{\textbf{U}}(p_{net}^{b_j} ( o_i ,  c ) ))  }  }
			{ \left| \mathfrak{P'}_{BB}(N) \setminus \mathbf{X}_{ { \tau_{\rho}( N,f,t , j ) }|_{t}^{t'} } \right| }
			{ \left( \frac 
				{ \left| \mathfrak{P'}_{BB}(N) \setminus \mathbf{X}_{ { \tau_{\rho}( N,f,t , j ) }|_{t}^{t'} } \right| }
				{N} \right) }
			\right)
		}
		{ | \{ b_j \} | } \geq
		\]
		\[
		\geq
		\frac
		{ \sum\limits_{ b_j  } 
			\left( 
			A'_{max}
			{ \tau_{\rho}( N,f,t , j ) }|_{t}^{t'}
			+ 
			C_2
			\frac 
		{ \left| \mathfrak{P'}_{BB}(N) \setminus \mathbf{X}_{ { \tau_{\rho}( N,f,t , j ) }|_{t}^{t'} } \right| }{N}
			\right)
		}
		{ | \{ b_j \} | } =	
		\]
		\[
		=
		\frac
		{ \sum\limits_{ b_j  } 
			\left( 
			( A'_{max} - C_2 )
			{ \tau_{\rho}( N,f,t , j ) }|_{t}^{t'}
			+ 
			C_2
			\right)
		}
		{ | \{ b_j \} | }
		\]
	\end{proof}
\end{amslemma}

\noindent \\

\noindent \\

\subsection{Theorem 1 extended}


\begin{thm}\label{thmMainSIS}
	\noindent \\
	
	Let $ w \in \mathbf{L_U} $ be a network input. \\
	
	Let $ 0 < N \in \mathbb{N} $.\\
	
	Let $ \mathfrak{N'}_{BB}(N,f,t ,j) = ( G_t, \mathfrak{ P' }_{BB}(N), b_j ) $ be well-defined.\\
	
	Let $ t_0 \leq t \leq t' \leq t_{ |\mathrm{T}(G_t)|-1 } $.\\
	
	Let $ \myfunc{c}{ \mathbb{N} } { \mathfrak{C_{BB}} } { x } { c(x)=y } $  be a total computable function where $ c(x) \geq c_0 + t'+1 $. \\

	Then, we will have that:
	
	\begin{align*}
		& \lim\limits_{ N \to \infty }
		\mathbf{E}_{ \mathfrak{N'}_{BB}(N,f,t ) } 
		\left(
		{ {\displaystyle{\myDelta_{iso}^{net}} A} (o_i,c(x))} 
		\right)
		\geq
		\lim\limits_{ N \to \infty }
		\left( 
		{ \tau_{\mathbf{E}(\rho)}( N,f,t  ) }|_{t}^{t'}
		-
		\Omega(w,c(x))
		\right)
		\lg(N) - \\
		& - \Omega(w,c(x)) \lg(x) - 2 \, \Omega(w,c(x))\lg(\lg(x)) - A(w) - C_5
	\end{align*}
	
	\noindent \\
	
	\begin{noteunderthm}
		Thus, note for example that for bigger enough values of $x$ compared to $N$ one can make this lower bound always negative. One of the main ideas behind forthcoming results in this paper is to find optimal conditions where this lower bound is not only positive, but also goes to $\infty$.
	\end{noteunderthm}
	
	\begin{noteunderthm}
		Note that this lower bound for the expected emergent algorithmic complexity is dependent on the value in the domain of the function $c$ and not on function $c$ itself, even if it grows fast. And it holds as long as c is a total computable function. In fact, one may want to obtain this theorem for fixed values of $c$ in which it is not a function but an arbitrary value. And the same result also holds in this case. The reader is invited to check that, in addition to a slightly different constant $ C_5 $, a simple substitution of $c(x)$ for $c$ inside $ \Omega(w,c(x)) $ and of $x$ inside the logarithms for $c$ is enough\footnote{ Besides a slightly different constant $C_5$.}. 
	\end{noteunderthm}
	
	\begin{noteunderthm}\label{BnoteunderBthmSinglegraphinfamilyGSIS}
		The same result also holds if only one possible function $b_j$ is defined for each member of the family $ \mathbb{G}_{SIS}( f, t ) $. This way only one function $b_j$ will be taken into account within the sum in order to give the mean. Thus, in this case one can replace $\tau_{\mathbf{E}(\rho)}$ with $\tau_{\rho}$ not only in Theorem \ref{thmMainSIS} but also in \ref{corMainSIS} and \ref{thmMainCentralTimeSIS}. Such variation of these theorems becomes useful when one has algorithmic networks $\mathfrak{N'}_{BB}(N,f,t ,j)$ built upon a historical population-size growth in which each new node/program is linked (or not) to the previous existing algorithmic network. 
	\end{noteunderthm}
	
	\noindent \\
	
	\begin{proof}
		\noindent \\
		The proof will follow from Steps~\eqref{stepDefEEAC2} and~\eqref{stepMainthm1SIS} below.\\
		
		We have from our hypothesis on function $c$ and from AIT that there is $C_c \in \mathbb{N}$ such that, for every $ x \in \mathbb{N} $,  
		
		\begin{equation}\label{AITonfunctionc2}
			A( c(x) ) \leq C_c + A(x)
		\end{equation} \\
		
		Let $ C_5 = C_c + C_L + C_1 + C_4 - C_0 $. \\
		
		Note that, as in Step \eqref{stepSLLNonalgorithmicentropy}, we will have from Definition~\ref{BdefOmegawcSIS} that
		\begin{equation}\label{stepLLNinHaltisothm1SIS}
			\lim\limits_{ N \to \infty }
			\frac{\left| Halt_{iso}( L_{ \mathfrak{P'}_{BB}(N) }, w, c(x)) \right|}
			{ N }
			=
			\lim\limits_{ N \to \infty }
			\frac{1}{N}
			\sum\limits_{ p_i \in Halt_{iso}( \mathbf{L_U}(N), w, c(x)) }
				\frac{N}
					{2^{ | p_i | }}
			=
			\Omega( w, c(x) )
		\end{equation} \\
		
		From Definition \ref{BdefEEACN'_BB}, we have that the expected emergent algorithmic complexity of a node/program for $ \mathfrak{N'}_{BB}(N,f,t ,j) = ( G_t, \mathfrak{ P' }_{BB}(N), b_j ) $, where $ 0 < j \leq |\{ b_j \}| $ is given by
		
		\begin{align}\label{stepDefEEAC2}
			&\mathbf{E}_{ \mathfrak{N'}_{BB}(N,f,t ) } 
			\left(
			{ {\displaystyle{\myDelta_{iso}^{net}} A} (o_i,c(x))} 
			\right)
			=  
			\frac{
				{
					\sum\limits_{ b_j  }
				}
				\frac{
					\sum\limits_{ o_i \in \mathfrak{P'}_{BB}(N) } { A (\mathrm{\textbf{U}}(p_{net}^{b_j} ( o_i ,  c(x) ) )) - 
						A (\mathrm{\textbf{U}}(p_{iso} ( p_i ,  c(x) ))) }
				}
				{N}
			}
			{ | \{ b_j \} | }
		\end{align} \\
		
		And, from Definitions \ref{BdefHaltSIS} , \ref{BdefBarHaltSIS} , \ref{BdefOmegawcSIS} , \ref{BdefAveragePrevalenceN'_BB} , \ref{BdefU'_R} and Lemmas \ref{lemmaComplexityp_iSIS} , \ref{lemmaComplexityonBarHaltSIS} , \ref{lemmaComplexityonHaltp_iSIS} , \ref{lemmaMinComplexityonDiffusionSIS} , \ref{lemmaSLLNandAITSIS} and Steps \eqref{AITonfunctionc2} and \eqref{stepLLNinHaltisothm1SIS}, we will have that\footnote{ Note that $ \left| \mathbf{L_U}(N) \right| \leq N $.}
		
		\begin{align}\label{stepMainthm1SIS}
			\lim\limits_{ N \to \infty }
			\frac{
				{
					\sum\limits_{ b_j  }
				}
				\frac{
					\sum\limits_{ o_i \in \mathfrak{P'}_{BB}(N) } { A (\mathrm{\textbf{U}}(p_{net}^{b_j} ( o_i ,  c(x)) )) - 
						A (\mathrm{\textbf{U}}(p_{iso} ( p_i ,  c(x)))) }
				}
				{N}
			}
			{ | \{ b_j \} | } 
			= 
		\end{align}
		\[
		=
		\lim\limits_{ N \to \infty }
		\frac{
			{
				\sum\limits_{ b_j  }
			}
			\frac{
				\sum\limits_{ o_i \in \mathfrak{P'}_{BB}(N) } { A (\mathrm{\textbf{U}}(p_{net}^{b_j} ( o_i ,  c(x)) )) }
			}
			{N}
		}
		{ | \{ b_j \} | }
		-
		\]
		\[
		-
		\left(
		\frac{ \sum\limits_{ b_j  }
			\left(
			\frac{
				\sum\limits_{ p_i \in Halt_{iso}( L_{ \mathfrak{P'}_{BB}(N) }, w, c(x)) } {  
					A (\mathrm{\textbf{U}}(p_{iso} ( p_i ,  c(x)))) }
			}
			{N}
			+
			\frac{
				\sum\limits_{ p_i \in \overline{Halt}_{iso}( L_{ \mathfrak{P'}_{BB}(N) }, w, c(x)) } { 
					A (\mathrm{\textbf{U}}(p_{iso} ( p_i ,  c(x)))) }
			}
			{N}
			\right)
		}
		{ | \{ b_j \} | }
		\right)
		=
		\]
		\[
		=
		\lim\limits_{ N \to \infty }
		\frac{
			{
				\sum\limits_{ b_j  }
			}
			\frac{
				\sum\limits_{ o_i \in \mathfrak{P'}_{BB}(N) } { A (\mathrm{\textbf{U}}(p_{net}^{b_j} ( o_i ,  c(x)) )) }
			}
			{N}
		}
		{ | \{ b_j \} | }
		-
		\]
		\[
		-
		\left(
		\frac{ \sum\limits_{ b_j  }
			\left(
			\frac{
				\sum\limits_{p_i \in Halt_{iso}( L_{ \mathfrak{P'}_{BB}(N) }, w, c(x)) } {  
					A (\mathrm{\textbf{U}}(p_{iso} ( p_i ,  c(x)))) }
			}
			{N}
			\right)
		}
		{ | \{ b_j \} | }
		+
		C_0 ( 1 - \Omega( w, c(x)) )
		\right)	
		\geq
		\] 
		\[
		\geq
		\lim\limits_{ N \to \infty }
		\frac{
			{
				\sum\limits_{ b_j  }
			}
			\frac{
				\sum\limits_{ o_i \in \mathfrak{P'}_{BB}(N) } { A (\mathrm{\textbf{U}}(p_{net}^{b_j} ( o_i ,  c(x)) )) }
			}
			{N}
		}
		{ | \{ b_j \} | }
		-
		\]
		\[
		-
		\left(
		\frac{ \sum\limits_{ b_j  }
			\left(
			\frac{
				\sum\limits_{p_i \in Halt_{iso}( L_{ \mathfrak{P'}_{BB}(N) }, w, c(x)) } { C_1 + | p_i | + A(w) + A(c(x)) }
			}
			{N}
			+
			C_0 ( 1 - \Omega( w, c(x)) )
			\right)
		}
		{ | \{ b_j \} | }
		\right)	
		=
		\] 
		\[
		=
		\lim\limits_{ N \to \infty }
		\frac{
			{
				\sum\limits_{ b_j  }
			}
			\frac{
				\sum\limits_{ o_i \in \mathfrak{P'}_{BB}(N) } { A (\mathrm{\textbf{U}}(p_{net}^{b_j} ( o_i ,  c(x)) )) }
			}
			{N}
		}
		{ | \{ b_j \} | }
		-
		\]
		\[
		-
		\left(
		\frac{ \sum\limits_{ b_j  }
			\left(
			\frac{
				\sum\limits_{p_i \in Halt_{iso}( L_{ \mathfrak{P'}_{BB}(N) }, w, c(x)) } { | p_i | }
			}
			{N}
			\right)
		}
		{ | \{ b_j \} | }
		+
		\Omega(w,c(x)) \big( C_1 + A(w) + A(c(x)) \big)
		+
		C_0 \big( 1 - \Omega( w, c(x)) \big)
		\right)
		\geq
		\] 
		\[
		\geq
		\lim\limits_{ N \to \infty }
		\frac{
			{
				\sum\limits_{ b_j  }
			}
			\frac{
				\sum\limits_{ o_i \in \mathfrak{P'}_{BB}(N) } { A (\mathrm{\textbf{U}}(p_{net}^{b_j} ( o_i ,  c(x)) )) }
			}
			{N}
		}
		{ | \{ b_j \} | }
		-
		\]
		\[
		-
		\bigg(
		\Omega(w,c(x)) \lg(N)
		+
		\Omega(w,c(x)) \big( C_1 + A(w) + A(c(x)) \big)
		+
		C_0 \big( 1 - \Omega( w, c(x)) \big)
		\bigg)
		\geq
		\] 
		\[
		\geq
		\lim\limits_{ N \to \infty }
		\left( A'_{max} - C_2 \right) 
		{ \tau_{\mathbf{E}(\rho)}( N,f,t  ) }|_{t}^{t'}
		+
		C_2
		-
		\]
		\[
		-
		\left(
		\Omega(w,c(x)) \lg(N)
		+
		\Omega(w,c(x)) \big( C_1 + A(w) + A(c(x)) \big)
		+
		C_0 \big( 1 - \Omega( w, c(x)) \big)
		\right)
		\geq
		\]
		\[
		\geq
		\lim\limits_{ N \to \infty }
		\left( \lg(N) - C_4 - C_2 \right) 
		{ \tau_{\mathbf{E}(\rho)}( N,f,t ) }|_{t}^{t'}
		+
		C_2
		-
		\]
		\[
		-
		\Big(
		\Omega(w,c(x)) \lg(N)
		+
		\Omega(w,c(x)) \big( C_1 + A(w) + A(c(x)) \big)
		+
		C_0 \big( 1 - \Omega( w, c(x)) \big)
		\Big)
		=
		\]
		\[
		=
		\lim\limits_{ N \to \infty }
		\left( \lg(N) - C_4 - C_2 \right) 
		{ \tau_{\mathbf{E}(\rho)}( N,f,t ) }|_{t}^{t'}
		+
		C_2
		-
		\]
		\[
		-
		\Omega(w,c(x)) \lg(N)
		-
		\Omega(w,c(x)) C_1 - \Omega(w,c(x)) A(w) - \Omega(w,c(x)) A(c(x))
		-
		C_0 + C_0 \Omega( w, c(x))	
		=
		\]
		%
		\[
		=
		\lim\limits_{ N \to \infty }
		\left( 
		{ \tau_{\mathbf{E}(\rho)}( N,f,t  ) }|_{t}^{t'}
		-
		\Omega(w,c(x))
		\right)
		\lg(N)
		-
		\Omega(w,c(x)) A( c(x))
		-
		\]
		\[
		- ( C_4 + C_2 ) { \tau_{\mathbf{E}(\rho)}( N,f,t  ) }|_{t}^{t'}
		-
		\Omega(w,c(x)) C_1
		+ 
		\Omega(w,c(x)) C_0
		+ C_2 - C_0
		-
		\Omega(w,c(x)) A(w)
		\geq 	
		\]
		\[
		\geq
		\lim\limits_{ N \to \infty }
		\left( 
		{ \tau_{\mathbf{E}(\rho)}( N,f,t  ) }|_{t}^{t'}
		-
		\Omega(w,c(x))
		\right)
		\lg(N)
		-
		\Omega(w,c(x)) A( c(x))
		-
		\]
		\[
		- ( C_4 + C_2 ) 
		-
		C_1
		+ C_2 - C_0
		-
		A(w)
		=	
		\]
		\[
		=
		\lim\limits_{ N \to \infty }
		\left( 
		{ \tau_{\mathbf{E}(\rho)}( N,f,t  ) }|_{t}^{t'}
		-
		\Omega(w,c(x))
		\right)
		\lg(N)
		-
		\Omega(w,c(x)) A( c(x))
		-
		\]
		\[
		- C_4 
		-
		C_1
		-
		C_0
		-
		A(w)
		\geq	
		\]
		\[
		\geq
		\lim\limits_{ N \to \infty }
		\left( 
		{ \tau_{\mathbf{E}(\rho)}( N,f,t  ) }|_{t}^{t'}
		-
		\Omega(w,c(x))
		\right)
		\lg(N)
		-
		\]
		\[
		- \Omega(w,c(x)) \lg(x) - \Omega(w,c(x))(1+\epsilon)\lg(\lg(x)) - \Omega(w,c(x)) C_L
		-
		\]
		\[
		- \Omega( w, c(x) ) C_c
		- C_4 
		-
		C_0
		-
		C_1
		-
		A(w)	
		\]
		\[
		\geq
		\lim\limits_{ N \to \infty }
		\left( 
		{ \tau_{\mathbf{E}(\rho)}( N,f,t ) }|_{t}^{t'}
		-
		\Omega(w,c(x))
		\right)
		\lg(N)
		-
		\]
		\[
		- \Omega(w,c(x)) \lg(x) 
		- 2 \, \Omega(w,c(x))\lg(\lg(x)) 
		- C_5 
		- A(w)
		\]
	\end{proof}

\end{thm}

\noindent \\


\subsection{Corollary 1 extended}

\begin{corollaryundersubsection}\label{corMainSIS} 
	Let $ w \in \mathbf{L_U} $ be a network input.
	Let $ 0 < N \in \mathbb{N} $.
	Let $ \mathfrak{N'}_{BB}(N,f,t_z ,j) = ( G_t, \mathfrak{ P' }_{BB}(N), b_j ) $ be well-defined.
	Let $ t_z \in \mathrm{T}(G_t) $.
	Let $ \myfunc{c}{ \mathbb{N} } { \mathfrak{C_{BB}} } { x } { c(x)=y } $  be a total computable function where 
	\[ 
	c(z + f( N, t_z  ) + 2) \geq c_0 + z + f( N, t_z  ) + 2 
	\]
	and
	\[
	c( z + f( N, t_z  ) + 2 ) - c_0 - 1 \leq t_{ |\mathrm{T}(G_t)|-1 } 
	\]
	Then, we will have that:
	%
	\begin{align*}
	& \lim\limits_{ N \to \infty } 
	\mathbf{E}_{ \mathfrak{N'}_{BB}(N,f,t_z ) } 
	\left(
	{ {\displaystyle{\myDelta_{iso}^{net}} A} (o_i, c( z + f( N, t_z  ) + 2 ) )} 
	\right)
	\geq \\
	& \geq
	\lim\limits_{ N \to \infty }
	\left( 
	{ \tau_{\mathbf{E}(\rho)}( N,f,t_z  ) }|_{ t_z }^{ 
		c( z + f( N, t_z  ) + 2 ) - c_0 - 1  
	}
	-
	\Omega(w, c( z + f( N, t_z  ) + 2 ))
	\right)
	\lg(N) - \\
	& - \Omega(w, c( z + f( N, t_z  ) + 2 )) \lg( z + f( N, t_z  ) + 2) - \\
	& - 2 \, \Omega(w, c( z + f( N, t_z  ) + 2 ))\lg(\lg( z + f( N, t_z  ) + 2)) - A(w) - C_5
	\end{align*}
	
	\begin{proof}
		Remember the Definition \ref{BdefN'_BB} and conditions for Theorem~\ref{thmMainSIS} to hold. Note that 
		\[ 
		c( z + f( N, t_z  ) + 2 ) \geq 
		\]
		\[
		\geq
		c_0 + c( z + f( N, t_z  ) + 2 ) - c_0 - 1 +1 
		\geq
		\]
		\[
		\geq
		c_0 + z + f( N, t_z  ) + 2 \geq  c_0 + t_z + 1
		\geq 
		t_z
		\geq
		t_0
		\]
		\noindent satisfying conditions $ c(x) \geq c_0 + t' +1 $ and $ t_0 \leq t \leq t' \leq t_{ |\mathrm{T}(G_t)|-1 } $  in Theorem~\ref{thmMainSIS}.
		
		Thus, the proof follows directly from replacing $t$ with $t_z$, $t'$ with $ c( z + f( N, t_z  ) + 2 ) - c_0 - 1 $ and $x$ with  $z + f( N, t_z  ) +2$ in Theorem~\ref{thmMainSIS}.
	\end{proof}
	
\end{corollaryundersubsection}

\noindent \\

\subsection{Theorem 2 extended}

\begin{thm}[\ref{BthmMainCentralTimeSIS}](SIS) \label{thmMainCentralTimeSIS}
	\noindent \\
	
	Let $ w \in \mathbf{L_U} $ be a network input.\\
	
	Let $ 0 < N \in \mathbb{N} $.\\
	
	Let $ \myfunc{c}{ \mathbb{N} } { \mathfrak{C_{BB}} } { x } { c(x)=y } $  be a total computable non-decreasing function where 
	\[ 
	c(z + f( N, t_z  ) + 2) \geq c_0 + z + f( N, t_z  ) + 2 
	\]
	and
	\[
	c( z + f( N, t_z  ) + 2 ) - c_0 - 1 \leq t_{ |\mathrm{T}(G_t)|-1 } 
	\] \\
	
	If there is $ 0 \leq z_0 \leq | \mathrm{T}(G_t) | -1 $ and $ \epsilon, \, \epsilon_2 > 0 $ such that 
	
	\[
	z_0 + f( N, t_{z_0}  )  + 2
	= 
	\mathbf{ O }
	\left( \frac
	{ N^{ C } }
	{ \lg(N) } 
	\right)
	\]
	
	\noindent where
	
	\[ 
	0
	<
	C = 
	\]
	\[
	=
	\frac
	{
		{ \tau_{\mathbf{E}(\rho)}( N,f,t_{z_0}  ) }|_{ t_{z_0} }^{c( z_0 + f( N, t_{z_0}  ) + 2 ) - c_0 - 1  }
		-
		\Omega(w, c_0 + z_0 + f( N, t_{z_0} ) + 2 )
		-
		\epsilon
	}
	{ \Omega(w,  c_0 + z_0 + f( N, t_{z_0}  ) + 2  ) }
	\leq
	\]
	\[
	\leq
	\frac{1}{ \epsilon_2 }
	\] \\
	
	\noindent and $ \mathfrak{N'}_{BB}(N,f,t_{z_0} ,j) = ( G_t, \mathfrak{ P' }_{BB}(N), b_j ) $ is well-defined.
	
	Then, there are $ t_{cen_2}(c) $ and $ t_{cen_1}(c) $ such that 
	
	\[
	t_{cen_2}(c) = t_{cen_1}(c) \leq t_{ z_0 }
	\]

	\noindent \\
	
	\begin{proof}
		\noindent \\
		We know from Corollary \ref{corMainSIS} that
		
		\begin{equation}\label{stepFromcorMain2}
		\lim\limits_{ N \to \infty } 
		\mathbf{E}_{ \mathfrak{N'}_{BB}(N,f,t_z ) } 
		\left(
		{ {\displaystyle{\myDelta_{iso}^{net}} A} (o_i, c( z + f( N, t_z  ) + 2 ) )} 
		\right)
		\geq
		\end{equation}
		
		\begin{align*}
		& \geq
		\lim\limits_{ N \to \infty }
		\left( 
		{ \tau_{\mathbf{E}(\rho)}( N,f,t_{z}  ) }|_{ t_{z} }^{ c( z + f( N, t_z  ) + 2 ) - c_0 - 1  
		}
		-
		\Omega(w, c( z + f( N, t_z  ) + 2 ))
		\right)
		\lg(N) - \\
		& - \Omega(w, c( z + f( N, t_z  ) + 2 )) \lg( z + f( N, t_z  ) + 2) - \\
		& - 2 \, \Omega(w, c( z + f( N, t_z  ) + 2 ))\lg(\lg( z + f( N, t_z  ) + 2)) - A(w) - C_5 \\
		\end{align*}
		
		Suppose that there is $ t_{z_0} \in \mathrm{T}(G_t) $, where $ 0 \leq {z_0} \leq | \mathrm{T}(G_t) | -1 $, and $\epsilon > 0$ such that
		
		\begin{equation}
		{z_0} + f( N, t_{z_0}  ) + 2
		= 
		\mathbf{ O }
		\left( \frac
		{ N^{ C } }
		{ \lg(N) } 
		\right)
		\end{equation}
		
		\noindent where
		
		\[ 
		0 < C = 
		\frac
		{
			{ \tau_{\mathbf{E}(\rho)}( N,f,t_{z_0}  ) }|_{ t_{z_0} }^{ c( z_0 + f( N, t_{z_0}  ) + 2 ) - c_0 - 1  }
			-
			\Omega(w,  c_0 + z_0 + f( N, t_{z_0}  ) + 2 )
			-
			\epsilon
		}
		{ \Omega(w, c_0 + z_0 + f( N, t_{z_0}  ) + 2 ) }
		\] \\
		
		From the Definition \ref{BdefOmegawcSIS} we have that, for every $y \in \mathbb{N}$, if $ y \geq c_0 + z_0 + f( N, t_{z_0}  ) + 2 $, then
		
		\begin{equation}\label{stepNestingOmegas2}
		\Omega(w,y) \leq \Omega( w, c_0 + z_0 + f( N, t_{z_0} ) + 2 )
		\end{equation} \\
		
		Thus, since we are assuming $ c(z_0 + f( N, t_{z_0}  ) + 2) \geq c_0 + z_0 + f( N, t_{z_0} ) + 2 $, for fixed values of $ { \tau_{\mathbf{E}(\rho)}( N,f,t_{z_0}  ) }|_{ t_{z_0} }^{ c( z + f( N, t_z  ) + 2 ) - c_0 - 1  } $ and $\epsilon$ we will have from Step \eqref{stepNestingOmegas2} that

		\begin{equation}\label{stepCvsC'2}
		\frac
		{
			{ \tau_{\mathbf{E}(\rho)}( N,f,t_{z_0}  ) }|_{ t_{z_0} }^{ c( z_0 + f( N, t_{z_0}  ) + 2 ) - c_0 - 1  }
			-
			\Omega(w, c( z_0 + f( N, t_{z_0}  ) + 2 ))
			-
			\epsilon
		}
		{ \Omega(w, c( z_0 + f( N, t_{z_0}  ) + 2 )) }
		\geq
		\end{equation}
		\[
		\geq
		\frac
		{
			{ \tau_{\mathbf{E}(\rho)}( N,f,t_{z_0}  ) }|_{ t_{z_0} }^{ c( z_0 + f( N, t_{z_0}  ) + 2 ) - c_0 - 1  }
			-
			\Omega(w,  c_0 + z_0 + f( N, t_{z_0}  ) + 2 )
			-
			\epsilon
		}
		{ \Omega(w, c_0 + z_0 + f( N, t_{z_0}  ) + 2 ) }
		=
		C
		\geq
		0
		\] \\
		
		Let \[ C' = \frac
		{
			{ \tau_{\mathbf{E}(\rho)}( N,f,t_{z_0}  ) }|_{ t_{z_0} }^{ c( z_0 + f( N, t_{z_0}  ) + 2 ) - c_0 - 1  }
			-
			\Omega(w, c( z_0 + f( N, t_{z_0}  ) + 2 ))
			-
			\epsilon
		}
		{ \Omega(w, c( z_0 + f( N, t_{z_0}  ) + 2 )) } \]
		
		\noindent \\

		Remember that for every $x>0$ and $  t, t' \in \mathrm{T}(G_t) $ there is $\epsilon_2$ such that\footnote{ Remember that one can always have a program that halts for every input, so it will also halts for every partial output and, hence, halt on every cycle --- see Definition~\ref{BdefOmegawcSIS}. }
		
		\begin{equation}
		0 < \epsilon_2 \leq \Omega( w, x ) \leq 1
		\end{equation} 
		
		\noindent and thus, from the Definition \ref{BdefAveragePrevalenceN'_BB} , we will also have that
		
		\begin{equation}\label{stepBoundingC2}
		\frac{-1 - \epsilon}{ \epsilon_2 }
		\leq
		\frac
		{
			{ \tau_{\mathbf{E}(\rho)}( N,f,t  ) }|_{ t }^{ t' }
			-
			\Omega( w, x )
			- \epsilon
		}
		{ \Omega( w, x ) } \leq \frac{1}{ \epsilon_2 }
		\end{equation} \\
		
		Hence, from Steps \eqref{stepCvsC'2} and \eqref{stepBoundingC2} we will have that 
		\[
		z_0 + f( N, t_{z_0}  )  + 2
		= 
		\mathbf{ O }
		\left( \frac
		{ N^{ C' } }
		{ \lg(N) } 
		\right)
		\]
		\noindent where
		\[ 
		0
		\leq
		C' = 
		\]
		\[
		=
		\frac
		{
			{ \tau_{\mathbf{E}(\rho)}( N,f,t_{z_0}  ) }|_{ t_{z_0} }^{ c( z_0 + f( N, t_{z_0}  ) + 2 ) - c_0 - 1  }
			-
			\Omega(w, c( z_0 + f( N, t_{z_0}  ) + 2 ) )
			-
			\epsilon
		}
		{ \Omega(w, c( z_0 + f( N, t_{z_0}  ) + 2 ) ) }
		\leq
		\]
		\[
		\leq
		\frac{1}{ \epsilon_2 }
		\]  \\
		And, since $ {z_0} + f( N, t_{z_0}  ) + 2 $ is now assymptotically dominated by $ \frac
		{ N^{ C' } }
		{ \lg(N) } $, then by definition we will have that there is a constant $ C_6 $ such that
		
		\begin{equation}\label{stepMainthmMainCentralTime2}
		\lim\limits_{ N \to \infty }
		\left( 
		{ \tau_{\mathbf{E}(\rho)}( N,f,t_{z_0}  ) }|_{ t_{z_0} }^{ c( z_0 + f( N, t_{z_0}  ) + 2 ) - c_0 - 1  }
		-
		\Omega(w, c( z_0 + f( N, t_{z_0}  ) + 2 ))
		\right)
		\lg(N) - \\
		\end{equation}
		
		\begin{align*}
		& - \Omega(w, c( z_0 + f( N, t_{z_0}  ) + 2 )) \lg( {z_0} + f( N, t_{z_0}  ) + 2) - \\
		& - 2 \, \Omega(w, c( z_0 + f( N, t_{z_0}  ) + 2 ))\lg(\lg( {z_0} + f( N, t_{z_0}  ) + 2)) - A(w) - C_5 
		\geq \\
		& \geq
		\lim\limits_{ N \to \infty }
		\left( 
		{ \tau_{\mathbf{E}(\rho)}( N,f,t_{z_0}  ) }|_{ t_{z_0} }^{ c( z_0 + f( N, t_{z_0}  ) + 2 ) - c_0 - 1  }
		-
		\Omega(w, c( z_0 + f( N, t_{z_0}  ) + 2 ))
		\right)
		\lg(N) - \\
		& - \Omega(w, c( z_0 + f( N, t_{z_0}  ) + 2 )) \lg( C_6 \, \frac
		{ N^{ C' } }
		{ \lg(N) } ) - \\
		& - 2 \, \Omega(w, c( z_0 + f( N, t_{z_0}  ) + 2 ))\lg( \lg( C_6 \, \frac
		{ N^{ C' } }
		{ \lg(N) } ) ) - A(w) - C_5 
		\geq \\
		& \geq
		\lim\limits_{ N \to \infty }
		\left( 
		{ \tau_{\mathbf{E}(\rho)}( N,f,t_{z_0}  ) }|_{ t_{z_0} }^{ c( z_0 + f( N, t_{z_0}  ) + 2 ) - c_0 - 1  }
		-
		\Omega(w, c( z_0 + f( N, t_{z_0}  ) + 2 ))
		\right)
		\lg(N) - \\
		& - \Omega(w, c( z_0 + f( N, t_{z_0}  ) + 2 )) \left( 
		\lg( C_6 ) + C' \, \lg( N ) - \lg(\lg(N))  )
		\right) - \\
		& - 2 \, \Omega(w, c( z_0 + f( N, t_{z_0}  ) + 2 ))\lg\left( 
		\lg( C_6 ) + \lg( N^{ C' } ) - \lg(\lg(N))  ) 
		\right) - A(w) - C_5 
		\geq \\
		& \geq
		\lim\limits_{ N \to \infty }
		\left( 
		\epsilon
		\right)
		\lg(N) - \Omega(w, c( z_0 + f( N, t_{z_0}  ) + 2 )) \left( 
		\lg( C_6 ) - \lg(\lg(N))  )
		\right) - \\
		& - 2 \,  \Omega(w, c( z_0 + f( N, t_{z_0}  ) + 2 ))\lg\left( 
		\lg( C_6 ) + \lg( N^{ C' } ) - \lg(\lg(N))  ) 
		\right) - A(w) - C_5 
		\geq \\
		& \geq
		\lim\limits_{ N \to \infty }
		\left( 
		\epsilon
		\right)
		\lg(N) - \Big( 
		\lg( C_6 ) - \lg(\lg(N))  
		\Big) - 2 \, \lg\Big( 
		\lg( C_6 ) + \lg( N^{ C' } ) - \lg(\lg(N))  
		\Big) - \\
		& - A(w) - C_5 
		\geq \\
		& \geq
		\lim\limits_{ N \to \infty }
		\left( 
		\epsilon
		\right)
		\lg(N) - 
		\lg( C_6 ) + \lg(\lg(N)) - 2 \, \lg\left( \lg( N^{ \frac{1}{ \epsilon_2 } } ) \right) - A(w) - C_5 
		\geq \\
		& \geq
		\lim\limits_{ N \to \infty }
		\left( 
		\epsilon
		\right)
		\lg(N) - 
		\lg( C_6 ) + \lg(\lg(N)) 
		- 2 \, \lg( \frac{1}{ \epsilon_2 } \, \lg( N ) ) 
		- A(w) - C_5 
		\geq \\
		& \geq
		\lim\limits_{ N \to \infty }
		\left( 
		\epsilon
		\right)
		\lg(N) - 
		\lg( C_6 ) + \lg(\lg(N)) 
		- 2 \, \lg( \frac{1}{ \epsilon_2 } ) - 2 \, \lg( \lg( N ) )
		- A(w) - C_5 
		\geq \\
		& \geq
		\lim\limits_{ N \to \infty }
		\left( 
		\epsilon
		\right)
		\lg(N) - 
		\lg( C_6 ) 
		- 2 \, \lg( \frac{1}{ \epsilon_2 } )
		- \lg( \lg( N ) )
		- A(w) - C_5 
		=
		\infty \\
		\end{align*}

		Thus, from Steps \eqref{stepFromcorMain2} and \eqref{stepMainthmMainCentralTime2}, we will have that

		\begin{equation}\label{stepPrevthmMainCentralTime2}
		\lim\limits_{ N \to \infty } 
		\mathbf{E}_{ \mathfrak{N'}_{BB}(N,f,t_{z_0} ) } 
		\left(
		{ {\displaystyle{\myDelta_{iso}^{net}} A} (o_i, c( z_0 + f( N, t_{z_0}  ) + 2 ) ) } 
		\right)
		=
		\infty
		\end{equation} \\
		
		Then, directly from the Definitions \ref{BdefTimecentrality1SIS} and \ref{BdefTimecentrality2SIS} and Step \eqref{stepPrevthmMainCentralTime2}, since $ t_{ z_0 } $ satisfies these definitions, we will have that
		
		\[
		t_{cen_2}(c) = t_{cen_1}(c) \leq t_{ z_0 }
		\]

	\end{proof}
	
	\begin{noteunderthm}
		The reader is also invited to note that the same result also hold for condition 
		\[
		C=
		\frac
		{
			{ \tau_{\mathbf{E}(\rho)}( N,f,t_{z_0}  ) }|_{ t_{z_0} }^{c( z + f( N, t_z  ) + 2 ) - c_0 - 1  }
			-
			\Omega(w, c( c_0 + z_0 + f( N, t_{z_0} ) + 2 ) )
			-
			\epsilon
		}
		{ \Omega(w,  c( c_0 + z_0 + f( N, t_{z_0} ) + 2 ) }
		\]
		\noindent instead of 
		\[
		\frac
		{
			{ \tau_{\mathbf{E}(\rho)}( N,f,t_{z_0}  ) }|_{ t_{z_0} }^{c( z + f( N, t_z  ) + 2 ) - c_0 - 1  }
			-
			\Omega(w, c_0 + z_0 + f( N, t_{z_0} ) + 2 )
			-
			\epsilon
		}
		{ \Omega(w,  c_0 + z_0 + f( N, t_{z_0}  ) + 2  ) }
		\]
		\noindent In order to prove it, just make $ C'=C $ in the proof of Theorem~\ref{thmMainCentralTimeSIS}.
	\end{noteunderthm}
	
\end{thm}

\noindent \\

\subsection{Main Corollary extended}

\begin{corollaryundersubsection}[\ref{BcorSIS}](SIS)\label{corSIS}
	\noindent \\
	
	Let $ w \in \mathbf{L_U} $ be a network input.\\
	
	Let $ 0 < N \in \mathbb{N} $. \\
	
	Let $ {\mathfrak{N'}_{BB}}(N,f,t_{z_0} ,j) = ( G_t, \mathfrak{ P' }_{BB}(N), b_j ) $ be well-defined.\\ 
	
	Let $ \myfunc{c}{ \mathbb{N} } { \mathfrak{C_{BB}} } { x } { c(x)=y } $  be a total computable non-decreasing function where 
	
	\[ 
	c(z_0 + f( N, t_{z_0}  ) + 2) \geq c_0 + z_0 + f( N, t_{z_0}  ) + 2 
	\]
	and
	\[
	c( z_0 + f( N, t_{z_0}  ) + 2 ) - c_0 - 1 \leq t_{ |\mathrm{T}(G_t)|-1 } 
	\] \\
	
	If 
	\[ f(N,t_{z_0} ) =\mathbf{O}\big( \lg( N ) \big) \]
	where every $ G_t \in \mathbb{G}_{SIS}( f, t_{z_0} ) $ achieves stationary  prevalence $ \rho $ in a number of time intervals
	\[
	 \Delta^*_{t_{z_0}}  \leq c( z_0 + f( N, t_{z_0}  ) + 2 ) - c_0 - 1
	 \]
	 \noindent after time instant $ t_{z_0} $ and
	\[ 
	\rho \sim \exp( - \frac{1}{m \lambda}) 
	>
	\Omega(w, c_0 + z_0 + f( N, t_{z_0}  ) + 2 ) 
	\]
	\noindent then, there are $ t_{cen_2}(c) $ and $ t_{cen_1}(c) $ such that
	\[
	t_{cen_2}(c) = t_{cen_1}(c) \leq t_{ z_0 }
	\] 
\end{corollaryundersubsection}
\begin{proof}
	The proof comes directly from Theorem~\ref{thmMainCentralTimeSIS} and Definition~\ref{BdefN'_BB}: \\
	
	We have by supposition that
	\begin{align}\label{stepMaincentraltimeSIS}
	f(i,t ) =\mathbf{O}\big( \lg( i ) \big) 
	\end{align} 
	
	\noindent Thus, 
	\begin{align}\label{stepFunctionfandSIS}
	z_0 + f( N, t_{z_0}  )  + 2
	=
	z_0 + \mathbf{ O }\left( \lg(N) 
	\right)
	+ 2
	=
	\mathbf{ O }
	\left( \lg(N) 
	\right)
	\end{align}
	\noindent where $ 0 \leq z_0 \leq | \mathrm{T}(G_t) | -1 $. \\
	
	By supposition, we have that the time interval to achieve stationary prevalence is upper bounded by  $ c( z_0 + f( N, t_{z_0}  ) + 2 ) - c_0 - 1 $ and that there is $ \epsilon' > 0 $ such that
	\[
	\frac{1}{ e^{ \left(\frac{1}{m \lambda} \right) } } = \Omega(w, c_0 + z_0 + f( N, t_{z_0}  ) + 2 ) + \epsilon'  
	\]
	
	\noindent \\ 
	
	Hence, we will have, from Definitions~\ref{BdefPrevalenceN'_BB} and \ref{BdefN'_BB} and from the definition of stationary prevalence $\rho$ in~\cite{Pastor-Satorras2001,Pastor-Satorras2001a,Pastor-Satorras2002}, that there is $ \epsilon > 0 $ such that
	\begin{align}\label{stepCfromSIS}
	& \frac{-1 - \epsilon}{ \epsilon_2 } 
	< 0 
	<
	C = 
	\frac
	{
		\frac{1}{ e^{ \left(\frac{1}{m \lambda} \right) } }
		-
		\Omega(w, c_0 + z_0 + f( N, t_{z_0}  ) + 2 )
		-
		\epsilon
	}
	{ \Omega(w,  c_0 + z_0 + f( N, t_{z_0}  ) + 2  ) }
	=
	\end{align}
	\[
	=
	\frac
	{
		{ \tau_{\mathbf{E}(\rho)}( N,f,t_{z_0}  ) }|_{ t_{z_0} }^{ c( z_0 + f( N, t_{z_0}  ) + 2 ) - c_0 - 1  }
		-
		\Omega(w, c_0 + z_0 + f( N, t_{z_0}  ) + 2 )
		-
		\epsilon
	}
	{ \Omega(w,  c_0 + z_0 + f( N, t_{z_0}  ) + 2  ) }
	\leq
	\]
	\[
	\leq
	\frac
	{
		1
		-
		\Omega(w, c_0 + z_0 + f( N, t_{z_0}  ) + 2 )
		-
		\epsilon
	}
	{ \Omega(w,  c_0 + z_0 + f( N, t_{z_0}  ) + 2  ) }
	\leq
	\frac{1}{ \epsilon_2 }
	\]
	\noindent where $ \epsilon' > \epsilon $. \\
	From Step \eqref{stepFunctionfandSIS} we have that if $\frac{1}{\epsilon_2}\geq C > 0$, then
	\begin{align}\label{stepDiameterandSWSIS}
	z_0 + f( N, t_{z_0}  )  + 2
	=
	z_0 + \mathbf{ O }\left( \lg(N) 
	\right)
	+ 2
	=
	\mathbf{ O }
	\left( \lg(N) 
	\right)
	= 
	\mathbf{ O }
	\left(
	\frac
	{ N^{ C } }
	{ \lg(N) } 
	\right)
	\end{align}
	Then, from Steps \eqref{stepCfromSIS} and \eqref{stepDiameterandSWSIS} and Theorem~\ref{thmMainCentralTimeSIS}  we will have that there are $ t_{cen_2}(c) $ and $ t_{cen_1}(c) $ such that
	\[
	t_{cen_2}(c) = t_{cen_1}(c) \leq t_{ z_0 }
	\] 
\end{proof}

\end{document}